\documentclass[11pt]{article}

\usepackage{xargs}

\usepackage{natbib}

\usepackage{amsmath,amssymb,float,amsthm}

\usepackage{url,verbatim,array,multirow,subfig}

\usepackage{graphics}
\usepackage{graphicx, algorithm}
\usepackage{enumerate}

\usepackage{xargs}
\usepackage{amsfonts}
\usepackage[margin=1in]{geometry}

\usepackage{color}

\newtheorem{theorem}{Theorem} 

\newtheorem{prop}{Proposition}
\newtheorem{definition}{Definition}

\newtheorem{observation}{Observation}

\newcommand{\B}{\boldsymbol}
\newcommand{\M}{\mathbf}
\newcommand{\sbt}{\mathrm{s.t.}}
\newcommand{\rnk}{\mathrm{rank}}
\newcommand{\tr}{\mathrm {Tr}}

\DeclareMathOperator*{\argmin}{arg\,min}

\DeclareMathOperator*{\mini}{minimize}
\DeclareMathOperator*{\maxi}{maximize}

\newcommand{\R}{\mathbb{R}}

\newcommand{\sub}{\subseteq}

\newcommand{\st}{\operatorname{s.t.}}
\newcommand{\rank}{\operatorname{rank}}

\newcommand{\mb}[1]{\mathbf{#1}}
\newcommand{\Tr}{\operatorname{Tr}}

\newcommand{\bs}{\boldsymbol}

\newcommand{\F}{\mathcal{F}_{\s}}
\newcommand{\Fs}{\mathcal{F}_{\s}}

\newcommand{\cg}{\succcurlyeq}
\newcommand{\psd}{\cg\mb0}

\newcommand{\uu}{\mb u}

\newcommand{\s}{\bs\Sigma}

\newcommand{\zz}{\mb z}

\newcommand{\ds}{\displaystyle}

\newcommand{\bk}{\bs\mu}
\newcommand{\bsig}{\bs\sigma}
\newcommand{\f}{\bs f}
\newcommand{\nn}{\mb N}
\newcommand{\mm}{\mb M}
\newcommand{\pp}{\mb P}

\newcommand{\cW}{{\bs \Psi_{p,p-r}}}

\newcommand{\convenv}{\operatorname{convenv}}
\newcommand{\tol}{\texttt{TOL}}
\newcommand{\nodes}{\texttt{Nodes}}

\newcommand{\T}{\bs\Theta}
\newcommand{\Ph}{\bs\Phi}
\newcommand{\ph}{\bs \phi}
\newcommand{\diag}{\operatorname{diag}}
\newcommand{\W}{\mb W}

\newcommand{\LL}{\mb L}
\newcommand{\bell}{\bs \ell}

\newcommand{\n}{\texttt{n}}

\title{Certifiably Optimal Low Rank Factor Analysis}

\author{ Dimitris Bertsimas \and Martin S. Copenhaver\and Rahul Mazumder\thanks{Sloan School of Management and Operations Research Center, MIT. Emails: \texttt{$\{$dbertsim,mcopen,rahulmaz$\}$@mit.edu}. Copenhaver is supported by the Department of Defense, Office of Naval Research, through the National Defense Science and Engineering Graduate (NDSEG) Fellowship.}}

\date{}

\begin{document}

\maketitle

\begin{abstract}%
Factor Analysis  (FA) is a technique of fundamental importance that is widely used in classical and modern multivariate statistics, psychometrics and econometrics. In this paper,
 we revisit  the classical rank-constrained FA problem, which seeks to approximate  an observed covariance matrix ($\B\Sigma$), by the sum of a
 Positive Semidefinite (PSD) low-rank component ($\B\Theta$) and a diagonal matrix ($\B\Phi$) (with nonnegative entries) subject to 
 $\B\Sigma - \B\Phi$ being PSD. 
  We propose a flexible family of rank-constrained, nonlinear Semidefinite Optimization based formulations for this task.
 We introduce a reformulation of the problem as a smooth optimization problem with convex compact constraints and 
 propose a unified algorithmic framework, utilizing state of the art techniques in 
nonlinear optimization 
to obtain high-quality feasible solutions for our proposed formulation. 
At the same time, by using a variety of techniques from discrete and global optimization, we show that these solutions are \emph{certifiably optimal} in many cases, even for problems with thousands of variables. Our techniques are general and make \emph{no} assumption on the underlying problem data. The estimator proposed herein, aids statistical interpretability, provides computational scalability and significantly improved accuracy when compared to current, publicly available popular 
methods for rank-constrained FA. We demonstrate the effectiveness of our proposal on an  array of synthetic and real-life datasets. To our knowledge, this is the first paper that demonstrates how a previously intractable rank-constrained optimization problem can be solved to provable optimality by coupling developments in convex analysis and in discrete optimization.
\end{abstract}

\section{Introduction}\label{sec:intro1}
Factor Analysis (FA)~\citep{anderson2003, Bartholomew_Knott_Moustaki_2011,mardia}, a widely used methodology in classical and modern multivariate statistics is used as a tool to obtain a parsimonious representation of
the correlation structure among a set of variables in terms of a smaller number of common hidden factors. 
A basic FA model is of the form $\M{x} = \M{L} \M{f} + \B{\epsilon},$ where $\M{x}_{p \times 1}$ is the observed random vector, $\M{f}_{r_1 \times 1}$ (with $r_{1} \leq p$, note that we do not necessarily restrict $r_{1}$ to be small) 
is a random vector of common factor variables or scores, 
$\M{L}_{ p \times r_1}$ is a matrix of factor loadings and $\B{\epsilon}_{p \times 1}$ is a vector of uncorrelated random  variables. We assume that the variables are mean-centered,
$\M{f}$ and $\B{\epsilon}$ are uncorrelated and without loss of generality, 
the covariance of $\M{f}$ is the identity matrix. We will denote
$\mathrm{Cov}(\B{\epsilon})=\B{\Phi}= \mathrm{diag}(\Phi_1, \ldots, \Phi_p).$  
It follows that
 \begin{equation}\label{decompose-1}
\B\Sigma = \B\Sigma_{c} + \B{\Phi},
\end{equation}
 where,  
 $\B\Sigma$ is the covariance matrix of $\M{x}$ and $\B\Sigma_{c}:=  \M{L}\M{L}'$ is the covariance matrix corresponding to the common factors.
  Decomposition~\eqref{decompose-1}
suggests that $\B\Sigma$ can be written as the sum of a positive semidefinite (PSD) matrix $\B\Sigma_{c}$ of rank $r_{1}$
 and a nonnegative diagonal matrix ($\B\Phi$), corresponding to the errors. 
In particular, 
the variance of the $i$th coordinate of $\M{x}:=(x_{1}, \ldots, x_{p}),$ i.e.,  
$\text{var}(x_{i})= \sum_{k} l^2_{ik} + \Phi_{i}, i = 1, \ldots, p,$  splits into two parts,
where, $\M{L} = (( l_{ik} ))$.
The first part ($\sum_{k} l^2_{ik}$) is known as the 
\emph{communality estimate} (since this is the variance of the factors common to all the $x_{i}$'s) 
and the remaining part $\Phi_{i}$ is the variance specific to the $i$th variable ($\Phi_{i}$'s are also referred to as
the \emph{unique variances} or simply \emph{uniquenesses}).
In the finite sample setting, when we are provided with $n$ samples $\M{x}_{i}, i = 1, \ldots, n$, 
we consider the empirical covariance matrix  in place of $\B\Sigma$.

\paragraph{Formulation of the estimator:} 
In decomposition~\eqref{decompose-1}, the assumption that the rank ($r_{1}$) of $\B\Sigma_{c}$ is small compared to $p$
is fairly stringent, see~\cite{guttman1958extent,rank-shapiro-82,ten1998some} for a historical overview of the concept. In a classical paper~\cite{guttman1958extent}, the author argued based on psychometric evidence that $\B\Sigma_{c}$ is often found to have high algebraic rank. 
In psychometric case studies it is rather rare that the covariance structure can be \emph{completely} explained by a \emph{few} common factors corresponding to mental abilities---in fact there is evidence of at least hundreds  of common factors being present with the number growing without an upper bound. Formally, this means that instead of assuming  
that $\B\Sigma_{c}$ has \emph{exactly} low-rank it is practical to assume that it can be well-approximated by a low-rank matrix, namely,
 $\M{L}_{1}\M{L}_{1}'$ with $\M{L}_{1} \in \R^{p \times r}$.  More precisely,  $\M{L}_{1}\M{L}_{1}'$ is the best rank-$r$ approximation to $\B\Sigma_{c}$ in the matrix 
$q$-norm (also known as the Schatten norm), as defined in~\eqref{mat-q-norm} and 
 $\left(\B\Sigma_{c} -\M{L}_1\M{L}_1'  \right)$ is the \emph{residual} component. Following psychometric terminology,  $\M{L}_{1}$ corresponds to the 
$r$ \emph{most significant}  factors representative of mental abilities and the residual $\B\Sigma_{c} - \M{L}_{1}\M{L}_{1}'$ 
corresponds to the remaining psychometric factors unexplained by $\M{L}_1\M{L}_{1}'$.
 Thus we can rewrite decomposition~\eqref{decompose-1} as follows:
\begin{equation}\label{cov-decomp-1}
\B{\Sigma} = \underbrace{\M{L}_{1}\M{L}_1'}_{:=\B\Theta} +  \underbrace{ \left( \B\Sigma_{c} - \M{L}_1\M{L}_1'  \right)}_{:=\B{\cal N}} + \B{\Phi},
\end{equation}
where, we use the notation $\B\Theta = \M{L}_1 \M{L}_1'$ and $\B{\cal N} = \left( \B\Sigma_{c}  -\M{L}_1\M{L}_1'  \right)$ with 
$\B\Theta + \B{\cal N}  = \B\Sigma_{c}=   \B{\Sigma} - \B\Phi$.
Note that $\B\Theta$ denotes  the best rank-$r$ approximation to $(\B{\Sigma}  - \B\Phi)$, with 
the residual component being $\B{{\cal N}} = \B{\Sigma}  - \B\Phi - \B\Theta$. Note that the entries in $\B\Phi$ need to be non-negative\footnote{Negative estimates of  the diagonals of $\B\Phi$ are unwanted since they correspond to variances, but 
some FA estimation procedures often lead to negative estimates of $\B\Phi$---these are popularly known in the literature as Heywood cases and have invited a significant amount of discussion in the community.} and $\B\Sigma - \B\Phi \succeq \M{0}$.
In fact, in the words of~\cite{ten1998some} (see p.\ 326) \\
``\ldots \emph{However, when $\B\Sigma-\B\Phi$ the
covariance matrix for the common parts of the variables, would appear to be
indefinite, that would be no less embarrassing than having a negative unique
variance in $\B\Phi$}\ldots''

We further refer the reader to~\cite{mardia} discussing the importance of $\B\Sigma - \B\Phi$ being PSD.\footnote{However, the estimation method described in~\cite{mardia} does not guarantee that 
$\B\Sigma - \B\Phi \succeq \M{0}$.}
We thus have the following natural structural constraints on the parameters:
\begin{equation}\label{feas-set-1}
\B\Theta \succeq \M{0}, \;\;\;\;\; \B\Phi=\diag(\Phi, \ldots, \Phi_{p}) \succeq \M{0}  \;\;\;\;\;\text{and}\;\;\;\;\; \B{\Sigma} - \B\Phi \succeq \M{0}.
\end{equation}
 
Motivated by the above considerations, we present the following rank-constrained estimation problem for FA: 
\begin{equation}\label{obj-2-0}
\begin{aligned}
\mini\;\; & \eta_{\B\Sigma}(\B\Theta, \B\Phi):= \;\;\; \| \B{\Sigma} - (\B{\Theta} + \B{\Phi})  \|^q_q \\
\sbt \;\;\; & \;\;\; \rnk(\B{\Theta})\leq r \\
& \;\;\; \B{\Theta}  \succeq \M{0} \\
& \;\;\; \B{\Phi}=\mathrm{diag}(\Phi_1, \ldots, \Phi_p) \succeq \M{0} \\
&\;\;\; \B\Sigma - \B\Phi \succeq \M{0},
\end{aligned}
\end{equation}
where  $\B{\Theta} \in \R^{p \times p}, \B{\Phi} \in \R^{p \times p}$ are the optimization variables, and
for a real symmetric matrix $\M{A}_{ p \times p}$,  
its matrix $q$-norm, also known as the Schatten norm (or Schatten-von-Neumann norm) is defined as:
\begin{equation}\label{mat-q-norm}
\| \M{A} \|_q :=  \left(\sum_{i=1}^{p} |\lambda_{i}(\M{A})|^{q} \right)^{\frac{1}{q}},
\end{equation}
where $\lambda_{i}(\M{A}), i = 1, \ldots, p$, are the (real) eigenvalues of $\M{A}$. 

\paragraph{Interpreting the estimator:} 

The estimation criterion~\eqref{obj-2-0} seeks to \emph{jointly} obtain the (low-rank) common factors and \emph{uniquenesses} that 
best explain $\B\Sigma$ in terms of minimizing the matrix $q$-norm of the error $\B\Sigma - (\B\Theta + \B\Phi)$
under the PSD constraints~\eqref{feas-set-1}. Note that criterion~\eqref{obj-2-0} does not necessarily assume that $\B\Sigma$ \emph{exactly} decomposes into a low-rank 
PSD matrix and a non-negative diagonal matrix. Problem~\eqref{obj-2-0} enjoys curious similarities with Principal Component Analysis (PCA). 
In PCA, given a PSD matrix $\B\Sigma$ the leading $r$ principal component directions of $\B\Sigma$ are obtained 
by minimizing $\| \B\Sigma - \B\Theta \|_{q}$  subject to $\B\Theta \succeq \M{0}$ and $\rnk(\B\Theta) \leq r$.
If the optimal solution  $\B\Phi$ to Problem~\eqref{obj-2-0} is \emph{given}, Problem~\eqref{obj-2-0} is analogous to a rank-$r$ PCA
on the residual matrix $\B\Sigma - \B\Phi$ --- thus it is naturally desirable to have $\B\Sigma - \B\Phi \succeq \M{0}$. In PCA one is interested in understanding the proportion of variance explained 
by the top-$r$ principal component directions: $\sum_{i=1}^{r} \lambda_{i}(\B\Sigma)/\sum_{i=1}^{p} \lambda_{i}(\B\Sigma)$. The denominator 
$\sum_{i=1}^{p} \lambda_{i}(\B\Sigma) = \tr(\B\Sigma)$ accounts for the total variance explained by the covariance matrix $\B\Sigma$. Analogously, the proportion of variance explained by 
$\widehat{\B\Theta}_{r}$ (which denotes the best rank $r$ approximation to $\B\Sigma - \B\Phi$) is given by $\sum_{i=1}^{r} \lambda_{i}(\B\Sigma - \B\Phi)/\sum_{i=1}^{p} \lambda_{i}(\B\Sigma - \B\Phi)$ --- for this quantity to be 
interpretable it is imperative that $\B\Sigma - \B\Phi \succeq \M{0}$. 
In the above argument, of course,  we assumed that $\B\Phi$ is given. In general, $\B\Phi$ needs to be estimated: Problem~\eqref{obj-2-0} achieves this goal by \emph{jointly} 
learning $\B\Phi$  and $\B\Theta$. 
We note that certain popular approaches of FA (see Sections~\ref{sec:related-work} and~\ref{sec:categories}) do not impose the PSD 
constraint $\B\Sigma - \B\Phi$ as a part of the estimation scheme --- leading to indefinite 
$\B\Sigma - \B\Phi$, thereby rendering statistical interpretations troublesome. Our numerical evidence suggests that the quality of estimates of 
$\B\Theta$ and $\B\Phi$ obtained from Problem~\eqref{obj-2-0} outperform those obtained by other competing procedures which do not take into account the PSD constraints 
into their estimation criterion.

\paragraph{Choice of $\M{r}$:} In exploratory FA, it is standard to consider several choices of $r$ and study the manner in which the proportion of variance explained by the 
common factors \emph{saturates} with increasing $r$. We refer the reader to popularly used methods, described in~\cite{anderson2003, Bartholomew_Knott_Moustaki_2011,mardia}
 and more modern techniques~\cite{bai-ng-2008large-review} (see also, references therein) for the choice of $r$.

\medskip
\medskip

In this paper, we propose a general computational framework to solve Problem \eqref{obj-2-0} 
for  any $q \geq 1$. The very well known Schatten $q$-norm appearing in the loss function is chosen for flexibility---it underlines the fact that our approach can be applied for \emph{any} $q \geq 1$. 
Note that the estimation criterion~\eqref{obj-2-0} (even for the case $q=1$) does not seem to appear in prior work on  \emph{approximate minimum rank Factor Analysis} ({\textsc{\small MRFA}})~\citep{berge-91,shapiro-inference-FA-02}. However, we 
show in Proposition~\ref{lem:reform-1} that Problem~\eqref{obj-2-0} for the special case $q=1$, turns out to be equivalent to {\textsc{\small MRFA}}. 
For $q=2$, the loss function is the familiar squared Frobenius norm also used in {\textsc{\small MINRES}}, though the latter formulation is not equivalent to Problem~\eqref{obj-2-0}, as explained in Section~\ref{sec:related-work}.  We place more emphasis on studying the computational properties for the more common norms $q \in \{ 1, 2\}$.

The presence of the rank constraint in Problem~\eqref{obj-2-0} makes the optimization problem non-convex. Globally optimizing 
Problem~\eqref{obj-2-0} or for that matter obtaining a good stationary point is quite challenging. We propose a new \emph{equivalent}
smooth formulation to Problem~\eqref{obj-2-0} which does not contain the combinatorial rank constraint. We employ simple and 
tractable sequential convex relaxation techniques with guaranteed convergence properties and excellent computational properties
to obtain a stationary point for Problem~\eqref{obj-2-0}. An important novelty of this paper is to present certifiable lower bounds on Problem~\eqref{obj-2-0} without resorting to structural assumptions, thus making it possible to
solve Problem~\eqref{obj-2-0} to \emph{provable} optimality. Towards this end we propose new methods and ideas that 
significantly outperform existing off-the-shelf methods for nonlinear global optimization.\footnote{The class of optimization problems studied in this paper involve global minimization of nonconvex continuous SDOs. Computational methods for this class of problems are in a nascent stage; further, such methods are significantly less developed when compared to those for mixed integer linear optimization problems, thus posing a major 
challenge in this work}.

 \subsection{A selective overview of related FA estimators}\label{sec:related-work}
FA has a long and influential history which dates back to more than a hundred years. 
The notion of FA possibly first appeared in~\cite{spearman1904} for the one factor model, which was then 
generalized to the multiple factors model by various authors~(see for example, \cite{thurstone1947multiple}). Significant contributions related to computational methods for FA have been nicely documented in~\cite{bai-ng-2008large-review,minres-66,mle-factor-analysis,joreskog-78,lawley-max-62,lederman-37,Rao73,rank-shapiro-82,berge-91}, among others.

We will briefly describe some widely used approaches for FA that are closely related to the approach pursued herein and also point out their connections.
\paragraph{Constrained Minimum Trace Factor Analysis (\textsc{\small{MTFA}}):}
This approach~\citep{tenberge-mtfa-81,shapiro-mtfa-82} seeks to decompose $\B\Sigma$
exactly into the sum of a diagonal matrix and a low-rank component, which are estimated via the following convex optimization problem:
\begin{equation}\label{orig-mtfa1}
\begin{array}{rl}
\mini\limits_{}\;\;\;&   \tr(\B\Theta) \\
 \sbt\;\;\; & \B\Theta \succeq \M{0} \\
 & \B\Sigma = \B\Theta + \B\Phi \\
 & \B\Phi = \diag(\Phi_1, \ldots, \Phi_p) \succeq  \M{0},
\end{array}
\end{equation}
with variables $\B\Theta, \B\Phi$. $\B\Theta$ being PSD,   
$\tr(\B\Theta) = \sum_{i=1}^{p} \lambda_{i}(\B\Theta)$ is a convex surrogate~\citep{fazel-thes} for the rank of $\B\Theta$---Problem~\eqref{orig-mtfa1} 
may thus be viewed as a 
convexification of the rank minimization problem:
\begin{equation}\label{orig-mtfa1-exact}
\begin{array}{rl}
\mini\limits_{} \;\;\;&    \rnk(\B\Theta) \\
 \sbt\;\;\; & \B\Theta \succeq \M{0} \\
 & \B\Sigma = \B\Theta + \B\Phi \\
 &\B\Phi = \diag(\Phi_1, \ldots, \Phi_p) \succeq  \M{0}. 
\end{array}
\end{equation}
In general, Problems \eqref{orig-mtfa1} and  \eqref{orig-mtfa1-exact}   are \emph{not} equivalent. 
See~\cite{shapiro-mtfa-82,venkat-2012-factor,rank-shapiro-82} (and references therein) for further connections between 
the minimizers  of \eqref{orig-mtfa1} and   \eqref{orig-mtfa1-exact}.

A main difference between formulations~\eqref{obj-2-0} and~\eqref{orig-mtfa1-exact} is that the former allows an error in 
the residual $(\B\Sigma - \B\Theta - \B\Phi)$ by constraining  $\B\Theta$ to have low-rank, 
unlike~\eqref{orig-mtfa1-exact} which imposes a hard constraint $\B\Sigma = \B\Theta + \B\Phi$. This can be quite restrictive
in various applications as substantiated by our experimental findings in Section~\ref{sec:compute}.

\paragraph{Approximate Minimum Rank Factor Analysis (\textsc{\small MRFA}):}
This method~\citep[see for example,][]{berge-91,shapiro-inference-FA-02}
considers the following optimization problem:
\begin{equation}\label{obj-2-0-margin-mrfa}
\begin{array}{rl}
\mini\limits\;\; & 
\displaystyle  \sum_{i=r+1}^{p} \lambda_i(\B{\Sigma} - \B{\Phi} ) \\
\sbt \;\; & \B{\Phi}=\mathrm{diag}(\Phi_1, \ldots, \Phi_p) \succeq \M{0} \\
 &  \B\Sigma - \B\Phi \succeq \M{0},
\end{array}
\end{equation}
where the  optimization variable is $\B\Phi$.
Proposition~\ref{lem:reform-1}, presented below establishes that Problem~\eqref{obj-2-0-margin-mrfa} is  equivalent 
to the rank-constrained FA formulation~\eqref{obj-2-0} for the case $q=1$. This connection does not appear to be formally established 
in~\cite{berge-91}. We believe that criterion~\eqref{obj-2-0} for $q=1$ is easier to interpret as an estimation criterion for FA models over~\eqref{obj-2-0-margin-mrfa}.
\cite{berge-91} also describe a method for numerically optimizing~\eqref{obj-2-0-margin-mrfa}---as
documented in the code description for \textsc{\small MRFA}~\citep{mrfa-2003-code}, their implementation can handle problems of size $p\leq20$.


\paragraph{Principal Component (PC) Factor Analysis:}
Principal Component factor analysis or PC in short~\citep{connor1986performance,bai-ng-2008large-review}
implicitly assumes that $\B\Phi = \sigma^2 \M{I}_{p \times p}$, for some 
$\sigma^2 > 0$ and performs a low-rank PCA on $\B\Sigma$.
It is not clear how to estimate $\B\Phi$ via this method such that $\Phi_{i} \geq 0$ and $\B\Sigma - \B\Phi \succeq \M{0}$. 
Following~\cite{mardia}, the $\B\Phi$'s may be estimated after estimating $\widehat{\B\Theta}$ via
the update rule $\widehat{\B\Phi} = \diag(\B\Sigma - \widehat{\B\Theta})$---the estimates thus obtained, however, need not be non-negative.
Furthermore, it is not guaranteed that the condition $\B\Sigma - \B\Phi \succeq \M{0}$ is met.

\paragraph{Minimum Residual Factor Analysis (\textsc{{\textsc{\small MINRES}}}):}
This approach~\citep{minres-66,shapiro-inference-FA-02} considers the following optimization 
problem (with respect to the variable ${\M{L}_{p \times r}}$):
\begin{equation}\label{orig-minres1}
\mini \;\;\;\;  \sum_{1 \leq i\neq j \leq p}  \left ( \sigma_{ij} - (\M{L}\M{L}')_{ij} \right)^2,
\end{equation}
where $\B\Sigma := ((\sigma_{ij}))$, $\B\Theta = \M{L}\M{L}'$ and
the sum (in the objective) is taken over all the off-diagonal entries.  Formulation~\eqref{orig-minres1} is equivalent to the non-convex optimization problem:
\begin{equation}\label{orig-minres1-equiv}
\begin{array}{rl}
\mini \;\;\;&  \| \B\Sigma - ( \B\Theta + \B\Phi) \|_2^2\\
 \sbt\;\;\; & \rnk(\B\Theta) \leq r \\
 & \B\Theta \succeq \M{0} \\
 & \B\Phi = \diag(\Phi_1, \ldots, \Phi_p),
\end{array}
\end{equation}
with variables $\B\Theta, \B\Phi$, where $\Phi_{i}, i \geq 1$ are unconstrained. If $\widehat{\B\Theta}$ is a minimizer of Problem~\eqref{orig-minres1-equiv}, then \emph{any} 
$\widehat{\M{L}}$ satisfying $\widehat{\M{L}}\widehat{\M{L}}'=\widehat{\B\Theta}$ minimizes~\eqref{orig-minres1} and vice-versa. 

Various heuristic approaches are used to for Problem~\eqref{orig-minres1}. 
The {\texttt R} package \texttt{psych} for example, uses a black box gradient based tool \texttt{optim} to minimize the non-convex Problem~\eqref{orig-minres1} with respect to $\M{L}$.
 Once $\widehat{\M{L}}$  is estimated,  the diagonal entries of $\B\Phi$ are estimated as $\widehat{\Phi}_{i}=\sigma_{ii} -  (\widehat{\M{L}}\widehat{\M{L}}')_{ii}$, for $i \geq 1$.
 Note that  $\widehat{\Phi}_{i}$ obtained in this fashion may be negative\footnote{If $\widehat{\Phi}_{i}<0$, some ad-hoc procedure is used to threshold it to a nonnegative quantity.} and the condition $\B\Sigma - \B\Phi\succeq \M{0}$ may be violated.
\paragraph{Generalized Least Squares, Principal Axis and variants:}
The Ordinary Least Squares ({\textsc{\small OLS}}) method for FA~\citep{Bartholomew_Knott_Moustaki_2011} 
considers formulation~\eqref{orig-minres1-equiv} with the additional constraint that 
$\Phi_{i}\geq 0, \forall i$.  The Weighted Least Squares ({\textsc{\small WLS}}) or the generalized least squares 
method~(see for example, \cite{Bartholomew_Knott_Moustaki_2011})  considers 
a weighted least squares objective:
$$\left \| \M{W} \left ( \B\Sigma - ( \B\Theta + \B\Phi)  \right ) \right \|_2^2.$$
As in the ordinary least squares case, here too, we assume that $\Phi_{i} \geq 0$. Various choices of $\M{W}$ are possible depending upon the
 application---$\M{W} \in \{ \B{\Sigma}^{-1}, \B\Phi^{-1} \}$ being a couple of popular choices.
 
The Principal Axis (PA) FA method~\citep{Bartholomew_Knott_Moustaki_2011,psych}
is  popularly used to estimate factor model parameters based on criterion~\eqref{orig-minres1-equiv} along with the constraints 
$\Phi_{i} \geq 0, \forall i$. 
This method starts with a nonnegative estimate $\widehat{\Phi}_{i}$ 
and performs a rank $r$ eigendecomposition on  $\B\Sigma - \widehat{\B\Phi}$
to obtain $\widehat{\B\Theta}$. The matrix $\widehat{\B\Phi}$ is then updated to match the diagonal entries of $\B\Sigma - \widehat{\B\Theta}$, 
and the above steps are repeated until the estimate $\widehat{\B\Phi}$ stabilizes\footnote{We are not aware of a proof, however,  showing the convergence of this procedure.}.
Note that in this procedure, the estimate  $\widehat{\B\Theta}$ may fail to be PSD, 
the entries of $\widehat{\Phi}_{i}$ may be negative as well. 
Heuristic restarts and various initializations are often carried out to arrive at a reasonable solution~(see for example discussions in~\cite{Bartholomew_Knott_Moustaki_2011}).

In summary, the least squares stylized methods described above may lead to estimates that violate one or more of the constraints:
 $\B\Sigma -  \B\Phi \succeq \M{0}$ , $\B\Theta \succeq \M{0}$ and $\B\Phi \succeq \M{0}$.
 
\paragraph{Maximum Likelihood  for Factor Analysis:}
This approach~\citep{bai2012statistical,mle-factor-analysis,mardia,roweis1999unifying,em-factor-rubin} is another widely used method in FA and 
typically assumes that the data follows a multivariate Gaussian distribution. This procedure maximizes a likelihood function and is quite different 
from the loss functions pursued in this paper and discussed above.
The estimator need not exist for any $\B\Sigma$---see for example~\cite{robertson2007maximum}.

\medskip
\medskip
\medskip

Most of the methods described in Section~\ref{sec:related-work} are widely used and their implementations are available in statistical packages 
{\texttt{psych}}~\citep{psych}, {\texttt{nFactors}}~\citep{r-nfactors-13}, {\texttt{GPArotation}}~\citep{gparotation-r}, and others
and are publicly available from \texttt{CRAN}.\footnote{\url{http://cran.us.r-project.org}}

\subsection{Broad categories of factor analysis estimators}\label{sec:categories}

A careful investigation of  the methods described above suggests that they can be divided into two broad categories.
Some of the above estimators explicitly incorporate  a PSD structural assumption
on the residual covariance matrix $\B\Sigma - \B\Phi$ in \emph{addition} to requiring $\B\Theta \succeq \M{0}$ and $\B\Phi \succeq \M{0}$; while the others do not. As already pointed out, these constraints are important for statistical interpretability. 
We propose to distinguish between the following two broad categories of FA algorithms:
\begin{itemize}
\item[(A)] This category comprises of FA estimation procedures cast as nonlinear Semidefinite Optimization (SDO) problems---estimation takes place 
in the presence of constraints of the form $\B\Sigma - \B\Phi \succeq \M{0}$, along with $\B\Theta \succeq \M{0}$ and $\B\Phi \succeq \M{0}$.
Members of this category are~\textsc{\small MRFA},  \textsc{\small{MTFA}}  and more generally Problem~\eqref{obj-2-0}.

Existing approaches for these problems are typically not scalable: for example, we are not aware of any algorithm (prior to this paper) for Problem~\eqref{obj-2-0-margin-mrfa} (\textsc{\small MRFA})
that scales to covariance matrices with $p$ larger than thirty. Indeed, while theoretical guarantees of optimality exist in certain cases \citep{venkat-2012-factor}, the conditions required for such results to hold are generally difficult to verify in practice.

\item[(B)] This category includes classical FA methods which are not based on nonlinear SDO based formulations (as in Category (A)). 
 \textsc{\small{MINRES}}, \textsc{\small{OLS}}, \textsc{\small{WLS}}, \textsc{\small{GLS}}, \textsc{\small{PC}} and \textsc{\small{PA}} based FA estimation procedures (as described in Section~\ref{sec:related-work}) belong to this category.
These methods are generally scalable to problem sizes where $p$ is of
the order of a few thousand---significantly larger than most procedures belonging to Category (A); 
and are implemented in open-source {\texttt R}-packages.

\end{itemize}

\paragraph{Contributions:}
Our contributions in this paper may be summarized as follows:
\begin{enumerate}

\item We consider a flexible family of FA estimators, which can be obtained as solutions to rank-constrained nonlinear SDO problems. 
In particular, our framework provides a unifying perspective on several existing FA estimation approaches.

\item  We propose a novel \emph{exact} reformulation of the rank-constrained FA Problem~\eqref{obj-2-0} 
as a smooth optimization problem with convex compact constraints. 
We also develop a unified algorithmic framework, utilizing modern optimization techniques 
to obtain high quality solutions to Problem~\eqref{obj-2-0}. Our algorithms, at every iteration, simply require
computing a low-rank eigen-decomposition of a $p\times p$ matrix and a structured scalable convex 
SDO.
Our proposal is capable of solving  
FA problems involving covariance matrices having dimensions upto a few thousand;
thereby making it at par with the most scalable FA methods used currently.

\item Our SDO formulation enables us to estimate the underlying factors and {unique variances} under 
the restriction that the residual covariance matrix is PSD---a characteristic that is 
absent in several popularly used FA methods. This aids statistical interpretability, especially in drawing parallels with PCA and 
understanding the proportion of variance explained by a given number of factors. 
Methods proposed herein, produce superior quality estimates, in terms of various performance metrics, when compared to existing FA approaches. To our knowledge, this is the first paper demonstrating that certifiably optimal solutions to a rank-constrained problem can be found for 
problems of realistic sizes, without making any assumptions on the underlying data.

\item Using techniques from discrete and global optimization, we develop a branch-and-bound algorithm which proves that the low-rank solutions found are often optimal in seconds for problems on the order of $p=10$ variables, in days for problems on the order of $p=100$, and in minutes for some problems on the order of $p=4000$. As the selected rank increases, so too does the computational burden of proving optimality. It is particularly crucial to note that the optimal solutions for all problems we consider are found very quickly, and that vast majority of computational time is then spent on proving optimality. Hence, for a practitioner who is not particularly concerned with certifying optimality, our techniques for finding feasible solutions provide high-quality estimates quickly.

\item We provide computational evidence demonstrating the favorable performance of our proposed method.
Finally, to the best of our knowledge this is the first paper that views 
various FA methods in a unified fashion via a modern optimization lens and attempts to compare
a wide range of FA techniques in large scale. 

\end{enumerate}

\paragraph{Structure of the paper:}
The paper is organized as follows.
In Section~\ref{sec:intro1} we propose a flexible family of optimization Problems~\eqref{obj-2-0} for the task of statistical estimation in FA models.
 Section~\ref{sec:method1} presents an exact reformulation of Problem~\eqref{obj-2-0} 
 as a nonlinear SDO without the rank constraint. Section~\ref{sec:CG-method1} describes the use of nonlinear optimization
techniques such as the Conditional Gradient (CG) method~\citep{bertsekas-99}
adapted to provide feasible solutions (upper bounds) to our formulation. 
First order methods employed to compute the convex SDO subproblems are also described in the same section.
In Section~\ref{sec:alg}, we describe our method for certifying optimality of the solutions from Section~\ref{sec:CG-method1} in the case when $q=1$. In Section~\ref{sec:compute}, we present computational results demonstrating the effectiveness of our proposed method in terms of
(a) modeling flexibility in the choice of the number of factors $r$ and the parameter $q$ (b) scalability and (c) the quality of solutions obtained in a wide array of real and synthetic datasets---comparisons with several existing methods for FA are considered.
Section~\ref{sec:conclude} contains our conclusions.

\section{Reformulations}\label{sec:method1}

Let $\B{\lambda}(\M{A})$ and $\B{\lambda}(\M{B})$ denote the vector of eigenvalues of $\M{A}$ and  $\M{B}$, respectively, arranged in decreasing order, i.e., 
\begin{equation}\label{order-e-vals-1}
 \lambda_{1}(\M A) \geq \lambda_2(\M A) \geq \ldots \geq \lambda_{p}(\M A), \;\;\;\;  \lambda_{1}(\M B) \geq \lambda_2(\M B) \geq \ldots \geq \lambda_{p}(\M B).
 \end{equation}

The following proposition presents the first reformulation of Problem~\eqref{obj-2-0}  as a continuous  eigenvalue optimization problem with convex compact constraints. Proofs of all results can be found in Appendix A. 
\begin{prop}\label{lem:reform-1}
{\bf (a)} For any $q \geq 1$, Problem~\eqref{obj-2-0} is equivalent to:
\begin{equation}\label{obj-2-0-margin}
\begin{array}{lrl}
(\textsc{CFA}_{q}) \;\;\;\;\;\;\; & \mini \;\; & 
f_{q}(\B\Phi; \B\Sigma) := \displaystyle  \sum_{i=r+1}^{p} \lambda^q_i(\B{\Sigma} - \B{\Phi} ) \\
&\sbt \;\;\; & \B{\Phi}=\mathrm{diag}(\Phi_1, \ldots, \Phi_p) \succeq \M{0} \\
 &  &  \B\Sigma - \B\Phi \succeq \M{0},
\end{array}
\end{equation}
where  $\B\Phi$ is the optimization variable. 

{\bf (b)} Suppose $\B\Phi^*$ is a minimizer of Problem~\eqref{obj-2-0-margin}, and let
 \begin{equation}\label{eqn-exp-updt-theta}
 \B\Theta^* =  \M{U} \; \diag \big ( \lambda_{1}( \B{\Sigma} - \B{\Phi}^*  ), \ldots, \lambda_{r}(\B{\Sigma} - \B{\Phi}^* ) , 0, \ldots, 0 \big) \; \M{U}',
 \end{equation}
where $\M{U}_{p \times p}$  is the matrix of eigenvectors of $\B\Sigma - \B\Phi^*$. Then $(\B\Theta^*, \B\Phi^*)$ is a solution to Problem~\eqref{obj-2-0}.
\end{prop}

Problem~\eqref{obj-2-0-margin} is a nonlinear SDO in $\B\Phi$, 
unlike the original formulation~\eqref{obj-2-0} that estimates 
$\B\Theta$ and $\B\Phi$ jointly. Note that the rank constraint does not appear in Problem~\eqref{obj-2-0-margin} and the  constraint set of Problem~\eqref{obj-2-0-margin} is convex and compact.  
However, Problem~\eqref{obj-2-0-margin} is non-convex due to the non-convex objective function 
$\sum_{i={r+1}}^{p} \lambda^q_{i} (\B\Sigma - \B\Phi)$. 
For $q=1$  the function appearing in the objective of~\eqref{obj-2-0-margin}  is concave  and for $q >1$,  it  is neither convex nor concave.

\begin{prop}\label{equi-minus-0}
The estimation Problem~\eqref{obj-2-0} is equivalent to
\begin{equation}\label{obj-2-0-0}
\begin{aligned}
\mini\;\;\; & \;\;\; \| \B{\Sigma} - (\B{\Theta} + \B{\Phi})  \|^q_q \\
\sbt \;\;\; & \;\;\; \rnk(\B{\Theta})\leq r \\
&  \;\;\; \B{\Theta}  \succeq \M{0} \\
& \;\;\; \B{\Phi}=\mathrm{diag}(\Phi_1, \ldots, \Phi_p) \succeq \M{0} \\
&\;\;\; \B\Sigma - \B\Phi \succeq \M{0} \\
&\;\;\; \B\Sigma - \B\Theta \succeq \M{0}.
\end{aligned}
\end{equation}
\end{prop}

Note that Problem~\eqref{obj-2-0-0} has an additional PSD constraint $\B\Sigma - \B\Theta \succeq \M{0}$ which does not explicitly appear in Problem~\eqref{obj-2-0}. 
It is interesting to note that 
the two problems are equivalent. By virtue of  Proposition~\ref{equi-minus-0}, Problem~\eqref{obj-2-0-0} can as well be used as the estimation criterion for rank constrained FA.
However, we will work with formulation~\eqref{obj-2-0}  because it is easier to interpret from a  statistical perspecitve.

\paragraph{Special instances of $(\textsc{CFA}_{q})$:}
We show that some well-known FA estimation problems can be viewed as special cases of our general framework.

For $q=1$, Problem~\eqref{obj-2-0-margin} reduces to {\textsc{\small MRFA}}, 
as described in~\eqref{obj-2-0-margin-mrfa}.      
For $q =1$ and $r =0$, Problem~\eqref{obj-2-0-margin} reduces to {\textsc{\small MTFA}}~\eqref{orig-mtfa1}.
When $q=2,$ we get a variant of~\eqref{orig-minres1-equiv}, i.e., a PSD \emph{constrained} analogue of \textsc{\small{MINRES}}
\begin{equation}\label{obj-2-equiv22}
\begin{array} {lrl}
& \mini\limits_{  }\;\;\; & 
 \sum\limits_{i=r+1}^{p} \lambda^2_i(\B{\Sigma} - \B{\Phi} ) \\
&\sbt \;\;\; & \B{\Phi}=\mathrm{diag}(\Phi_1, \ldots, \Phi_p) \succeq \M{0} \\
 &  &  \B\Sigma - \B\Phi \succeq \M{0}. 
\end{array}
\end{equation}
Note that unlike Problem~\eqref{orig-minres1-equiv}, Problem~\eqref{obj-2-equiv22}
explicitly imposes PSD constraints on $\B\Phi$ and $\B\Sigma - \B\Phi$.
The objective function in~\eqref{obj-2-0-margin} is continuous but non-smooth. 
The function is differentiable at $\B\Phi$ if and only if  the $r$ and $(r+1)$th eigenvalues of $\B\Sigma - \B\Phi$ are distinct~\citep{lewis-96,shapiro-inference-FA-02}, i.e., 
$\lambda_{r+1}(\B{\Sigma} - \B{\Phi} )  < \lambda_{r}(\B{\Sigma} - \B{\Phi} ).$
The non-smoothness of the objective function in Problem~\eqref{obj-2-0-margin}, makes the use of standard gradient based methods problematic~\citep{bertsekas-99}.
Theorem~\ref{thm:FA-reform1-gen-q} presents  a reformulation of Problem~\eqref{obj-2-0-margin}  in which the objective function is continuously differentiable. 
\begin{theorem}\label{thm:FA-reform1-gen-q}
{\bf (a)} The estimation criterion given by Problem~\eqref{obj-2-0} 
is equivalent to\footnote{For any PSD matrix $\M{A}$, with eigendecomposition 
 $\M{A} = \M{U}_{A} \diag(\lambda_{1}, \ldots, \lambda_{p})\M{U}'_{A} $,
we define $\M{A}^q := \M{U}_{A} \diag(\lambda^{q}_{1}, \ldots, \lambda^{q}_{p})\M{U}'_{A}$, for any $q\geq 1$.}
\begin{equation}\label{FA-reform1-gen-q}
\begin{array} {l r c l }
\;\;\; & \mini\;\;\; & g_{q}(\M{W},\B{\Phi}):=& \tr(\M{W} (\B{\Sigma} - \B{\Phi} )^q ) \\
&\sbt \;\;\; && \B{\Phi}=\mathrm{diag}(\Phi_1, \ldots, \Phi_p) \succeq \M{0} \\
  & &&  \B\Sigma - \B\Phi \succeq \M{0} \\
  && & \M{I}\succeq \M{W} \succeq \M{0} \\
  &&& \tr(\M{W})= p - r .
\end{array}
\end{equation}
{\bf (b)} The solution~$\widehat{\B\Theta}$ of Problem~\eqref{obj-2-0} can be recovered from the solution
$\widehat{\M{W}},\widehat{\B\Phi}$ of Problem~\eqref{FA-reform1-gen-q} via:
\begin{equation}\label{eq:hat-theta1}
\widehat{\B\Theta} :=  \widehat{\M{U}} \diag(\widehat{\lambda}_{1}, \ldots, \widehat{\lambda}_{r}, 0, \ldots, 0) \widehat{\M{U}}',
\end{equation}
where $\widehat{\M{U}}$ is the matrix formed by  the $p$ eigenvectors corresponding to the eigenvalues 
$\widehat{\lambda}_{1}, \ldots, \widehat{\lambda}_{p}$ (arranged in decreasing order) of the matrix $\B\Sigma - \widehat{\B\Phi}$. Given
$\widehat{\B\Phi}$, any solution $\widehat{\B\Theta}$ (given by~\eqref{eq:hat-theta1}) is independent of $q$.
\end{theorem}

In Problem~\eqref{FA-reform1-gen-q}, if we partially minimize the function $g_{q}(\M{W},\B{\Phi})$ over $\B\Phi$ (with fixed $\M{W}$), the resulting function is 
concave in $\M{W}$. This observation leads to the following proposition.
\begin{prop}\label{prop:margin-1}
The function $G_{q}(\M{W})$ obtained upon (partially) minimizing $g_{q}(\M{W},\B{\Phi})$ over $\B\Phi$ (with $\M{W}$ fixed) in Problem~\eqref{FA-reform1-gen-q}, given by
\begin{equation}\label{eg:marg-fnw1}
\begin{array}{l c l}
G_{q}(\M{W}):=& \inf\limits_{\B\Phi}&\; \;\; g_{q}(\M{W},\B{\Phi})\\
&\sbt\;\; & \B\Phi = \diag(\Phi_{1}, \ldots, \Phi_{p})\succeq \M{0} \\
&&  \B\Sigma - \B\Phi \succeq \M{0},
\end{array}
\end{equation}
is concave in $\M{W}$. The sub-gradients of the function $G_{q}(\M W)$ exist and are given by:
\begin{equation}\label{eq:subdiff1}
\nabla G_{q}(\M{W}) = (\B\Sigma - \widehat{\B\Phi}(\M{W}))^{q},
\end{equation}
where $\widehat{\B\Phi}(\M{W})$ is a minimizer of the convex optimization Problem~\eqref{eg:marg-fnw1}.
\end{prop}

In light of Proposition~\ref{prop:margin-1},  we present another reformulation of Problem~\eqref{obj-2-0}, as the following concave minimization problem:
\begin{equation}\label{eq:opt-margin1}
\begin{array}{rl}
\mini\limits_{} & \;\; G_{q}(\M{W})  \\
\sbt&  \;\; \M{I} \succeq \M{W} \succeq \M{0} \\
& \;\;  \tr(\M{W}) = p - r ,
\end{array}
\end{equation}
where  the function $G_{q}(\M{W})$ is differentiable if and only if  $\widehat{\B\Phi}(\M{W})$ is unique. 

Note that  by virtue of Proposition~\ref{lem:reform-1},
 Problems~\eqref{obj-2-0-margin} and~\eqref{eq:opt-margin1} are equivalent.
 It is natural to ask therefore, whether one formulaton  might be favored over the other from a computational perspective. 
 Towards this end, note that both Problems~\eqref{obj-2-0-margin} and~\eqref{eq:opt-margin1} involve the minimization of 
 a non-smooth objective function, over convex compact constraints.
  However,
the objective function of Problem~\eqref{obj-2-0-margin} is non-convex (for $q>1$) whereas the one
 in Problem~\eqref{eq:opt-margin1} is concave (for all $q\geq 1$). 
We will see in Section~\ref{sec:CG-method1} that, 
CG based algorithms can be applied to a concave minimization problem (even if the objective function is not differentiable)---CG  
applied to general non-smooth objective functions, however,  has limited convergence guarantees. Thus, formulation~\eqref{eq:opt-margin1}
is readily amenable to CG based optimization algorithms, unlike formulation~\eqref{obj-2-0-margin}. This seems to make Problem~\eqref{eq:opt-margin1} 
computationally more appealing than Problem~\eqref{obj-2-0-margin}.

\section{Upper Bounds}\label{sec:CG-method1}
This section presents a unified computational framework for the class of problems ($\textsc{CFA}_{q}$).
Problem~\eqref{FA-reform1-gen-q} is a non-convex smooth optimization problem and obtaining a stationary point is quite challenging.
 We propose iterative schemes based on the Conditional Gradient Algorithm~\citep{bertsekas-99}---a generalization of  
 the Frank Wolfe Algorithm~\citep{frank-wolfe} to obtain a stationary point of the problem. The appealing aspect of our framework is that  
every iteration of the algorithm requires solving a convex SDO problem, which is computationally tractable. 
 While off the shelf 
interior point algorithms can be used to solve the convex SDO problems, they typically do not scale well for large problems due to intensive memory requirements.
In this vein, first order algorithms have received a lot of attention~\citep{nest_05,nest-07,nesterov2004introductory} in convex optimization of late, due 
to their low cost per iteration, low-memory requirements 
 and ability to deliver solutions of moderate accuracy for large problems, within a modest time limit. 
 We use first order methods to solve the convex SDO problems. 
We present two schemes based on the CG algorithm: 

\smallskip\smallskip

\noindent {\bf Algorithm~\ref{algo:CG:factor1}:} This  
procedure, described in Section~\ref{sec:smooth-cg} applies CG to Problem~\eqref{FA-reform1-gen-q}, where the objective function 
$g_{q}(\M{W}, \B\Phi)$ is smooth. 

\smallskip\smallskip

\noindent {\bf Algorithm~\ref{algo:ao1-sk}:} This scheme, described in Section~\ref{sec:ao1}, applies CG on
the optimization Problem~\eqref{eq:opt-margin1}, where the
 function $G_{q}(\M{W})$ defined in~\eqref{eg:marg-fnw1} is concave (and possibly non-smooth).

\smallskip

To make notations simpler, we will 
use the following shorthand:
\begin{equation}\label{feas-W}
{\B \Psi}_{p,p- r} := \{ \M{W} : \M{I}\succeq \M{W} \succeq \M{0},\;\; \tr(\M{W})= p - r \}
\end{equation}
and $\F = \{\Ph: \s-\Ph\psd,\;\; \Ph = \diag(\phi_1,\ldots,\phi_p)\psd\}$.

\subsection{A CG based algorithm for formulation~\eqref{FA-reform1-gen-q}}\label{sec:smooth-cg}
Since Problem~\eqref{FA-reform1-gen-q} involves the minimization of a smooth function over a compact convex set, 
the CG method requires iteratively solving the following convex optimization problem:
\begin{equation}
\begin{aligned}\label{smooth-cg-bh1}
\mini_{} \;\;& \left \langle \nabla  g_{p}(\M{W}^{(k)}, \B{\Phi}^{(k)}), \begin{pmatrix} \M{W} \\ \B\Phi \end{pmatrix}   \right \rangle\\
\sbt \;\; & \M{W} \in {\B \Psi}_{p,p- r} \\
& \Ph\in\Fs,
\end{aligned}
\end{equation}
where $\nabla  g_{p}(\M{W}^{(k)}, \B{\Phi}^{(k)})$ is the derivative of $ g_{p}(\M{W}^{(k)}, \B{\Phi}^{(k)})$ at the 
current iterate $(\M{W}^{(k)},$ $ \B{\Phi}^{(k)})$. Note that  due to the separability of the constraints in $\M{W}$ and $\B\Phi$, 
Problem~\eqref{smooth-cg-bh1} splits into two independent optimization problems 
with respect to $\M{W}$ and $\B\Phi$.
The overall procedure is outlined in Algorithm~\ref{algo:CG:factor1}.
\begin{algorithm}[h]
 \begin{itemize}
\item[1] Initialize with $(\M{W}^{(1)}, \B{\Phi}^{(1)})$, feasible for Problem~\eqref{FA-reform1-gen-q} and repeat the following 
Steps~2-3 until the convergence criterion described in~\eqref{smooth-cg-convtol1} is met.
\item[2] Solve the linearized Problem~\eqref{smooth-cg-bh1}, which requires solving two separate (convex) SDO problems over $\M{W}$ and $\B\Phi$:
\begin{flalign}
\overline{\M{W}}^{(k+1)} &\in&  \argmin_{\M{W} \in \B\Psi_{p, p -r} } \;\;\; & \langle \M{W}, \nabla_{\M{W}} g_{p}(\M{W}^{(k)}, \B{\Phi}^{(k)}) \rangle \;\;\;\; &&&& \label{eq:cg-fa-updw}\\
\overline{\B\Phi}^{(k+1)} &\in& \argmin_{\B\Phi\in \Fs}\;\; &\;\;  \langle \B\Phi, \nabla_{\B\Phi} g_{p}(\M{W}^{(k)}, \B{\Phi}^{(k)}) \rangle  \;\;\;\; \label{eq:cg-fa-updphi}&&&& 
\end{flalign}
where $\nabla_{\M{W}} g_{p}(\M{W}, \B{\Phi})$  (and   $\nabla_{\B\Phi} g_{p}(\M{W}, \B{\Phi})$) is the partial derivative with respect to $\M{W}$ 
(respectively, $\B\Phi$). 
\item[3] Obtain the new iterates:
\begin{equation}
\begin{array}{l c l c l}
\M{W}^{(k+1)} &=& \M{W}^{(k)} &+&  \eta_k (\overline{\M{W}}^{(k+1)} - \M{W}^{(k)}), \\
\B{\Phi}^{(k+1)} &=& \B{\Phi}^{(k)} &+&  \eta_k (\overline{\B{\Phi}}^{(k+1)} - \B{\Phi}^{(k)}) .
 \end{array}
 \end{equation} 
 with $\eta_k \in [0,1]$ chosen via
 an Armijo type line-search rule~\citep{bertsekas-99}.
\end{itemize}
\caption{A CG based algorithm for Problem~\eqref{FA-reform1-gen-q}}\label{algo:CG:factor1}
\end{algorithm}

Since $\nabla_{\M{W}} g_{p}(\M{W}^{(k)}, \B{\Phi}^{(k)})=(\B\Sigma -  \B{\Phi}^{(k)})^{q}$,
the update for $\M{W}$ appearing in~\eqref{eq:cg-fa-updw} requires solving:
\begin{equation}\label{eq:cg-fa-updw-2}
\mini_{\M{W} \in \B\Psi_{p, p - r} } \;\;\;  \langle \M{W}, (\B\Sigma -  \B{\Phi}^{(k)})^{q} \rangle.
\end{equation}

Similarly, the update for $\B\Phi$ appearing 
in~\eqref{eq:cg-fa-updphi} requires solving:
\begin{equation}  \label{eq:cg-fa-updphi-2}
\begin{aligned}
\mini_{\B\Phi\in \Fs} \;\; \;\;   \sum_{i=1}^{p} \Phi_{i}\ell_{i},
\end{aligned}
\end{equation}
where the vector $(\ell_{1}, \ldots, \ell_{p})$  is given by
$\diag(\ell_1, \ldots, \ell_{p})  = -q\; \diag \left ( \M{W^{(k)}}(\B\Sigma - \B\Phi)^{q-1}  \right ),$
where $\diag(\M{A})$ is a diagonal matrix having the same diagonal entries as $\M{A}$.
The sequence $(\M{W}^{(k)},\B\Phi^{(k)})$ is recursively computed via Algorithm~\ref{algo:CG:factor1}, until a convergence criterion is met:
\begin{equation}\label{smooth-cg-convtol1}
 g_{p}(\M{W}^{(k)},\B\Phi^{(k)}) - g_{p}(\M{W}^{(k+1)},\B\Phi^{(k+1)})  \leq \mbox{\small TOL} \cdot g_{p}(\M{W}^{(k)},\B\Phi^{(k)}),
\end{equation}
for some user-defined tolerance $ \mbox{\small TOL} >0$.

A tuple $(\M{W}^*, \B\Phi^*)$ satisfies the first order stationarity conditions~\citep{bertsekas-99}  for Problem~\eqref{FA-reform1-gen-q}, if the following condition is satisfied:
\begin{equation}\label{fo-stat-11}
\begin{aligned}
\inf \;\;\;& \left \langle \nabla  g_{p}(\M{W}^{*}, \B{\Phi}^{*}), \begin{pmatrix} \M{W} - \M{W}^{*} \\ \B\Phi - {\B\Phi}^{*}  \end{pmatrix}   \right \rangle & \geq \quad&\quad\;\; 0 \\
\sbt \;\;\; & \M{W} \in {\B \Psi}_{p,p- r} \\
& \Ph\in\Fs.
\end{aligned}
\end{equation}
Note that $\B\Phi^*$ defined above also satisfies the first order stationarity conditions for Problem~\eqref{obj-2-0-margin}.

The following theorem presents a global convergence guarantee for Algorithm~\ref{algo:CG:factor1}:
\begin{theorem}\citep{bertsekas-99} \label{thm:cg-bert1}
Every limit point $(\M{W}^{(\infty)},\B\Phi^{(\infty)})$ of a sequence  $(\M{W}^{(k)},$ $\B\Phi^{(k)})$ produced by
Algorithm~\ref{algo:CG:factor1}
is a first order stationary point of the optimization Problem~\eqref{FA-reform1-gen-q}.
\end{theorem}

\subsection{A CG based algorithm for Problem~\eqref{eq:opt-margin1}}\label{sec:ao1}
The CG method for Problem~\eqref{eq:opt-margin1} 
requires solving a linearization of the concave objective function.
At iteration $k$, if $\M{W}^{(k)}$ is the current estimate of $\M{W}$, 
the new estimate $\M{W}^{(k+1)}$ is obtained by
\begin{equation} \label{upd-w-ao-1}
\begin{aligned}
\M{W}^{(k+1)} &\in&   \argmin_{\M{W} \in \B\Psi_{n, p-r} } \;\; \left\langle \nabla G_{q}(\M{W}^{(k)})  , \M{W} \right\rangle & = & \argmin_{\M{W} \in \B\Psi_{n, p-r} } \;\; \tr (\M{W} (\B\Sigma - \B\Phi^{(k)})^{q}) ,
 \end{aligned}
 \end{equation}
 where by Proposition~\ref{prop:margin-1}, $(\B\Sigma - \B\Phi^{(k)})^{q}$ is a sub-gradient of  $G_{q}(\M{W})$ at $\M{W}^{(k)}$ with
 $\B\Phi^{(k)}$ given by
\begin{equation}\label{upd-phi-ao-1}
\begin{array}{ll rl}
\B\Phi^{(k)} &\in& \argmin\limits_{\B\Phi\in\Fs}& \;\; \tr (\M{W}^{(k)} (\B\Sigma - \B\Phi)^{q})
\end{array}
\end{equation}
No explicit line search is necessary here, since the minimum will always be at the new point, i.e., $\B\Phi^{(k)}$ due to the concavity of the objective function.
The sequence $\M{W}^{(k)}$ is recursively computed via~\eqref{upd-w-ao-1}, until the  convergence criterion 
\begin{equation}\label{concave-cg-con-1}
G_{q}(\M{W}^{(k)}) -  G_{q}(\M{W}^{(k+1)}) \leq \mbox{\small TOL} \cdot G_{q}(\M{W}^{(k)}),
\end{equation}
for some user-defined tolerance $ \mbox{\small TOL} >0$, is met.
A short description of the procedure appears in Algorithm~\ref{algo:ao1-sk}.
\begin{algorithm} [htpb!]
{{\fontsize{12pt}{12pt}\selectfont
\begin{itemize}
\item[1] Initialize with $\M{W}^{(1)}$ ($k = 1$), feasible for Problem~\eqref{eq:opt-margin1} and repeat, for $k \geq 2$,  Steps~2-3 until 
convergence criterion~\eqref{concave-cg-con-1} is satisfied.
\item[2] Update $\B\Phi$ (with  $\M{W}$ fixed) by solving~\eqref{upd-phi-ao-1}.
\item[3] Update $\M{W}$ (with $\B\Phi$  fixed) by solving~\eqref{upd-w-ao-1}, to get $\M{W}^{(k+1)}$.
\end{itemize}
\caption{A CG based algorithm for formulation~\eqref{eq:opt-margin1}}\label{algo:ao1-sk}}}
\end{algorithm}

Before we present the convergence rate of Algorithm~\ref{algo:ao1-sk}, we will need to introduce some  notation. 
For any point $\overline{\M W}$ belonging to the feasible set of Problem~\eqref{eq:opt-margin1}
let us define $\Delta(\overline{\M{W}})$ as follows:
\begin{equation}\label{defn-delta-w-1}
\Delta(\overline{\M{W}}) : = \inf_{\M{W} \in \B\Psi_{p , p - r}} \langle \nabla G_{q}(\overline{\M{W}}) ,  \M{W} - \overline{\M{W}}  \rangle.
\end{equation}
$\M{W}^*$ satisfies the first order stationary condition for Problem~\eqref{eq:opt-margin1} if: 
$\M{W}^*$ is feasible for the problem and $\Delta(\M{W}^*) \geq 0$.

We now present Theorem~\ref{lem:conv-rate-1} establishing 
the rate of convergence and associated convergence properties of Algorithm~\ref{algo:ao1-sk}.

\begin{theorem}\label{lem:conv-rate-1}
If $\M{W}^{(k)}$ is a sequence produced by Algorithm~\ref{algo:ao1-sk}, then 
$ G_{q}({\M{W}}^{(k)})$ is a monotone decreasing sequence and
every limit point $\M{W}^{(\infty)}$ of the sequence $\M{W}^{(k)}$ is a stationary point of  Problem~\eqref{eq:opt-margin1}.
Furthermore, Algorithm~\ref{algo:ao1-sk} has a convergence rate of 
$O(1/K)$ (with $K$ denoting the iteration index) to a first order stationary point of Problem~\eqref{eq:opt-margin1}, i.e., 
\begin{equation}\label{complexity-rule1-adapt}
\min\limits_{i=1, \ldots, K}  \left \{ -\Delta(\M{W}^{(i)}) \right\}  \leq \frac{ G_{q}({\M{W}}^{(1)} )- G_{q}({\M{W}}^{(\infty)} )}{K}. 
\end{equation}
\end{theorem}

\subsection{Solving the convex SDO Problems}\label{sec:subprob1}
Both the variants of CG---Algorithms~\ref{algo:CG:factor1} and~\ref{algo:ao1-sk} presented in Sections~\ref{sec:smooth-cg} and~\ref{sec:ao1}---require sequentially solving convex SDO problems in $\M{W}$ and $\B\Phi$. 
We describe herein, how these subproblems can be solved efficiently.
\subsubsection{Solving the SDO problem with respect to $\M{W}$}\label{sec:subprob-w1}
A generic SDO problem associated with Problems~\eqref{upd-w-ao-1} and~\eqref{eq:cg-fa-updw-2} requires updating $\M{W}$ as:
\begin{equation}\label{eq:gen-upd-w1}
\widehat{\M{W}}  \;\;\; \in \;\;\; \argmin_{\M{W} \in \B\Psi_{p, p-r}} \;\;\; \langle \M{W}, \widetilde{\M{W}}  \rangle,
\end{equation}
for some fixed symmetric $\widetilde{\M{W}}_{p \times p}$, depending upon the algorithm and the choice of $q$, described as follows.
For Algorithm~\ref{algo:CG:factor1}  the update for $\M{W}$ at iteration $k$ 
for solving Problem~\eqref{eq:cg-fa-updw-2} corresponds to $\widetilde{\M{W}}=(\B\Sigma -  \B{\Phi}^{(k)})^{q}$.
For Algorithm~\ref{algo:ao1-sk}  the update in $\M{W}$ at iteration $k$ for Problem~\eqref{upd-w-ao-1},
corresponds to $\widetilde{\M{W}}= (\B\Sigma - \B\Phi^{(k+1)})^{q}$.

A solution to Problem~\eqref{eq:gen-upd-w1}  is given by 
$\widehat{\M{W}} = \sum_{i = r +1 }^{p} \M{u}_{i} \M{u}'_{i} $, where $\M{u}_1, \ldots, \M{u}_{p}$ are the eigenvectors of the matrix 
$\widetilde{\M{W}}$, corresponding to the eigenvalues 
$\lambda_{1}(\widetilde{\M{W}}) , \ldots, \lambda_{p}(\widetilde{\M{W}})$.

\subsubsection{Solving the SDO problem with respect to $\B\Phi$}\label{sec:subprob-phi1}
The SDO problem arising from the update of $\B\Phi$ is
not as straightforward as the update with respect to $\M{W}$.
Before presenting the general case, it helps to consider a few special cases of ($\textsc{CFA}_{q}$).

For $q=1$  the objective function of Problem~\eqref{FA-reform1-gen-q}
\begin{equation}\label{eq:g_1-c-d-1} 
 g_{1}(\M{W}, \B\Phi) =  \langle \M{W}, \B\Sigma \rangle - \sum_{i=1}^{p} w_{ii} \Phi_{i}
\end{equation}
is linear in $\B\Phi$ (for fixed $\M{W}$). For $q=2$, the objective function of Problem~\eqref{FA-reform1-gen-q}
\begin{equation}\label{line1-simpli}
g_{2}(\M{W}, \B\Phi) =  \tr (\M{W} \B{\Sigma}^2) + \sum_{i=1}^{p} (w_{ii}\Phi_{i}^2 -  2\langle \M{w}_{i},\B\sigma_{i} \rangle \Phi_{i})  
\end{equation}
is a convex quadratic in $\B\Phi$ (for fixed $\M{W}$).

While solving Problem~\eqref{smooth-cg-bh1} (as a part of Algorithm~\ref{algo:CG:factor1}) all subproblems with respect to $\B\Phi$ have a linear objective function. 
For Algorithm~\ref{algo:ao1-sk}, the partial minimizations with respect to $\B\Phi$, 
for $q=1$ and $q=2$,  
require minimizing Problems~\eqref{eq:g_1-c-d-1} and~\eqref{line1-simpli}, respectively. 

 Various instances of optimization problems with respect to $\B{\Phi}$ appearing in Algorithms~\ref{algo:CG:factor1} and~\ref{algo:ao1-sk},  
can be viewed as special cases of  the following family of SDO problems:
\begin{equation}\label{gen-quad-upd-phi-2}
\begin{array} {l r l}
& \mini\limits_{\B\Phi\in\Fs} \;\; & \sum\limits_{i=1}^{p}  \left ( c_{i} \Phi^2_i + d_i \Phi_i  \right)
\end{array}
\end{equation}
where $c_{i}\geq 0$ and $d_{i}$ for $i = 1, \ldots, p$ are problem parameters
that depend upon the choice of Algorithm~\ref{algo:CG:factor1} or~\ref{algo:ao1-sk} and $q$. 
We now present a first order convex optimization scheme for solving~\eqref{gen-quad-upd-phi-2}.
 \subsubsection*{A First Order Scheme for Problem~\eqref{gen-quad-upd-phi-2}}
 With the intention of providing a
 simple and scalable algorithm for the convex SDO problem,
 we use the Alternating Direction Method of Multipliers~\citep{bertsekas-99,boyd-admm1} (\textsc{\small{ADMM}}). 
 We introduce a splitting variable $\B\Lambda=\B\Sigma - \B\Phi$ and rewrite Problem~\eqref{gen-quad-upd-phi-2} in the following equivalent form:
\begin{equation}\label{FA-reform1-phi-1}
\begin{aligned}
\mini_{\B{\Phi},\B\Lambda}\;\; \;&  \sum_{i=1}^{p}  ( c_i \Phi^2_i + d_i \Phi_i  ) \\
\sbt \;\;\; & \B{\Phi}=\mathrm{diag}(\Phi_1, \ldots, \Phi_p) \succeq \M{0} \\
  & \B\Lambda \succeq \M{0} \\
  & \B\Lambda =  \B\Sigma - \B\Phi .
 \end{aligned}
\end{equation}
The Augmented Lagrangian for the above problem is:
\begin{equation}\label{FA-reform1-phi-lag}
{\cal L}_\rho ( \B\Phi, \B\Lambda, \B\nu) := 
  \sum_{i=1}^{p}  ( c_i \Phi^2_i + d_i \Phi_i  )  + \Big \langle \B\nu, \B\Lambda - \left ( \B\Sigma - \B\Phi \right) \Big\rangle + \frac{\rho}{2}\| \B\Lambda - ( \B\Sigma - \B\Phi) \|_2^2,
\end{equation}
where $\rho>0$ is a scalar, $\langle \cdot , \cdot \rangle$ denotes the standard trace inner product.  
\textsc{\small{ADMM}} involves the following three updates:
\begin{align}
\B\Phi^{(k+1)} \in \;\;\;& \argmin_{\B{\Phi}=\mathrm{diag}(\Phi_1, \ldots, \Phi_p) \succeq \M{0}}  \;\;\; {\cal L}_\rho ( \B\Phi, \B\Lambda^{(k)}, \B\nu^{(k)}),  \label{line-1-phi}&&&& \\
\B\Lambda^{(k+1)} \in \;\;\;& \argmin_{\B\Lambda \succeq \M{0}}  \;\;\; {\cal L}_\rho ( \B\Phi^{(k+1)}, \B\Lambda, \B\nu^{(k)}) ,\label{line-2-phi} &&&& \\
\B\nu^{(k+1)} =  \;\;\;& \B\nu^{(k)} + \rho(\B\Lambda^{(k+1)} - (\B\Sigma - \B\Phi^{(k+1)}) ),&\label{line-3-phi} &&&&
\end{align}
and produces a
sequence $\{(\B\Phi^{(k)},\B\lambda^{(k)},\B\nu^{(k)})\}, k \geq 1$ ---
the convergence properties of  the algorithm are quite well known~\citep{boyd-admm1}.

Problem~\eqref{line-1-phi} can be solved in closed form as:
\begin{equation}
\begin{aligned}
\text{i.e}, \;\;\; & \Phi^{(k+1)}_{i} =& \frac{\rho}{\rho + 2c_{i}} \max \left \{  (\sigma_{ii} - \lambda^{(k)}_{ii})  - \frac{(d_{i} + \nu^{(k)}_{ii})}{\rho} , 0 \right \},\;\;\; i = 1,\ldots, p. 
\end{aligned}
\end{equation}
The update with respect to $\B\Lambda$ in~\eqref{line-2-phi} requires an eigendecomposition:
\begin{equation}\label{upd-fa-lam-1}
\begin{aligned}
\B\Lambda^{(k+1)}=& \argmin_{\B{\Lambda} \succeq \M{0}} \;\; \left \| \B{\Lambda} - ( \B\Sigma - \B\Phi^{(k+1)} -   \frac{1}{\rho} \B\nu^{(k)} ) \right \|_2^2 \\
=& {\cal P}_{S^{+}_p}\left ( \B\Sigma - \B\Phi^{(k+1)} -  \frac{1}{\rho} \B\nu^{(k)}  \right ),
\end{aligned}
\end{equation}
where, the operator ${\cal P}_{S^{+}_p}(\M{A})$ denotes the projection of a symmetric matrix $\M{A}$ onto the cone of PSD matrices of dimension $p \times p$: 
${\cal P}_{S^{+}_p}(\M{A}) = \M{U}_{A}\diag\bigg(\max \left\{ \lambda_{1}, 0 \right \} , \ldots, \max \left\{ \lambda_{p}, 0 \right\}\bigg) \M{U} _A',$ where,
$\M{A}= \M{U}_{A}\diag(\lambda_1, \ldots, \lambda_{p}) \M{U}'_{A}$ is the eigendecomposition of $\M{A}$.

\paragraph{Stopping criterion:}
The ADMM iterations~\eqref{line-1-phi}---\eqref{line-3-phi} are continued till the values of $\|\B\Lambda^{(k+1)} - (\B\Sigma - \B\Phi^{(k+1)}) \|_2$
 and the relative change in the objective values of Problem~\eqref{gen-quad-upd-phi-2} become 
smaller than a certain threshold, say, $\mathrm{TOL}\times\alpha$, where $\alpha \in \{ 10^{-1} , \ldots, 10^{-3} \}$ --- this 
is typically taken to be smaller than the convergence threshold for the CG iterations ($\mathrm{TOL}$).

\paragraph{Computational cost of Problem~\eqref{gen-quad-upd-phi-2}:}

The most intensive computational stage in the \textsc{\small{ADMM}} procedure is in performing the projection 
operation~\eqref{upd-fa-lam-1}---requiring $O(p^3)$ due to the associated eigendecomposition. This needs to be done for as many \textsc{\small{ADMM}} steps, until convergence. 

Since Problem~\eqref{gen-quad-upd-phi-2} is embedded inside iterative procedures like  Algorithms~\ref{algo:CG:factor1} and~\ref{algo:ao1-sk}, 
the estimates of $(\B\Phi, \B{\Lambda}, \B\nu)$ obtained by solving Problem~\eqref{gen-quad-upd-phi-2} for a iteration index (of the CG algorithm) 
 provides a good warm-start for the Problem~\eqref{gen-quad-upd-phi-2} in the subsequent CG iteration.  This is often found to decrease the number of iterations required by the ADMM algorithm to  converge to a prescribed level of accuracy. 

\subsection{Computational cost of Algorithms~\ref{algo:CG:factor1} and~\ref{algo:ao1-sk}}\label{sec:comp-complex-1}
For both Algorithms~\ref{algo:CG:factor1} and~\ref{algo:ao1-sk} the computational bottleneck is in performing 
the eigendecomposition of a $p \times p$ matrix: the $\M{W}$ update requires performing a low-rank eigendecomposition of a 
$p \times p$ matrix and the $\B\Phi$ update requires solving a problem of the form~\eqref{gen-quad-upd-phi-2}, which 
also costs $O(p^3)$.  Since eigendecompositions can easily be done for $p$ of the order of a few thousands, the proposed 
algorithms can be applied to that scale.

Note that most existing popular algorithms for FA, belonging to Category (B) (see Section~\ref{sec:categories}) also perform an 
eigendecomposition with cost $O(p^3)$. Thus it appears that Category (B) and the algorithms proposed herein have the same computational complexity 
and hence these two estimation approaches are equally scalable. 

While both Algorithms~1 and~2 can be used for $q \in \{1, 2\}$, for arbitrary $q \geq 1$ we recommend using Algorithm~1 due to the simple 
linear SDO problems in $\M{W}$ and $\B\Phi$ that needs to be solved at every iteration.


\section{Certificates of Optimality via Lower Bounds}
\label{sec:alg}

In this section, we outline our approach to computing lower bounds to ($\textsc{CFA}_{q}$). In particular, we focus on the case when $q=1$. We begin with an overview of the method. We then discuss initialization parameters for the method as well as several algorithms for solving subproblems. We also discuss branching rules and other refinements employed in our approach. For notational simplicity, we will always use $\ph\in\R^p$ to denote the vector with $\ph = \diag(\Ph)$ and we will often omit $\diag(\cdot)$ notation when clear. For example, instead of a constraint such as $\s-\Ph\psd$ we will instead write $\s-\ph\psd$ to suggest that we have constraints $\s-\Ph\psd$ and $\Ph=\diag(\ph)$ is a diagonal matrix.

\subsection{Overview of Method}\label{ssec:methodOverview}

Our primary problem of interest is to provide lower bounds to ($\textsc{CFA}_{1}$), i.e.,
\begin{equation}\label{eqn:FAequiv}
\mini_{\Ph\in\F}\sum_{i>r}\sigma_i(\s-\Ph),
\end{equation}
or equivalently using the results of Section \ref{sec:method1},
\begin{equation}\label{eqn:FAauxVars}
\mini_{\substack{\W\in\cW\\\Ph\in\F}}\langle \W,\s-\Ph\rangle.
\end{equation}
We begin with a definition and result from convex analysis.

\begin{definition}
For a function $f:\Gamma\sub\R^n\to\R$ we define its convex envelope on $\Gamma$, denoted $\convenv_\Gamma(f)$, to be the largest convex function $g$ with $g\leq f$. In symbols,
$$g = \sup\{h:\Gamma\to\R \;|\; h \text{ convex on $\Gamma$ and } h\leq f\}.$$
\end{definition}

The convex envelope of certain functions is well-known. One such function of principle interest here is described in the following theorem. This is in some contexts referred to as a McCormick hull and is widely used throughout the nonconvex optimization literature \citep{floudas,tawarBook}.

\begin{theorem}[\citealp{tawarBook,akf}]\label{thm:mcHulls}
If $f:\Gamma=[0,1]\times[\ell,u]\to\R$ is defined by $f(x,y) = -xy$, then the convex envelope of $f$ on $\Gamma$ is precisely
$$\convenv_\Gamma(f)(x,y)  = \max\{-ux,\ell - \ell x - y\}.$$
In particular, if $|u-\ell|\to0$, then $\convenv_\Gamma(f)\to f$.
\end{theorem}

Using this result we are now ready to describe our approach for computing lower bounds to \eqref{eqn:FAauxVars}. Here for a fixed $\bell,\uu\in\R^p$ with $\bell\leq \uu$ we will repeatedly solve a problem of the form
\begin{equation}\label{eqn:bbLB}
\begin{array}{lc}
\ds\mini_{\W,\Ph,\zz}&\langle \W,\s\rangle-\ds\sum_iz_i\\
\st &\ph\in\F\\
&\bell\leq\ph\leq\uu\\
&\W\in\cW\\
&z_i \leq \phi_i+\ell_iW_{ii}-\ell_i\;\forall i\\
&z_i\leq u_iW_{ii}\;\forall i.
\end{array}
\end{equation}
Here $\zz$ is used as an auxiliary variable to represent convex envelopes; namely, for $\phi_i\in[\ell_i,u_i]$ and $W_{ii}\in[0,1]$, then $\convenv(-W_{ii}\phi_{i})(W_{ii},\phi_i) = \max\{\ell_i - \ell_iW_{ii}-\phi_i,-u_iW_{ii}\}$. This maximum can then be represented with auxiliary variables. Using this problem setup our method is displayed in pseudocode in Algorithm \ref{alg:bb}. The key ingredient is the use of McCormick hulls for the products $-W_{ii}\phi_i$, embedded in a branch and bound scheme.

More concretely, Algorithm \ref{alg:bb} involves treating a given node $\n=[\bell,\uu]$, which represents bounds on $\ph$, namely, $\bell\leq\ph\leq\uu$. Here we solve \eqref{eqn:bbLB} with the upper and lower bounds $\uu$ and $\bell$, respectively, and see whether the resulting new feasible solution is better (lower in objective value) than the best known incumbent solution encountered thus far. We then see if the the bound for such $\n$ (as in \eqref{eqn:bbLB}) is better than the currently known best feasible solution; if it is not at least the current best feasible solution's objective value (up to some numerical tolerance), then we must further branch on this node, generating two new nodes $\n_1$ and $\n_2$ which partition the existing node $\n$. Throughout, we keep track of the worst lower bound encountered, which allows for us to terminate the algorithm early while still having a provable suboptimality guarantee on the best feasible solution $\ph_f\in\F$ found thus far.

In light of Theorem \ref{thm:mcHulls}, we have as an immediate corollary the following theorem.

\begin{theorem}\label{thm:bbAlgWorks}
Given numerical tolerance $\tol>0$, Algorithm \ref{alg:bb} terminates in finitely many iterations and solves \eqref{eqn:FAequiv} to within an additive optimality gap of at most \tol. Further, if Algorithm \ref{alg:bb} is terminated early (i.e., before $\nodes=\emptyset$), then the best feasible solution $\ph_f$ at termination is guaranteed to be within an additive optimality gap of $z_f-z_\text{lb}$.
\end{theorem}

The algorithm we have considered here omits some important details. After discussing properties of \eqref{eqn:bbLB}, we will discuss various aspects of Algorithm \ref{alg:bb}. In Section \ref{ssec:params}, we detail how to choose input $\uu^0$. In Section \ref{ssec:subSolve} we give a variety of algorithms which can be used to solve \eqref{eqn:bbLB} (line 5 in Algorithm \ref{alg:bb}). We then turn our attention to branching (lines 13 and 14) in Section \ref{ssec:branching}. In Section \ref{ssec:weyl}, we use results from matrix analysis coupled with ideas from the modern practice of discrete optimization to make tailored refinements to Algorithm \ref{alg:bb}. In Section \ref{ssec:nodeSelect}, we include a discussion of node selection strategies. Finally, in Section \ref{ssec:globOptSOA} we discuss how our approach relates to state-of-the-art methods for globlal optimization.

\begin{algorithm}[h!]
 \begin{enumerate}
\item Given optimality tolerance \tol, a feasible solution $\ph_f$, and upper bounds $\uu^0$ so that $\phi_i\leq u_i^0$ for all $\ph\in\F$, initialize $z_f \leftarrow \sum_{i>r} \sigma_i(\s-\ph_f)$ and  $\nodes= \{ ([\mb 0,\uu^0],+\infty)\}$, $z_\text{lb} = -\infty$.

\item While $\nodes\neq\emptyset$, remove some node $([\bell^c,\uu^c],z^c)\in\nodes$.

\item Solve \eqref{eqn:bbLB} with $\bell=\bell^c$ and $\uu=\uu^c$. Let $\ph$ be an optimal solution to \eqref{eqn:bbLB}, with $z$ the optimal objective value to \eqref{eqn:bbLB}, and set $z_u\leftarrow \sum_{i>r} \sigma_i(\s-\ph)$.

\item If $z_u < z_f$ (i.e., a better feasible solution is found), update the best feasible solution found thus far ($\ph_f$) to be $\ph$ and update the corresponding value ($z_f$) to $z$.

\item If $z_u < z_f - \tol$ (i.e., a \tol--optimal solution has not yet been found), then pick some $i\in\{1,\ldots,p\}$ and some $\alpha\in(\ell_i^c,u_i^c)$. Then let add two new nodes to \nodes:
$$\left(\prod_{j<i} [\ell_j^c,u_j^c]\times [\ell_i^c,\alpha]\times \prod_{j>i}[\ell_j^c,u_j^c],z\right)\text{\quad and \quad}\left( \prod_{j<i} [\ell_j^c,u_j^c]\times [\alpha,u_i^c]\times \prod_{j>i}[\ell_j^c,u_j^c],z\right) .$$
Update the best lower bound $z_\text{lb}$ to be $\ds\min_{\substack{([\bell^c,\uu^c],z)\in\nodes }}z$ and return to Step 2.

\end{enumerate}
\caption{Branch and bound scheme to solve \eqref{eqn:FAauxVars}. Here the input upper bounds $\uu^0$ are such that for any $i$ and any $\ph\in\F$, we have $\phi_i\leq u_i^0$. We discuss in Section \ref{ssec:params} how to choose such input.}\label{alg:bb}
\end{algorithm}

\subsubsection{Duality and Properties}

We now examine \eqref{eqn:bbLB}, the main subproblem of interest. Observe that \eqref{eqn:bbLB} is a linear SDO, and therefore we can consider its dual, namely

\begin{equation}\label{eqn:dual}
\begin{array}{ll}
\ds\maxi_{\substack{q,\f_u,\f_\ell,\bk,\\\bsig,\mm,\nn,\pp}}&(p-r)q -\uu'\f_u - \Tr(\nn)-\langle \pp,\s-\diag(\bell)\rangle\\
\st&\bk+\bsig = \mb 1\\
&\diag(\pp) +\f_u-\f_\ell-\bk = \mb 0\\
&\s-\diag(\uu)+\mm+\nn+\diag(\diag(\uu-\bell)\bk)-q\mb I=\mb 0\\
& \f_u,\f_\ell,\bk,\bsig\geq\mb0\\
&\mm,\nn,\pp\psd.
\end{array}
\end{equation}

\begin{observation}\label{obs:PDprops}
We now include some remarks about structural properties of \eqref{eqn:bbLB} and its dual \eqref{eqn:dual}.
\begin{enumerate}

\item If $\rank(\s)=p$ then the Slater condition \citep{BV2004} holds and hence there is strong duality, so we can work with \eqref{eqn:dual} instead of \eqref{eqn:bbLB} as an exact reformulation.

\item There exists an optimal solution to the dual with $\f_u = \mb0$. This is a variable reduction which is not immediately obvious. Note that $\f_u$ appears as the multiplier for the constraints in the primal of the form
$$\ph\leq \mb u.$$
To claim that we can set $\mb \f_u=\mb0$ it suffices to show that the corresponding constraints in the primal can be ignored. Namely, if $(\W^*,\Ph^*)$ solves \eqref{eqn:bbLB} with the constraints $\phi_i\leq u_i\;\forall i$ omitted, then the pair $(\W^*,\widetilde{\Ph})$ is feasible and optimal to \eqref{eqn:bbLB} with all the constraints included, where $\widetilde{\Ph}$ is defined by
$$\widetilde{\phi}_i = \min\{\phi_i^*,u_i\}.$$
One can show by example that the constraints $\ph\geq\bell$ cannot be omitted, and therefore this behavior is particular to the upper bounds on $\ph$. Hereinafter we set $\f_u=\mb0$ and omit the constraint $\ph\leq\uu$ (with the minor caveat that, upon solving a problem and identifying some $\ph^*$, we must instead work with $\min\{\ph^*,\uu\}$, taken entrywise).

\end{enumerate}
\end{observation}

\subsection{Input Parameters} 
\label{ssec:params}

In solving the root node of Algorithm \ref{alg:bb}, we must begin with some choice of $\uu^0=\uu$. An obvious first choice for $u_i$ is $u_i = \Sigma_{ii}$, but one can do better. Let us optimally set $u_i$, defining it as
\begin{equation}\label{eqn:udefinition}
u_i:=\begin{array}{lc}
\ds\maxi_{x\in\R}&x\\
\text{s.t.}&\s-x\mb e_i^{}\mb e_i'\psd.
\end{array}
\end{equation}
Here $\{\mb e_i\}_{i=1}^p$ is the canonical basis for $\R^p$. These bounds are useful because if $\ph\in\F$, i.e., $\ph\geq\mb0$ with $\s-\ph\psd$, then $\phi_i\leq u_i$. Note that the problem in \eqref{eqn:udefinition} is a linear SDO for which strong duality holds. Its dual is precisely
\begin{equation}\label{eqn:udefinitionDual}
u_i=\begin{array}{ll}
\ds\mini_{\mb M\in\R^{p\times p}}& \langle \mb M,\s\rangle\\
\st& M_{ii}=1\\
&\mb M\psd,
\end{array}
\end{equation}
a linear SDO in standard form with a single equality constraint. Results on the rank structure of solutions to semidefinite optimization problems \citep{barvinok,pataki} imply that there exists a rank $\hat r$ solution to \eqref{eqn:udefinitionDual}, where $\hat r$ is the largest integer so that $\hat r(\hat r+1)\leq 2\kappa$, where $\kappa$ is the number of equality constraints. Here $\kappa=1$, and so $\hat r=1$. Therefore, there exists a rank one solution to \eqref{eqn:udefinitionDual}. This implies that \eqref{eqn:udefinitionDual} can actually be solved as a convex quadratic program:
$$\begin{array}{ll}
\ds\mini_{\mb M}& \langle \mb M,\s\rangle\\
\st& M_{ii}=1\\
&\mb M\psd,
\end{array}= \begin{array}{ll}
\ds\mini_{\mb m\in\R^p}& \mb m'\s\mb m\\
\st& m_{i}^2=1.
\end{array}$$
This can be written exactly as a convex quadratic problem with only a single linear equality constraint:
\begin{equation}\label{eqn:uQPform}
u_i=\begin{array}{ll}
\ds\mini_{\mb m}& \mb m'\s\mb m\\
\st& m_{i}=1.
\end{array}
\end{equation}
This formulation given in \eqref{eqn:uQPform} is computationally inexpensive to solve (given a large number of specialized convex quadratic problem solvers), in contrast to both formulations \eqref{eqn:udefinition} and \eqref{eqn:udefinitionDual}.\footnote{In the case when $\s\succ\mb0$, it is not necessary to use quadratic optimization problems to compute $\mb u$. In this case once can apply a straightforward Schur complement argument \citep{BV2004} to show that $\mb u$ can be computed by solving for the inverses of $p$ different $(p-1)\times(p-1)$ matrices.}

\subsection{Solving Subproblems}\label{ssec:subSolve}

We now detail how to solve a problem of the form \eqref{eqn:bbLB}. Given the large variety of solution techniques for SDOs \citep{sdpSurvey,monteiroSurvey}, we consider two primary options here: classical interior point methods and first order methods.

\subsubsection{Interior Point Methods}

We first consider solving \eqref{eqn:bbLB} via interior point methods (IPMs) \citep{BV2004}. There are a large variety of implementations of IPMs for SDOs, from academic code such as \texttt{SDPT3} \citep{sdpt3a} and \texttt{Yalmip} \citep{yalmip} to recently developed commercial code in \texttt{MOSEK} \citep{mosek}. For the problem sizes of interest in this paper, IPMs perform reasonably in terms of accuracy and speed; however, for problems on the order of $p=100$, IPMs as used to solve \eqref{eqn:bbLB} can already be prohibitively slow (especially when it is necessary to solve many times). The major drawback of interior point methods is that they do not accommodate warm starts, which are crucial for us as we are repeatedly solving similarly structured SDOs. Indeed, warm starts are well-known to perform quite poorly when incorporated into IPMs -- often they can perform worse than cold starts \citep{ws1,ws2}. For this reason in what follows we also consider another possible methods for solving subproblems of the form \eqref{eqn:bbLB}.

\subsubsection{First-Order Methods}

In stark contrast to IPMs, first-order methods (FOMs) have at their core a focus on speed and the ability to reliably leverage warm starts. The price one pays in using FOMs is generally a loss in accuracy relative to IPMs. Because speed and the need for warmstarts are not only desirable but necessary for our purposes, we choose to use a FOM approach to solving \eqref{eqn:bbLB}.

Observe that we cannot solve the primal form of \eqref{eqn:bbLB} within the branch-and-bound framework unless we solve it to optimality. Therefore, we instead choose to work with its dual as in \eqref{eqn:dual}. In this way, we apply FOMs to solve \eqref{eqn:dual} and find reasonably accurate, \emph{feasible} solutions for this dual problem, which guarantees that we have a lower bound to \eqref{eqn:bbLB}. In this way, we maintain the provable optimality properties of Algorithm \ref{alg:bb} without needing to solve nodes in the branch-and-bound tree exactly to optimality via IPMs.

The FOM we apply is an off-the-shelf solver \texttt{SCS} \citep{scs} which is based on ADMM \citep{bertsekasLag}. In our experience, this method performs approximately two orders of magnitude faster than IPMs for solving \eqref{eqn:bbLB}, with very little loss in accuracy. At the same time, it accommodates warm starts very well and these make a substantive difference in solve time as compared to using first-order methods without any warm starts. Further, \texttt{SCS} gives an approximately feasible primal solution $\ph$, which is useful as we consider branching, described next.\footnote{Note that in Algorithm \ref{alg:bb}, we can replace the best incumbent solution if we find a new $\ph$ which has better objective value $\sum_{i>r} \sigma_i(\s-\ph)$. Because $\ph$ may not be feasible ($\ph\in\F$), we take care here. Namely, compute $t=\sum_{i>r} \lambda_i(\s-\ph)$, where $\lambda_1(\s-\ph)\geq\cdots\geq\lambda_p(\s-\ph)$ are the sorted eigenvalues of $\s-\ph$. If $t < z_f$, then we perform an iteration of CG scheme for finding feasible solutions (outlined in Section \ref{sec:CG-method1}) to find a feasible $\bar\ph\in\F$. We then use this as a possible candidate for replacing the incumbent.}

\subsection{Branching}\label{ssec:branching}

Here we detail two methods for branching. The rule we select for our computational experiments is simple and straightforward. The problem of branching is as follows: having solved \eqref{eqn:bbLB} for some particular $\n=[\bell,\uu]$, we must choose some $i\in\{1,\ldots,p\}$ and split the interval $[\ell_i,u_i]$ to create two new subproblems. We begin with a simple branching rule. Given a solution $(\W^*,\Ph^*,\zz^*)$ to \eqref{eqn:bbLB}, compute
$$i\in\ds\operatorname{argmax}_i \left|z_i^* - W_{ii}^*\phi_i^*\right|$$
and branch on variable $\phi_i$, generating two new subproblems with the intervals
$$\prod_{j<i} [\ell_j,u_j]\times [\ell_i,\phi_i^*]\times \prod_{j>i}[\ell_j,u_j]\text{\quad and \quad}\prod_{j<i} [\ell_j,u_j]\times [\phi_i^*,u_i]\times \prod_{j>i}[\ell_j,u_j].$$
Observe that, so long as $\max_i |z_i^*-W_{ii}^*\phi_i^*|>0$, the solution $(\W^*,\Ph^*,\zz^*)$ is not optimal for either of the subproblems created. 

We now briefly describe an alternative rule which we employ instead of the simple rule. We again pick the branching index $i$ as 
$$i\in\ds\operatorname{argmax}_i \left|z_i^* - W_{ii}^*\phi_i^*\right|,$$
but now the two new nodes we generate are
$$\prod_{j<i} [\ell_j,u_j]\times [\ell_i,(1-\epsilon)\phi_i^*+\epsilon\ell_i]\times \prod_{j>i}[\ell_j,u_j]\text{\quad and \quad}\prod_{j<i} [\ell_j,u_j]\times [(1-\epsilon)\phi_i^*+\epsilon\ell_i,u_i]\times \prod_{j>i}[\ell_j,u_j],$$
where $\epsilon\in[0,1)$ is some parameter. For the computational results (Section \ref{sec:comp}), we set $\epsilon=0.4$.

Such an approach, which lowers the location of the branch in the $i$th interval $[\ell_i,u_i]$ from $\phi_i^*$, serves to improve the objective value from the first node, while hurting the objective value from the second node (here by objective value, we mean the objective value of the optimal solution to the two new subproblems). In this way, it spreads out the distance between the two, and so it is more likely that the first node may have an objective value that is higher than $z_f-\tol$ than before, and hence, this would mean there are fewer nodes necessary to consider to solve for an additive gap of \tol. While this heuristic explanation is only partially satisfying, we have observed throughout a variety of numerical experiments that this rule, even though simple, performs better across a variety of example classes than the basic branching rule outlined. Recent work on the theory of branching rules supports such a heuristic rule \citep{lebodic}. In Section \ref{sec:comp}, we give evidence on the impact of the use of the modified branching rule.

\subsection{Weyl's Method---Pruning and Bound Tightening}\label{ssec:weyl}

In this subsection, we develop another method for lower bounds for the factor analysis problem. While we use it to supplement our approach detailed throughout Section \ref{sec:alg}, it is of interest as a standalone method, particularly for its computational speed and simplicity. In Section \ref{sec:comp}, we discuss the performance of this approach in both contexts.

The central problem of interest in factor analysis involves the spectrum of a symmetric matrix ($\s-\Ph$), which is the difference of two other symmetric matrices ($\s$ and $\Ph$). The spectrum of the sum of real symmetric matrices is an extensively studied problem \citep{hornjohnson,bhatia}.

Therefore, it is natural to inquire how such results carry over to our setting. We discuss the implications of some well-known results from this literature, primarily \emph{Weyl's inequality}. Let us begin by recalling the result.

\begin{theorem}[Weyl's inequality, \citealp{hornjohnson}]
For symmetric matrices $\mb A,\mb B\in\R^{p\times p}$ with sorted eigenvalues
\[\lambda_1(\mb A)\geq \lambda_2(\mb A)\geq\cdots \geq\lambda_p(\mb A)\text{ \quad and \quad}\lambda_1(\mb B)\geq \lambda_2(\mb B)\geq\cdots \geq\lambda_p(\mb B)$$
one has for any $k\in\{1,\ldots,p\}$ that
$$\lambda_k(\mb A+\mb B) \geq \lambda_{k+j}(\mb A) +\lambda_{p-j}(\mb B)\;\forall j=0,\ldots,p-k.
\]
\end{theorem}

Now observe that there does not exist a $\T$ with $\rank(\T)\leq r$, $\T = \s-\Ph\psd$, and $\Ph=\diag(\phi_1,\ldots,\phi_p)\psd$ if and only if $z^*>0$, where
$$\begin{array}{lll}
z^* :=&\mini& \sigma_{r+1}(\s-\ph)\\
&\st& \ph\in\F\\
&&\ph\geq\mb0.
\end{array}$$
For any vector $\mb x\in\R^p$ we let $\{x_{(i)}\}_{i=1}^p$ denote sorted $\{x_i\}_{i=1}^p$ with
$$x_{(1)}\geq x_{(2)}\geq \cdots\geq x_{(p)}.$$
Using this notation, we arrive at the following theorem.

\begin{theorem}\label{thm:weylLBImpl}
For any  $\bar\ph\in\R^p$ one has that
$$\ds\mini_{\ph\in\F} \sum_{i>r}(\s-\ph)\geq \mini_{\ph\in\F} \ds\left(\sum_{i>r}\max_{j=0,\ldots,p-i}\left\{\lambda_{i+j}(\s-\bar\ph)+ \left(\bar\ph-\ph\right)_{(p-j)},0\right\}\right).$$
\end{theorem}
\begin{proof}
Apply Weyl's inequality with $\mb A =\s-\diag(\bar\ph)$ and $\mb B = \diag(\bar\ph-\ph)$, and use the fact that $\ph\in\F$ so $\s-\diag(\ph)\psd$.
\end{proof}

\subsubsection{Fast Method}

The new lower bound introduced in Theorem \ref{thm:weylLBImpl} is a non-convex problem in general. We begin by discussing one situation in which Theorem \ref{thm:weylLBImpl} provides computationally tractable (and fast) lower bounds;  we deem this \emph{Weyl's method}, detailed as follows:
\begin{enumerate}
\item Compute bounds $u_i$ as in \eqref{eqn:udefinition}, so that for all $\ph\geq\mb0$ with $\s-\ph\psd$, one has
$$\phi_i \leq u_i\;\forall i.$$
\item For each $r\in\{1,\ldots,p\}$, one can compute the lower bound to \eqref{eqn:FAequiv} with a given $r$ of
$$\sum_{i>r} \max\{\lambda_i(\s-\uu),0\},$$
by taking Theorem \ref{thm:weylLBImpl} with $\bar\ph =\uu$.
\end{enumerate}

As per the above remarks, computing $\uu$ in Step 1 of Weyl's method can be carried out efficiently. Step 2 only relies on computing the eigenvalues of $\s-\diag(\uu)$. Therefore, this lower bounding procedure is quite fast and simple to carry out. What is perhaps surprising is that this simple lower bounding procedure is effective as a fast, standalone method. We describe such results in Section \ref{sec:comp}. We now turn our attention to how Weyl's method can be utilized within the branch and bound tree as described in Algorithm \ref{alg:bb}.

\subsubsection{Pruning}

We begin by considering how Weyl's method can be used for \emph{pruning}. The notion of pruning for branch and bound trees is grounded in the theory and practice of discrete optimization. In short, pruning in the elimination of nodes from the tree without actually solving them. We make this precise in our context.

Consider some point in the branch and bound process in Algorithm \ref{alg:bb}, where we have some collection of nodes,
$$\left([\bell^c,\uu^c],z^c\right)\in\nodes.$$
Recall that $z^c$ is the objective value of the parent node of $\n=[\bell^c,\uu^c]$. As per Weyl's method, we know \emph{a priori}, without solving the node subproblem
\begin{equation*}
\begin{array}{llc}
z_\n:=&\ds\mini_{\W,\Ph,\zz}&\langle \W,\s\rangle-\ds\sum_iz_i\\
&\text{s.t.}&\ph\in\F\\
&&\bell^c\leq\ph\leq\uu^c\\
&&\W\in\cW\\
&&z_i \leq \phi_i+\ell_iW_{ii}-\ell_i\;\forall i\\
&&z_i\leq u_iW_{ii}\;\forall i,
\end{array}
\end{equation*}
that
$$\mini_{\substack{\ph\in\F\\\bell^c\leq\ph\leq\uu^c}} \sum_{i>r} \sigma_i(\s-\ph) \geq \sum_{i>r} \max\{\lambda_i(\s-\uu^c),0\}.$$
Hence, if $z_f -\tol< \sum_{i>r} \max\{\lambda_i(\s-\uu^c),0\}$, where $z_f$ is as in Algorithm \ref{alg:bb}, then node $\n$ can be discarded, i.e., there is no need to actually compute $z_\n$ or further consider this branch. This is because if we were to solve $z_\n$, and then branch again, solving further down this branch to optimality, then the final lower bound obtained would necessarily be at least as large as the best feasible objective already found (with tolerance \tol).

In this way, because Weyl's method is fast, this provides a simple method for pruning. In the computational results detailed in Section \ref{sec:comp}, we always use Weyl's method to discard nodes which are not fruitful to consider.

\subsubsection{Bound Tightening}

We now turn our attention to another way in which Weyl's method can be used to improve the performance of Algorithm \ref{alg:bb} -- \emph{bound tightening}. In short, bound tightening is the use of implicit constraints to strengthen bounds obtained for a given node. We detail this with the same node notation as pruning. Namely, consider a given node $\n=[\bell^c,\uu^c]$. Fix some $j\in\{1,\ldots,p\}$ and let $\alpha\in(\ell_j^c,u_j^c)$. If we have that
$$z_f-\tol < \sum_{i>r} \max\{\lambda_i(\s-\tilde\uu),0\},$$
where $\tilde\uu$ is $\uu^c$ with the $j$th entry replaced by $\alpha$, then we can replace the node $\n$ with the ``tightened'' node
$$\tilde\n = [\tilde\bell,\uu^c],$$
where $\tilde\bell$ is $\bell^c$ with the $j$th entry replaced by $\alpha$.

We consider why this is valid. Suppose that one were to compute $z_\n$ and choose to branch on index $j$ at $\alpha$. Then one would create two new nodes:
$$[\bell^c,\tilde \uu]\text{ and } [\tilde\bell,\uu^c].$$
We would necessarily then prune away the node $[\bell^c,\tilde \uu]$ as just described; hence, we can replace $[\bell^c,\uu^c]$ without loss of generality with $[\tilde\bell,\uu^c]$. Note that here $\alpha\in(\ell_j^c,u_j^c)$ and $j\in\{1,\ldots,p\}$ were arbitrary. Hence, for each $j$, one can choose the largest such $\alpha_j\in(\ell_j^c,u_j^c)$ (if one exists) so that $z_f-\tol < \sum_{i>k} \max\{\lambda_i(\s-\tilde\uu),0\}$, and then replace $\bell^c$ by $ \bs\alpha$.\footnote{For completeness, let us note how one would find such an $\alpha$. An obvious choice is a grid-search-based bisection method. For simplicity we use a linear search on a grid instead of resorting to the bisection method.}

Such a procedure is somewhat expensive (because of its use of repeated eigenvalue calculations), but can be thought of as ``optimal'' pruning via Weyl's method. In our experience the benefit of bound tightening does not warrant its computational cost when used at every node in the branch-and-bound tree except in a small number of problems. For this reason, in the computational results in Section \ref{sec:comp} we only employ bound tightening at the root node $\n=[\mb 0,\uu]$.

\subsection{Node Selection}\label{ssec:nodeSelect}

In this section, we briefly describe our method of node selection. The problem of node selection has been considered extensively in discrete optimization and is still an active area of research. Here we describe a simple node selection strategy.

To be precise, consider some point in Algorithm \ref{alg:bb} were we have a certain collection of nodes,
$$\left([\bell^c,\uu^c],z^c\right)\in\nodes.$$
The most obvious node selection strategy is to pick the node $\n=([\bell^c,\uu^c],z^c)$ for which $z^c$ is smallest among all nodes in \nodes. In this way, the algorithm is likely to improve the gap $z_f-z_\text{lb}$ at every iteration. Such greedy selection strategies tend to not perform particularly well in general global optimization problems.

For these reasons, we employ a slightly modified greedy selection strategy which utilizes Weyl's method. For a given node $\n$, we also consider its corresponding lower bound $w^c$ obtained from Weyl's method, namely,
$$w^c := \sum_{i>r} \max\{\lambda_i(\s-\uu^c),0\}.$$
For each node, we now consider $\max\{z^c,w^c\}$. There are two cases to consider:
\begin{enumerate}

\item With probability $\beta$, we select the node with smallest value of $\max\{z^c,w^c\}$.

\item In the remaining scenarios (occurring with probability $1-\beta$), we choose randomly between selecting the node with smallest value of $z^c$ and the node with smallest value of $w^c$. To be precise, let $Z$ be the minimum of $z^c$ over all nodes and likewise let $W$ be the minimum of $w^c$ over all nodes. Then with (independent) probability $\beta$, we choose the node with worst $z^c$ \emph{or} $w^c$ (i.e., with $\min\{z^c,w^c\} = \min\{Z,W\}$); with probability $1-\beta$, if $Z<W$ we choose a node with $w^c = W$, and if $Z>W$ we choose a node with $z^c=Z$.

\end{enumerate}

In this way, we allow for the algorithm to switch between trying to make progress towards improving the convex envelope bounds and making progress towards improving the best of the two bounds (the convex envelope bounds along with the Weyl bounds). We set $\beta=0.9$ for all computational experiments. It is possible that a more dynamic branching strategy could perform substantially better; however, the method here has a desirable level of simplicity.\footnote{While the node selection strategy we consider here appears na\"ive, it is not necessarily so simple. Improved node selection strategies from discrete optimization often take into account some sort of duality theory. Weyl's inequality is at its core a result from duality theory (principally Wielandt's minimax principle, \citealp{bhatia}), and therefore our strategy, while simple, does in some way incorporate a technique which is a level of sophistication beyond worst-bound greedy node selection.}

\subsection{Global Optimization: State of the Art}\label{ssec:globOptSOA}

We close this section by discussing similarities between the branch-and-bound approach we develop here and existing methods in non-convex optimization. Our approach is very similar in spirit to approaches to global optimization \citep{floudas}, and in particular for (non-convex) quadratic optimization problems, quadratically-constrained convex optimization problems, and bilinear optimization problems \citep{hansen,bst,tsExt,trx, costaliberti, anstreicherburer,misenerglobal}. The primary similarity is that we work within a branch and bound framework using successively better convex lower bounds. However, while global optimization software for a variety of non-convex problems with underlying vector variables are generally well-developed (as evidenced by both academic software and commercial implementations like \texttt{BARON}, \citealp{baron}), this is not the case for problems with underlying matrix variables and semidefinite constraints.

The presence of semidefinite structure presents several substantial computational challenges. First and foremost, algorithmic implementations for solving linear SDOs are not nearly as advanced as those which exist for linear optimization problems. Therefore, each subproblem, which is itself a linear SDO, carries a larger computational cost than the usual corresponding linear program which typically arises in other global optimization problems with vector variables. Secondly, a critical component of the success of global optimization software is the ability to quickly resolve multiple instances of subproblems which have similar structure (for example, the technique of dual simplex is often relevant here). Corresponding methods for SDOs, as solved via interior point methods, are generally not well-developed. Finally, semidefinite structure complicates the traditional process of computing convex envelopes. Such computations are critical to the success of modern global optimization solvers like \texttt{BARON}.

There are a variety of other approaches to computing lower bounds to \eqref{eqn:FAequiv}. One such approach is the method of moments \citep{momentsBook}. However, for problems of the size we are considering, such an approach is likely not computationally feasible, so we do not make a direct comparison here. There is also recent work in complementarity constraints literature \citep{mitchellcomp} which connects rank-constrained optimization problems to copositive optimization \citep{cop1}. In short, such an approach turns \eqref{eqn:FAequiv} into an equivalent convex problem; despite the transformation, the new problem is not particularly amenable to computation. For this reason, we do not consider the copositive optimization approach here.

\section{Computational experiments}\label{sec:compute}

In this section, we perform various computational experiments to study the  
properties of our different algorithmic proposals for ($\textsc{CFA}_{q}$), for $q\in \{1,2\}$. Using a variety of statistical measures, we compare our
methods with existing popular approaches for FA, as
implemented in standard \texttt{R} 
statistical packages {\texttt{psych}}~\citep{psych}, {\texttt{nFactors}}~\citep{r-nfactors-13}, and {\texttt{GPArotation}}~\citep{gparotation-r}. We then turn our attention to certificates of optimality as described in Section \ref{sec:alg} for ($\textsc{CFA}_{1}$).

\subsection{Synthetic examples}\label{sec:synthetic-eg1}
For our synthetic experiments, we considered distinctly different groups of examples. Classes $A_1$ and $A_2$ have subspaces of the low-rank common factors which are random and the values
 of $\Phi_{i}$ are taken to be equally spaced.  The underlying matrix corresponding to the common factors in type $A_1$ is exactly low-rank, while this is not the case in type $A_2$.

\noindent {\bf {Class $A_1(R/p)$.}}
We generated a matrix $\M{L}_{p \times R}:=((L_{ij}))$ (with $R < p$) with $L_{ij}\stackrel{\text{iid}}{\sim} N(0,1)$. 
The \emph{unique variances} $\Phi_1, \ldots, \Phi_{p}$,
are taken to be proportional to $p$ equi-spaced values on the interval 
$[\lambda_{R}(\M{L}'\M{L}),\lambda_{1}(\M{L}'\M{L})]$ such that
$\Phi_{i} = \bar{\phi} \cdot \left( \lambda_{1}(\M{L}'\M{L})  + ( \lambda_{R}(\M{L}'\M{L})  - \lambda_{1}(\M{L}'\M{L}) )\frac{i-1}{p} \right),$ for $1 \leq i \leq p.$
$\bar{\phi}$, which controls the ratio of the variances between the \emph{uniquenesses} and the common  latent  factors 
 is chosen such that 
 $\sum_{i=1}^{p} \Phi_{i} = \tr (\M{L} \M{L}')$, i.e.,  the contribution to the total variance from the common factors matches 
 that from
the uniqueness factors. 
 The covariance matrix is thus given by:
 $\B\Sigma = \M{L}\M{L}' + \B\Phi$.  

\medskip 

\noindent {\bf {Class $A_2(p)$.}}
Here $\mb L_{p\times p}$ is generated as $L_{ij}\stackrel{\text{iid}}{\sim} N(0,1)$.
We did a full singular value decomposition on $\M{L}$ --- let $\M{U}_{L}$ denote the set of 
$p$ (left) singular vectors. We created a positive definite matrix with exponentially decaying eigenvalues
as follows:
$\widetilde{\M{L}}\widetilde{\M{L}}' = \M{U}_{L} \diag(\lambda_{1}, \ldots, \lambda_{p})\M{U}_{L}',$ where the eigenvalues were chosen as 
$\lambda_{i} = 0.8^{i/2}, i = 1, \ldots, p.$
We chose the diagonal entries of $\B\Phi$  (like data Type-$A_1$), as a scalar multiple ($\bar{\phi}$) 
of a uniformly spaced grid in
$[\lambda_{p}, \lambda_{1}]$ and $\bar{\phi}$  was chosen such that
$\sum_{i} \Phi_{i} = \tr (\widetilde{\M{L}}\widetilde{ \M{L}}')$. \\

In contrast, classes $B_1$, $B_2$, and $B_3$ are qualitatively different from the aforementioned ones---the subspaces corresponding to the
common factors are more structured, and hence different from the coherence-like assumptions on the eigenvectors which are  necessary for nuclear norm based methods~\citep{venkat-2012-factor} to work well.

{\noindent \bf {Class $B_1(R/p)$.} }
We set $\bs\Theta = \mb L \mb L'$, where $\mb L_{p\times R}$ is given by
$$L_{ij} = \left\{\begin{array}{ll}
1,& i\leq j\\
0,& i>j.
\end{array}\right.$$

\medskip

{\noindent \bf {Class $B_2(r/R/p)$.} }
Here we set $\T = \LL\LL'$, where $\LL_{p\times R}$ is such that
$$L_{ij} = \left\{\begin{array}{ll}
1,& i,j = 1,\ldots,r\\
\sim N(0,1), & i>r, j=1,\ldots,R\\
0,& i=1,\ldots,r,\; j>r.
\end{array}\right.$$
\medskip

{\noindent \bf {Class $B_3(r/R/p)$.} }
Here we define $\T = \LL\LL'$, where $\LL_{p\times R}$ is such that
$$L_{ij} = \left\{\begin{array}{ll}
1,& j=1,\ldots,r,\;i\leq j\\
\sim N(0,1),& j > r,\; i =1,\ldots,R\\
0,& i>j, \;j=1,\ldots,r.
\end{array}\right.$$

In all the $B$ classes, we generated 
$\Phi_{ii}=\phi_i \sim \text{abs}(N(0,1))$ for $i = 1, \ldots, p$ and the covariance matrix $\B\Sigma$ was taken to be
$\B\Sigma = \T + \alpha\B\Phi$, where $\alpha$ is so that $\Tr(\T) = \alpha \Tr(\Ph)$. Across all examples, we work with the corresponding correlation matrix $\mb D\s\mb D$, where $\mb D$ is a diagonal matrix such that all diagonal entries of $\mb D\B\Sigma\mb D$ are equal to one. This allows for comparision across examples.

\subsubsection*{Comparisons with other FA methods}
We performed a suite of experiments using Algorithms 1 and 2  for the cases $q \in \{1, 2\}$.
 We compared our proposed algorithms  with the following popular FA estimation procedures as  described in Section~\ref{sec:related-work}:

1. \textsc{\small MINRES}: minimum residual factor analysis

2. \textsc{\small WLS}: weighted least squares method
with weights being the \emph{uniquenesses}

3. \textsc{\small  PA}: this is the principal axis factor analysis method 

4. \textsc{\small{MTFA}}: constrained minimum trace factor analysis---formulation~\eqref{orig-mtfa1} 

5. \textsc{\small PC}: The  method of principal component factor analysis

6. \textsc{\small MLE}: this is the maximum likelihood  estimator (MLE)

7. \textsc{\small GLS}: the generalized least squares method \\

For \textsc{\small MINRES},  \textsc{\small  WLS}, \textsc{\small  GLS}, 
and  \textsc{PA}, we used the implementations available in the \texttt{R} package {\texttt{psych}}~\citep{psych} available from \texttt{CRAN}. 
For \textsc{\small  MLE} we investigated the methods 
{\texttt {factanal}} from R-package {\texttt {stats}} and
the {\texttt fa} function from \texttt{R} package {\texttt{psych}}. 
The estimates obtained by the the MLE implementations were quite similar.

For \textsc{\small{MTFA}}, we used our own implementation by adapting the \textsc{ADMM} algorithm (Section~\ref{sec:subprob-phi1}) to solve Problem~\eqref{orig-mtfa1}.
For the experiments in Section~\ref{sec:synthetic-eg1}, we took the convergence thresholds for 
Algorithms~1 and 2 as $\mathrm{TOL}=10^{-5}$, and ADMM as $\mathrm{TOL}\times \alpha = 10^{-9}$. 
For the PC method we followed the description in~\cite{bai-ng-2008large-review} (as described in Section~\ref{sec:related-work})---the $\B\Phi$ estimates were thresholded at zero if they became negative.

Note that  all the methods considered in the experiments apart from~\textsc{\small{MTFA}},
allow the user to specify the desired number of factors in the problem.
Since standard implementations of \textsc{\small MINRES},  \textsc{\small  WLS} and  \textsc{\small  PA} 
require $\B\Sigma$ to be a correlation matrix, we standardized the covariance matrices to correlation matrices at the onset. 

\subsubsection{Performance measures}\label{sec:perf-meas1}
We consider the following 
measures of ``goodness of fit''~(See for example~\cite{Bartholomew_Knott_Moustaki_2011} and references therein)
to assess the performances of the different  FA estimation procedures.

\paragraph{Estimation error in $\B\Phi$:}
We use the following measure to assess the quality of an estimator for $\B\Phi$:
\begin{equation}\label{accu-phi-est1}
\text{Error}(\B\Phi) : = \sum_{i=1}^{p} (\widehat{\Phi}_{i} - \Phi_{i})^2.
\end{equation} 
The estimation of $\B\Phi$ plays an important role in FA---given a good estimate  $\widehat{\B\Phi}$, the $r$-common factors can be obtained by a
rank-r eigendecomposition on the residual covariance matrix $\B\Sigma - \widehat{\B\Phi}$.

\paragraph{Proportion of variance explained and semi-definiteness of $(\B\Sigma - \B\Phi)$:}
A fundamental objective in FA  
lies in understanding how well the $r$-common factors explain 
the residual covariance, i.e., $(\B\Sigma - \widehat{\B\Phi})$---a direct analogue of what is done in PCA, as explained in Section~\ref{sec:intro1}.
For a given $r$, the proportion of variance explained by the $r$ common factors is given by
\begin{equation}\label{pro-var}
\text{Explained Variance}= \sum_{i=1}^{r} \lambda_{i}(\widehat{\B\Theta})/\sum_{i=1}^{p} \lambda_{i}(\B\Sigma - \widehat{\B\Phi}).
\end{equation}
As $r$ increases, the explained variance increases to one. This trade-off between $r$ and ``Explained Variance", plays an 
important role in exploratory FA and in particular, the choice of $r$. 
For the expression~\eqref{pro-var} to be meaningful, it is desirable to have $\B\Sigma -  \widehat{\B\Phi} \succeq \M{0}$.
Note that our framework ($\textsc{CFA}_{q}$) and in particular, \textsc{\small{MTFA}} estimates $\B\Phi$
under a PSD constraint on $\B\Sigma - {\B\Phi}$. 
 However, as seen in our experiments  $(\widehat{\B\Phi},\widehat{\B\Theta})$ estimated by the remaining methods \textsc{\small MINRES}, \textsc{\small PA}, 
\textsc{\small WLS}, \textsc{\small GLS}, \textsc{\small MLE} and others 
 often violate the PSD condition on $\B\Sigma - \widehat{\B\Phi}$ for some choices of $r$, thereby rendering the interpretation of  ``Explained Variance'' troublesome.

For the \textsc{\small{MTFA}} method, with estimator $\widehat{\B\Theta}$, 
the measure~\eqref{pro-var} applies only for the value of $r = \rnk(\widehat{\B\Theta})$ and the explained variance is one.

For the methods we have included in our comparisons, we report the 
values of ``Explained Variance'' as delivered by the {\texttt R}-package implementations.

\paragraph{Proximity between $\widehat{\B\Theta}$ and $\B\Theta$:}
A relevant measure of the  proximity between
$\B\Theta$ and its estimate ($\widehat{\B\Theta}$) is given by
\begin{equation}\label{prox-thetar}
\text{Error}(\B\Theta) := \|\widehat{\B\Theta}- \B\Theta_{r}\|^2_2,
\end{equation}
where $\B\Theta_{r}$ is the best rank-$r$ approximation to $\B\Theta$ and can be 
viewed  as the natural ``oracle'' counterpart of $\widehat{\B\Theta}$.
Note that \textsc{\small{MTFA}} does not incorporate any constraint on 
$\rnk(\widehat{\B\Theta})$ in its formulation. Since the estimates obtained by this procedure 
satisfy $\widehat{\B\Theta} = {\B\Sigma} - \widehat{\B\Phi}$, $\rnk(\widehat{\B\Theta})$ may be quite different from $r$.

\paragraph{Discussion of experimental results.}
We next discuss our findings from the numerical experiments for the synthetic datasets. 

Table~\ref{tab:type1}  shows the performances of the various methods for different choices of $p,R,r$ for class $A_{1}$.
For the problems~$(\textsc{CFA}_{q}) , q \in \{1, 2\}$, we present the results of Algorithm 2. Results obtained by Algorithm~1 were similar.
In all the examples, with the exception of $\textsc{\small{MTFA}},$ we set the number of factors to be $r = (R-1)$. 
For $\textsc{\small{MTFA}},$
the rank of $\widehat{\B\Theta}$ was computed as the number of eigenvalues of $\widehat{\B\Theta}$ larger than 
$10^{-5}$. 
$\textsc{\small{MTFA}}$ and $(\textsc{CFA}_{q}), q \in \{1,2\}$ estimate 
$\B\Phi$ with zero error---significantly better than competing methods. 
$\textsc{\small{MTFA}}$ and $(\textsc{CFA}_{q}), q \in \{1,2\}$ result in estimates such that 
$\B\Sigma - \widehat{\B\Phi}$ is PSD, other methods however fail to do so---the discrepancy can often be quite large.
$\textsc{\small{MTFA}}$ performs poorly in terms of estimating $\B\Theta$ since the estimated $\B\Theta$ has rank different than $r$. 
In terms of the proportion of variance explained $(\textsc{CFA}_{q})$ performs significantly better than all other methods.
The notion of ``Explained Variance'' by $\textsc{\small{MTFA}}$ for $r = (R-1)$ is not applicable 
since the rank of the estimated $\B\Theta$ is larger than $r$.

\begin{table}[!htpb]
\centering
\scalebox{0.82}{\begin{tabular}{| l | c c c c c c c c | c |}\hline
Performance &    \multicolumn{8}{|c|}{Method Used} & Problem Size\\
 measure & $(\textsc{CFA}_{1})$ & $(\textsc{CFA}_{2})$ &$\textsc{\small{MTFA}}$ & \textsc{\small MINRES} & \textsc{\small WLS} & \textsc{\small PA}  & \textsc{\small PC} & \textsc{\small MLE} & $(R/p)$ \\ 
  \hline
Error($\Phi$) & 0.0 & 0.0 & 0.0 & 699.5 & 700.1 & 700.0 & 639.9 & 699.6 & 3/200  \\ 
  Explained Var. & 0.6889 & 0.6889 & - & 0.2898 & 0.2898 & 0.2898 & 0.2956 & 0.2898 & 3/200  \\ 
  $\lambda_{\min}(\B\Sigma - \widehat{\B\Phi})$ & 0.0 & 0.0 & 0.0 & -0.6086 & -0.6194 & -0.6195 & -0.6140 & -0.6034& 3/200  \\ 
   Error($\B\Theta$)  & 0.0 & 0.0  & 689.68 & 0.1023  & 0.1079  & 0.1146   & 0.7071  & 0.1111 & 3/200  \\ \hline
Error($\Phi$) & 0.0 & 0.0 & 0.0 & 197.18 & 197.01 & 196.98 & 139.60 & 197.18 & 5/200  \\ 
  Explained Var. & 0.8473 & 0.8473 & - & 0.3813 & 0.3813 & 0.3814 & 0.3925 & 0.3813 & 5/200 \\ 
  $\lambda_{\min}(\B\Sigma - \widehat{\B\Phi})$ & 0.0 & 0.0 & 0.0 & -0.4920 & -0.5071 & -0.5081 & -0.4983 & -0.4919& 5/200  \\ 
   Error($\B\Theta$)   & 0.0  & 0.0  & 190.26 & 0.0726  & 0.0802  & 0.0817  & 1.1638  & 0.0725  & 5/200  \\ \hline
Error($\Phi$) & 0.0 & 0.0 & 0.0 & 39.85 & 39.18 & 39.28 & 2.47 & 40.04 & 10/200 \\ 
  Explained Var. & 0.9371 & 0.9371 & - & 0.4449 & 0.4452 & 0.4452 & 0.4686 & 0.4449& 10/200 \\ 
  $\lambda_{\min}(\B\Sigma - \widehat{\B\Phi})$ & 0.0 & 0.0 & 0.0 & -0.2040 & -0.2289 & -0.2329 & -0.2060 & -0.2037& 10/200 \\ 
   Error($\B\Theta$)   & 0.0& 0.0& 35.94 & 0.0531 & 0.0399& 0.0419 & 2.47 & 0.0560& 10/200 \\ 
     \hline \hline
Error($\Phi$) & 0.0 & 0.0 & 0.0 & 3838 & 3859 & 3858 & 3715 & 3870 &2/500\\ 
  Explained Var. & 0.7103 & 0.7103 & - & 0.3036 & 0.3033 & 0.3033 & 0.3056 & 0.3031 &2/500\\ 
  $\lambda_{\min}(\B\Sigma - \widehat{\B\Phi})$ & 0.0 & 0.0 & 0.0 & -0.7014 & -0.7126 & -0.7125 & -0.7121 & -0.7168&2/500 \\ 
   Error($\B\Theta$)   & 0.0 & 0.0 & 3835.6 & 0.0726 & 0.0770 & 0.0766 & 0.5897 & 0.1134&2/500 \\ \hline
Error($\Phi$) & 0.0 & 0.0 & 0.0 & 1482.9 & 1479.1 & 1478.9 & 1312.9 & 1485.1 &5/500\\ 
  Explained Var. & 0.8311 & 0.8311 & - & 0.3753 & 0.3754 & 0.3754 & 0.3798 & 0.3752 &5/500\\ 
  $\lambda_{\min}(\B\Sigma - \widehat{\B\Phi})$ & 0.0 & 0.0 & 0.0 & -0.5332 & -0.5433 & -0.5445 & -0.5423 & -0.5379&5/500 \\ 
   Error($\B\Theta$)   & 0.0 & 0.0 & 1459.5 & 0.1176 & 0.0752 & 0.0746 & 1.1358 & 0.1301& 5/500 \\ \hline 
Error($\Phi$) & 0.0 & 0.0 & 0.0 & 301.94 & 301.08 & 301.04 & 160.29 & 302.1  & 10/500 \\ 
  Explained Var. & 0.9287 & 0.9287 & - & 0.4443 & 0.4444 & 0.4444 & 0.4538 & 0.4443 & 10/500 \\ 
  $\lambda_{\min}(\B\Sigma - \widehat{\B\Phi})$ & 0.0 & 0.0 & 0.0 & -0.3290 & -0.3280 & -0.3281 & -0.3210 & -0.3302 & 10/500 \\ 
   Error($\B\Theta$)   & 0.0 & 0.0 & 291.44 & 0.0508 & 0.0415 & 0.0412 & 2.3945 & 0.0516& 10/500 \\ 
     \hline \hline
Error($\Phi$) & 0.0 & 0.0 & 0.0 & 19123 & 19088 & 19090 & 18770 & 19120 & 2/1000 \\ 
  Explained Var. & 0.6767 & 0.6767 & - & 0.2885 & 0.2886 & 0.2886 & 0.2898 & 0.2885& 2/1000 \\ 
  $\lambda_{\min}(\B\Sigma - \widehat{\B\Phi})$ & 0.0 & 0.0 & 0.0 & -0.6925 & -0.7335 & -0.7418 & -0.7422 & -0.6780& 2/1000 \\ 
   Error($\B\Theta$)   & 0.0 & 0.0 & 19037 & 50.76 & 0.3571 & 0.0984 & 0.8092 & 18.5930& 2/1000 \\ 
\hline
Error($\Phi$) & 0.0 & 0.0 & 0.0 & 6872 & 6862 & 6861 & 6497 & 6876 & 5/1000\\ 
  Explained Var. & 0.8183 & 0.8183 & - & 0.3716 & 0.3717 & 0.3717 & 0.3739 & 0.3716 & 5/1000\\ 
  $\lambda_{\min}(\B\Sigma - \widehat{\B\Phi})$ & 0.0 & 0.0 & 0.0 & -0.6202 & -0.6785 & -0.6797 & -0.6788 & -0.6209& 5/1000 \\ 
   Error($\B\Theta$)   & 0.0 & 0.0 & 6818.1 & 0.1706 & 0.0776 & 0.0780 & 1.1377 & 0.1696 & 5/1000\\ 
\hline
Error($\Phi$) & 0.0 & 0.0 & 0.0 & 1682 & 1681 & 1681 & 1311 & 1682& 10/1000 \\ 
  Explained Var. & 0.9147 & 0.9147 &- & 0.4360 & 0.4360 & 0.4360 & 0.4408 & 0.4360& 10/1000 \\ 
  $\lambda_{\min}(\B\Sigma - \widehat{\B\Phi})$ & 0.0 & 0.0 & 0.0 & -0.2643 & -0.2675 & -0.2676 & -0.2631 & -0.2643& 10/1000 \\ 
   Error($\B\Theta$)   & 0.0 & 0.0 & 1654.3 & 0.0665 & 0.0568 & 0.0570 & 2.4202 & 0.0665 & 10/1000\\ 
  \hline
\end{tabular}}
\caption{\small{Comparative performances of the various FA methods for data of type $A_1$, for different choices of 
$R$ and $p$. 
In all the above methods (apart from \textsc{\small{MTFA}}) $r$ was taken to be $(R -1)$.  In all the cases $\rnk(\widehat{\B\Theta})$
obtained by \textsc{\small{MTFA}} is seen to be $R$. 
The ``-'' symbol implies that the notion of explained variance is not meaningful for \textsc{\small{MTFA}} for $r=R-1$. 
No method in Category (B) satisfies $\B\Sigma - \B\Phi \succeq \M{0}$.
Methods proposed in this paper seem to significantly outperform their competitors, across the different performance measures. }}\label{tab:type1}
\end{table}

Figure~\ref{fig:one} displays results for type $A_2$.
Here we present the results for $(\textsc{CFA}_{q}), q \in \{1, 2\}$ using Algorithm 2 (results obtained by Algorithm~1 were very similar). 
For all the methods (with the exception of \textsc{\small{MTFA}}) we 
computed estimates of $\B\Theta$ and $\B\Phi$ for
a range of values of $r$. 
 \textsc{\small{MTFA}} and  $(\textsc{CFA}_{1})$ 
do a perfect job in estimating $\B\Phi$ and 
both deliver PSD matrices $\B\Sigma - \widehat{\B\Phi}$. 
\textsc{\small{MTFA}} computes solutions ($\widehat{\B\Theta}$) with 
a higher numerical rank and with large errors in estimating $\B\Theta$ (for smaller values of $r$).
Among the four performance measures corresponding to \textsc{\small{MTFA}}, $\text{Error}(\B\Theta)$ 
is the only one that varies with different $r$ values. Each of the other three measures deliver a single value  corresponding to $r=\rnk(\widehat{\B\Theta})$.
Overall, it appears that $(\textsc{CFA}_{q})$ is significantly better than all other methods.

Figure~\ref{fig:rev-types478} shows the results for classes $B_{1}, B_{2}, B_{3}$. 
We present the results for $(\textsc{CFA}_{q})$ for  $q \in \{1, 2\}$ using Algorithm 2, as before. 
Figure~\ref{fig:rev-types478} shows the performance of the different methods in terms of four different metrics: 
error in $\B\Phi$ estimation, proportion of variance explained, violation of the PSD constraint on $\B\Sigma - \B\Phi$ and error in 
$\B\Theta$ estimation. For the case of $B_{1}$  we see that the proportion of explained variance for $(\textsc{CFA}_{q})$ 
reaches one at a rank smaller than that of $\textsc{\small{MTFA}}$---this shows that the non-convex criterion $(\textsc{CFA}_{q})$
provides smaller estimates of the rank than its convex relaxation $\textsc{\small{MTFA}}$ when one seeks a model that explains the 
full proportion of residual variance. This result is qualitatively different from 
the behavior seen for $A_{1}, A_{2}$, where the benefit of $(\textsc{CFA}_{q})$ over 
$\textsc{\small{MTFA}}$ was mainly due to its flexibility to control the rank of $\B\Theta$.
Algorithms in Category (A) do an excellent job in estimating $\B\Phi$. All other 
competing methods perform poorly in estimating $\B\Phi$ for small/moderate values of $r$.
We observe that none of the methods apart from $(\textsc{CFA}_{q})$ and $\textsc{\small{MTFA}}$ lead to 
PSD estimates of $\B\Sigma - \widehat{\B\Phi}$ (unless $r$ becomes sufficiently large which corresponds to a model with a saturated fit).
In terms of 
the proportion of variance explained, our proposal performs much better than the 
competing methods. We see that the error in $\B\Theta$ estimation incurred by $(\textsc{CFA}_{q})$, increases
marginally as soon as the rank $r$ becomes larger than a certain value for $B_{1}$. 
Note that around the same values of $r$, the proportion of explained variance reaches one in both these cases, thereby suggesting
that this is possibly not a region of statistical interest.  In summary, Figure~\ref{fig:rev-types478} suggests that 
$(\textsc{CFA}_{q})$ performs very well compared to all its competitors.

\begin{figure}[h!]
\centering
\resizebox{1.05\textwidth}{0.33\textheight}{\begin{tabular}{cc}
\includegraphics[width=0.45\textwidth, trim = 0cm 2.2cm 0cm 2cm, clip = true ]{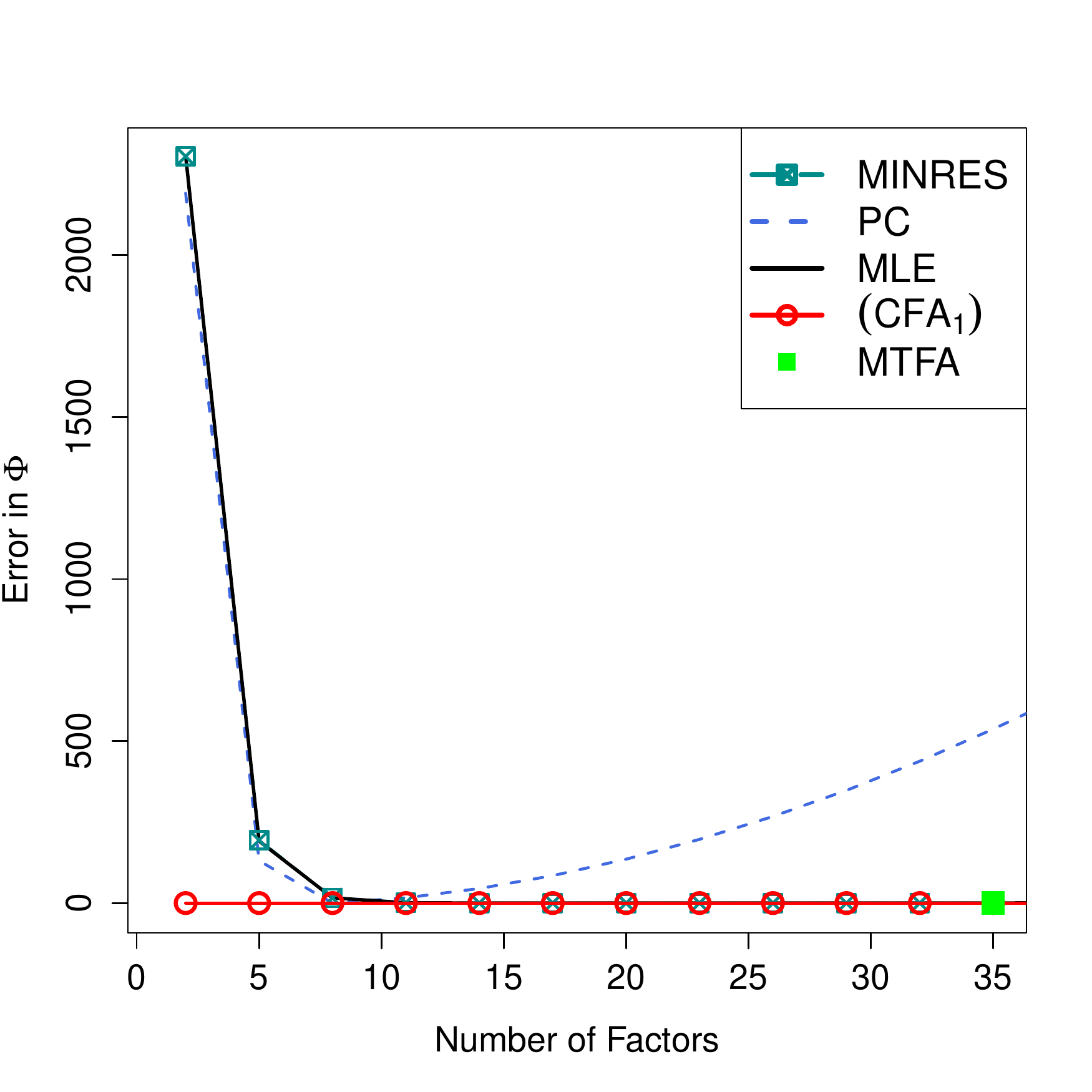}&
\includegraphics[width=0.45\textwidth,  trim = 0cm 2.2cm 0cm 2cm, clip = true ]{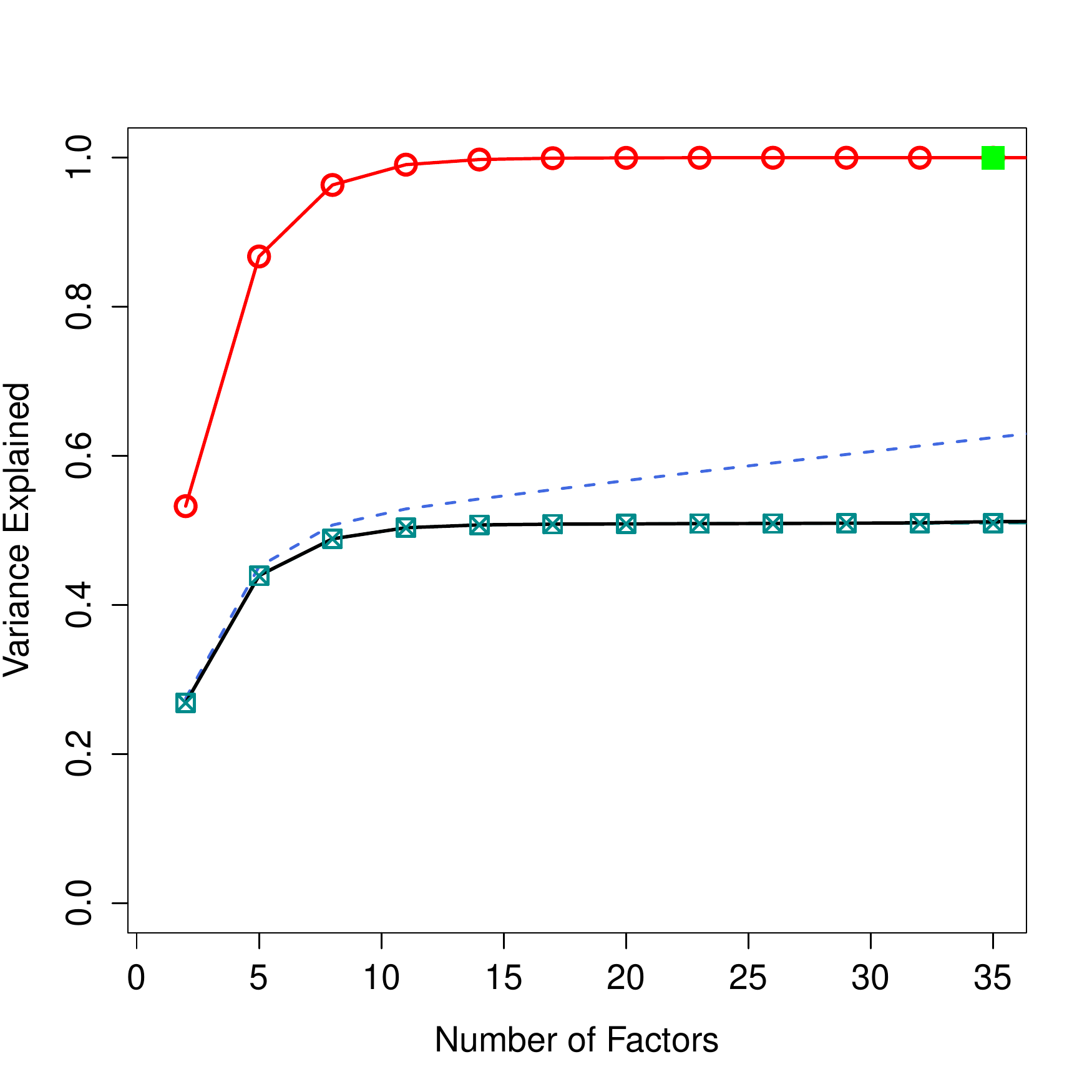}  \\
\includegraphics[width=0.45\textwidth,  trim = 0cm 2.2cm 0cm 2cm, clip = true ]{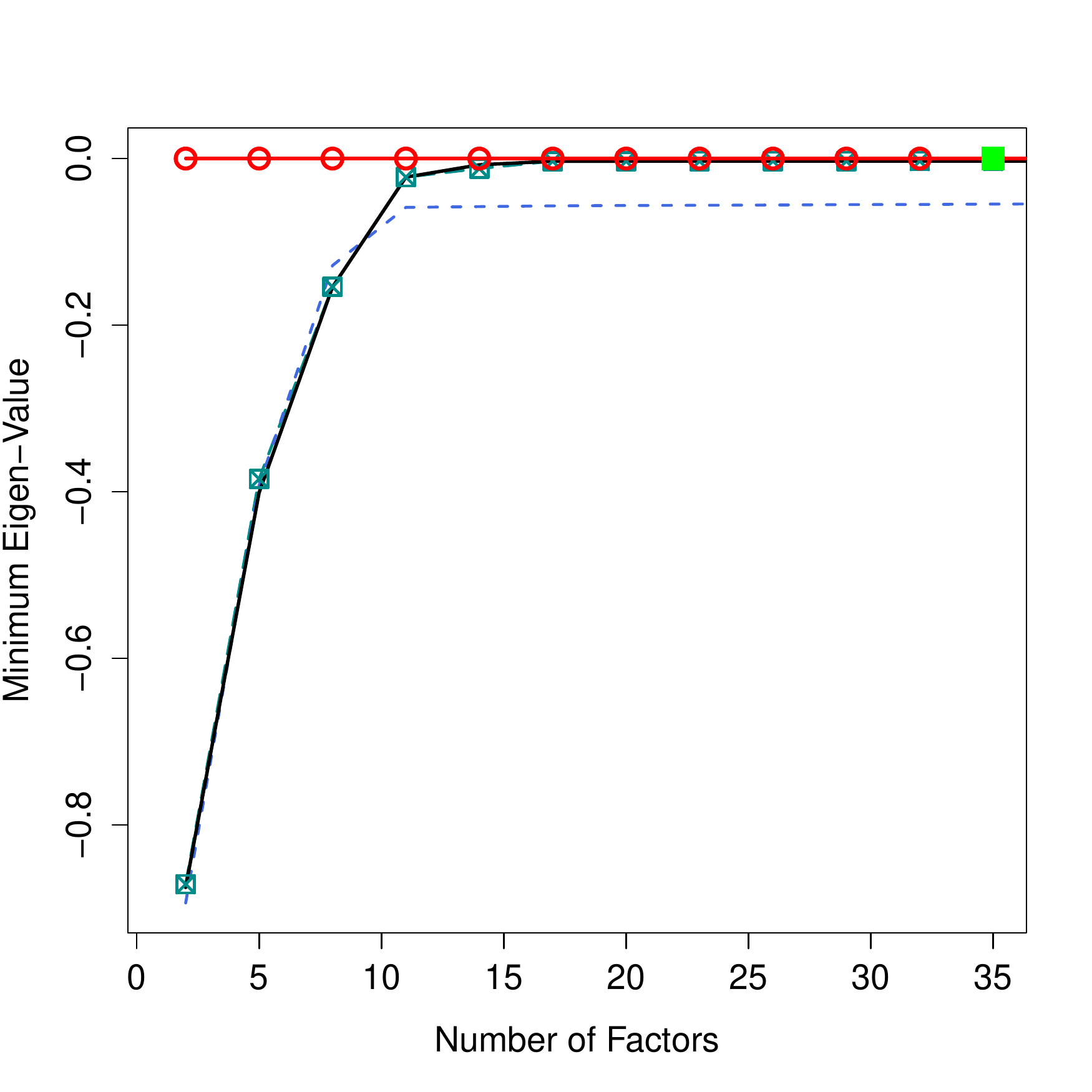}&
\includegraphics[width=0.45\textwidth, trim = 0cm 2.2cm 0cm 2cm, clip = true ]{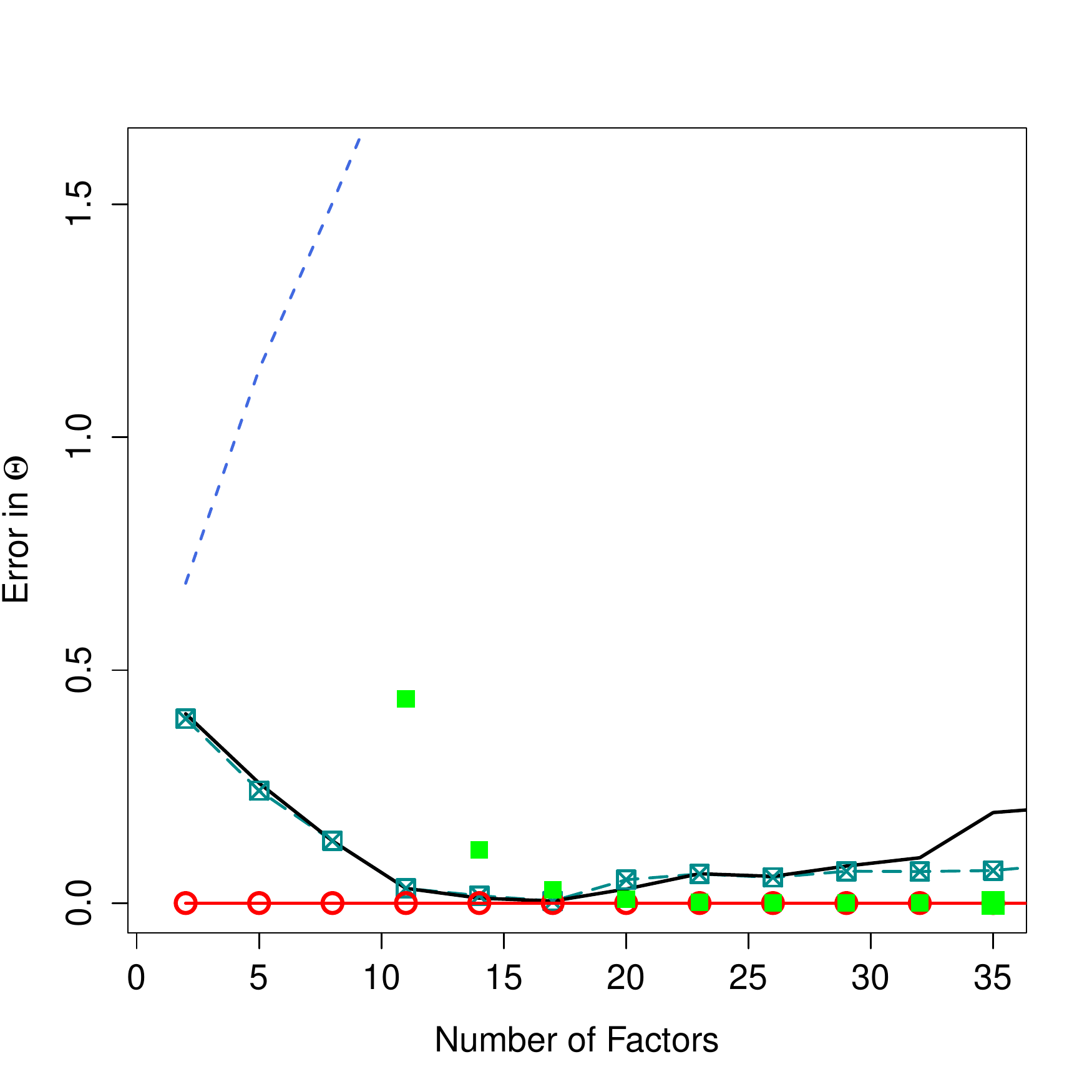} \\
 \sf \scriptsize Number of factors  &\sf \scriptsize Number of factors  \\
 \end{tabular}} \caption{{\small{Figure showing performances of various FA methods for examples from class $A_2(p=200)$ as a function of the number of factors.
The y-axis label ``Minimum Eigen Value'' refers to $\lambda_{\min}(\B\Sigma - \widehat{\B\Phi})$. We present the results of $(\textsc{CFA}_{1})$, as obtained via 
Algorithm~\ref{algo:ao1-sk} --- the results of $(\textsc{CFA}_{2})$ were similar, and hence omitted from the plot.
Our methods seems to perform better than all the other competitors. For large values of $r$, as the fits saturate, the methods become similar. The methods (as available from the {\texttt{R}} package implementations) that experienced convergence difficulties do not appear in the plot. }} }\label{fig:one}
\end{figure}

\begin{figure}
\centering
\resizebox{1.05\textwidth}{0.48\textheight}{\begin{tabular}{rccc}
& \sf  \scriptsize Type-$B_{1}$  &\sf \scriptsize Type-$B_{2}$ & \sf \scriptsize Type-$B_{3}$  \medskip \\

\rotatebox{90}{~~~~~~~~~~~\sf {\scriptsize{Error in $\Phi$}} }&\includegraphics[width=0.33\textwidth, trim =1cm 2.2cm 0cm 2cm, clip = true ]{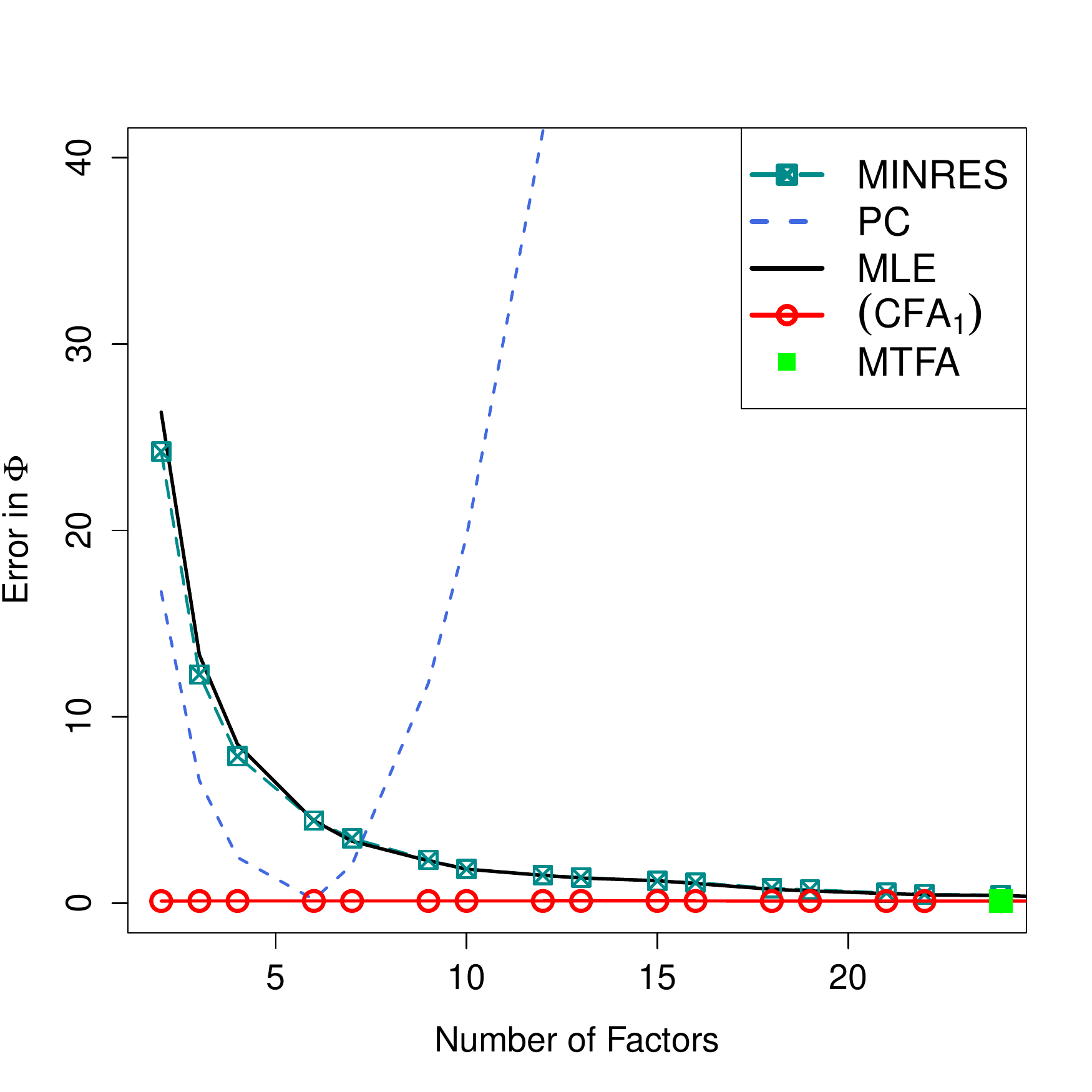}&
\includegraphics[width=0.33\textwidth, trim =  1cm 2.2cm 0cm 2cm, clip = true ]{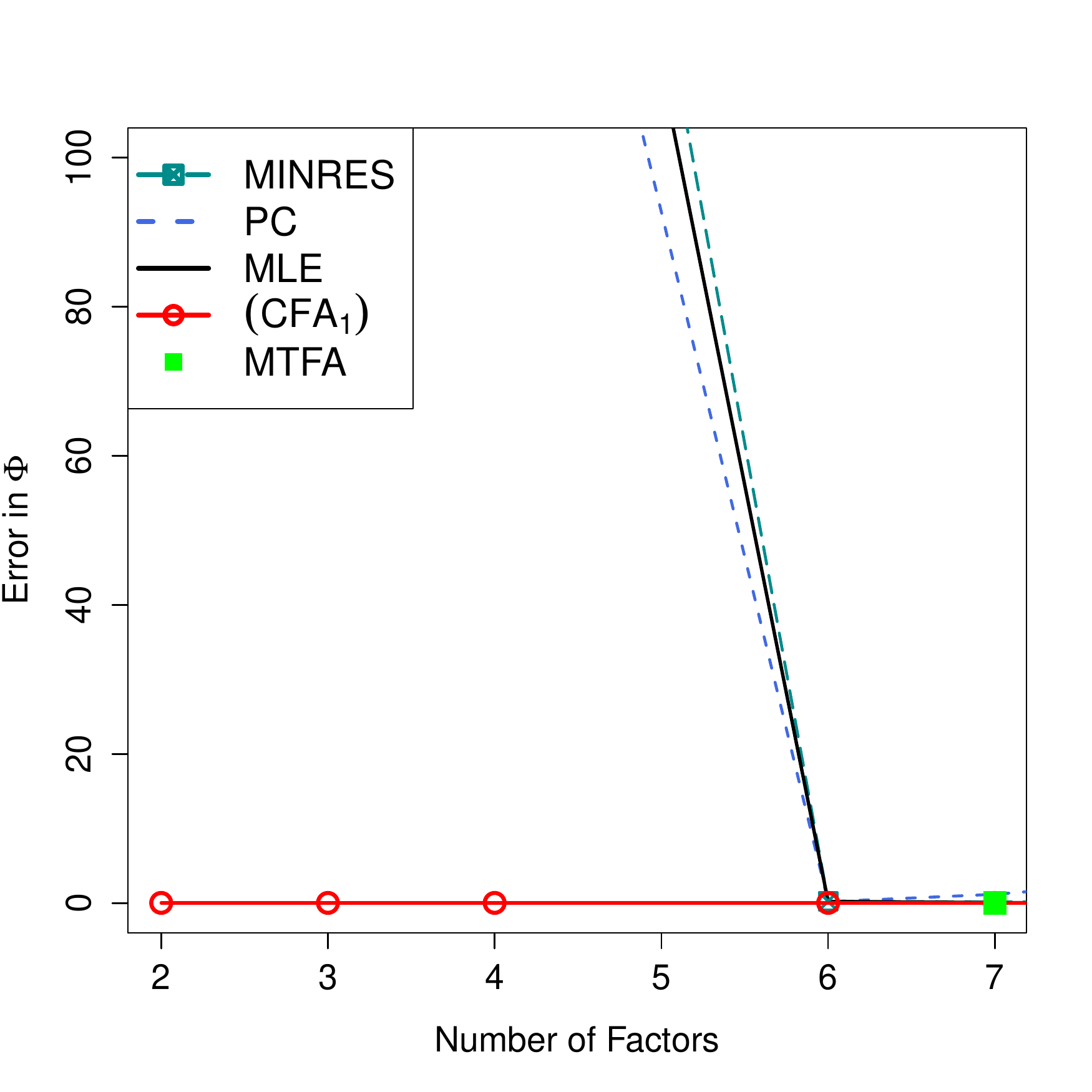}&
\includegraphics[width=0.33\textwidth, trim =  1cm 2.2cm 0cm 2cm, clip = true ]{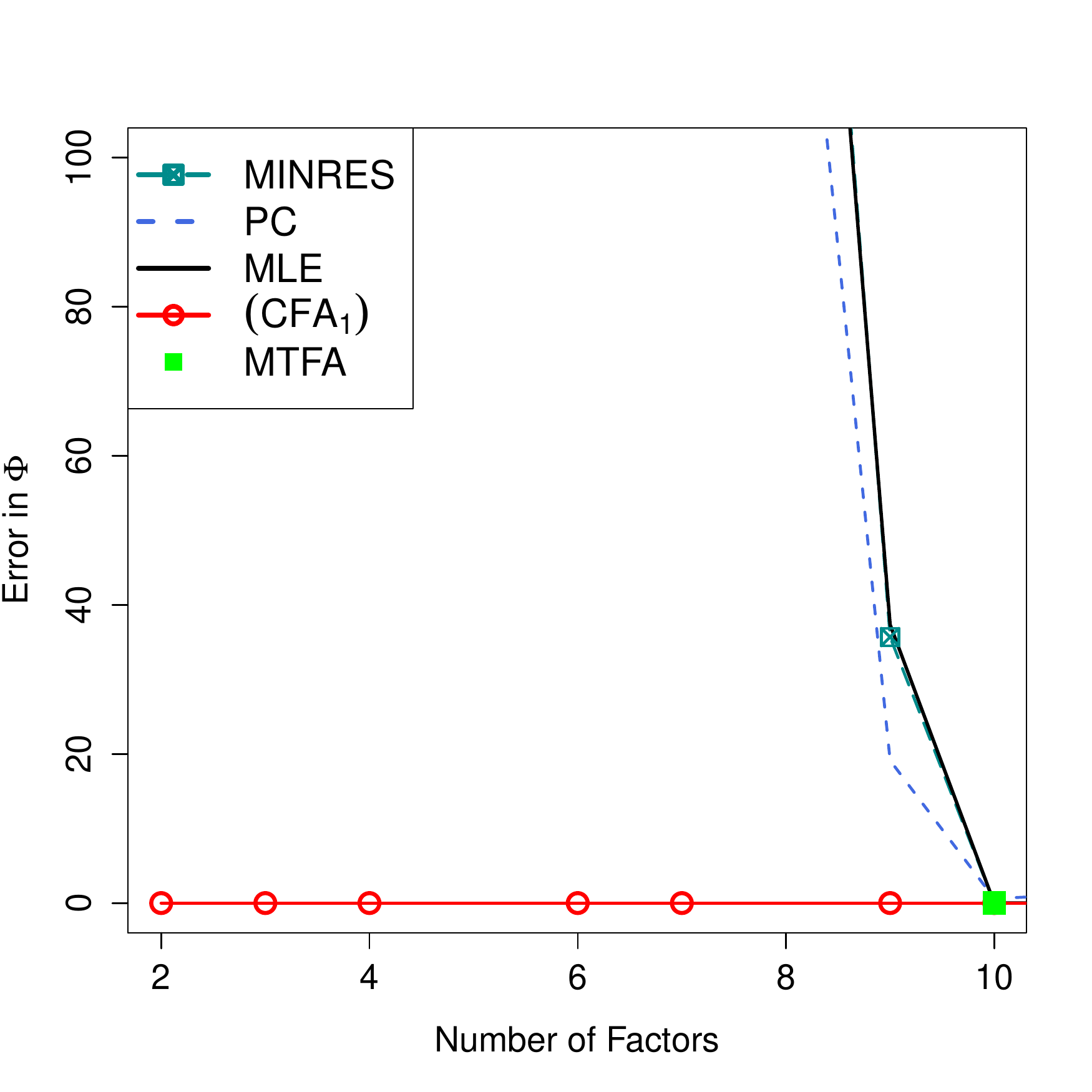}\\

\rotatebox{90}{ ~~~~~~\sf {\scriptsize{Variance Explained}} }&\includegraphics[width=0.33\textwidth,   trim = 1cm 2.2cm 0cm 1.8cm, clip = true ]{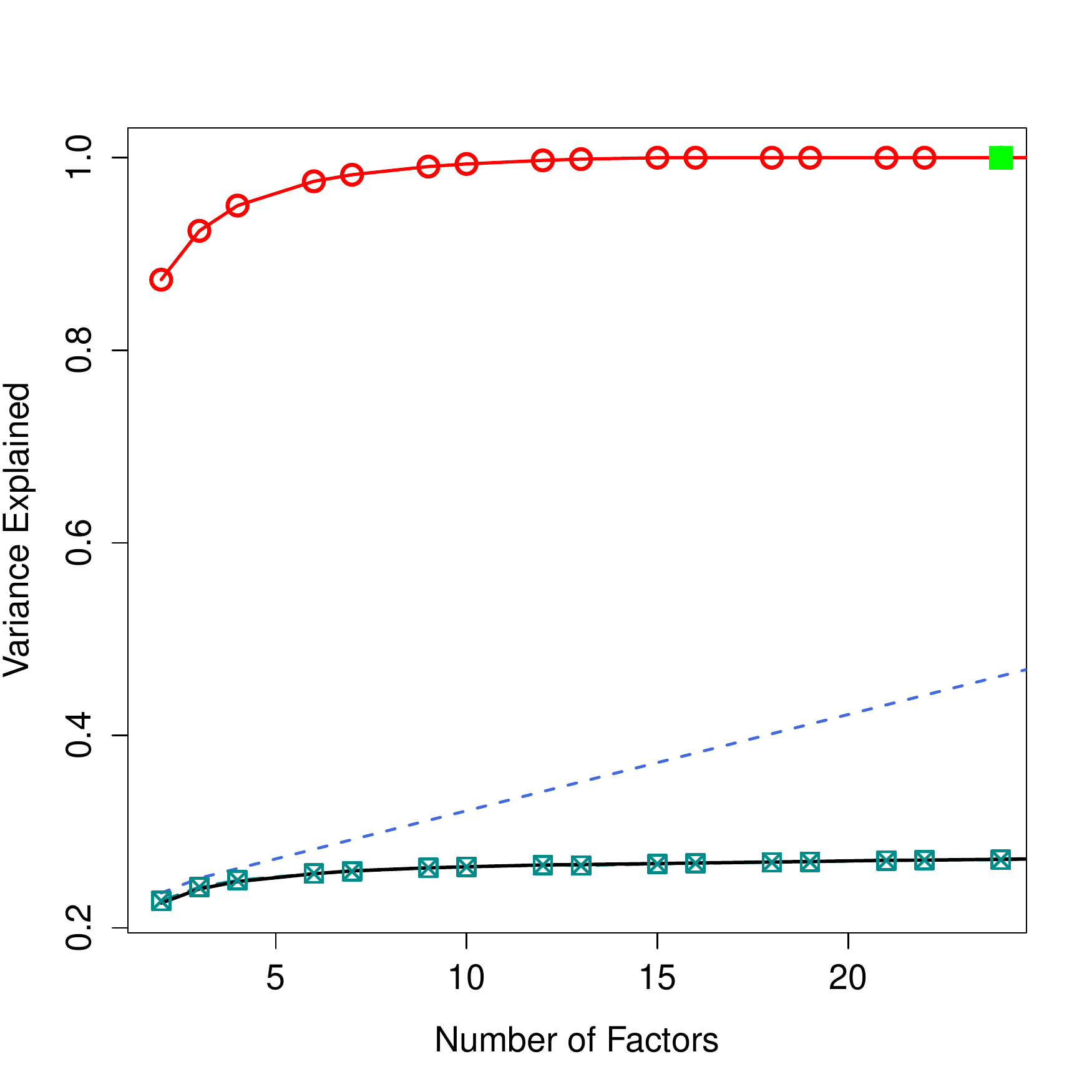}&
\includegraphics[width=0.33\textwidth,   trim =  1cm 2.2cm 0cm 1.8cm, clip = true ]{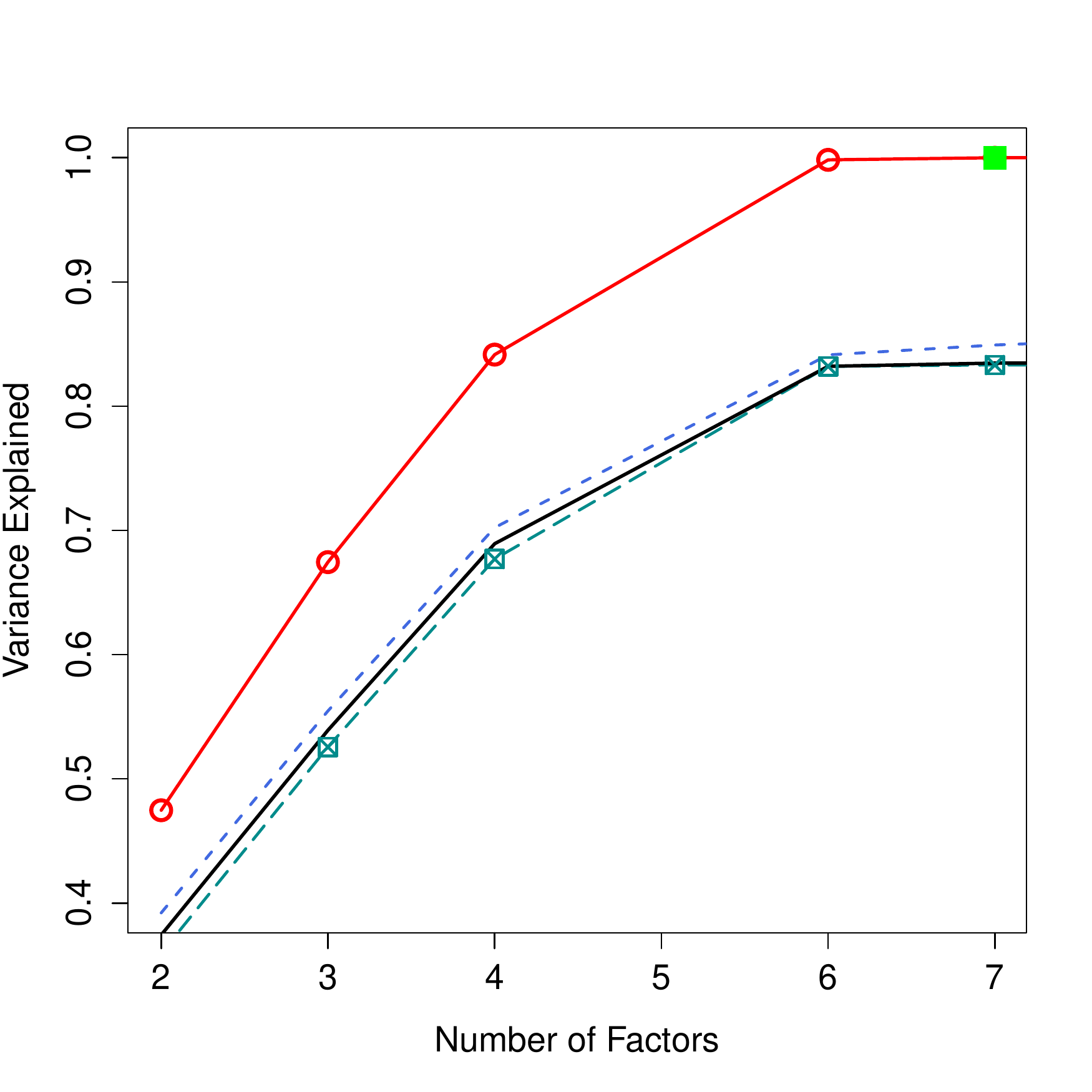}&
\includegraphics[width=0.33\textwidth,   trim =  1cm 2.2cm 0cm 1.8cm, clip = true ]{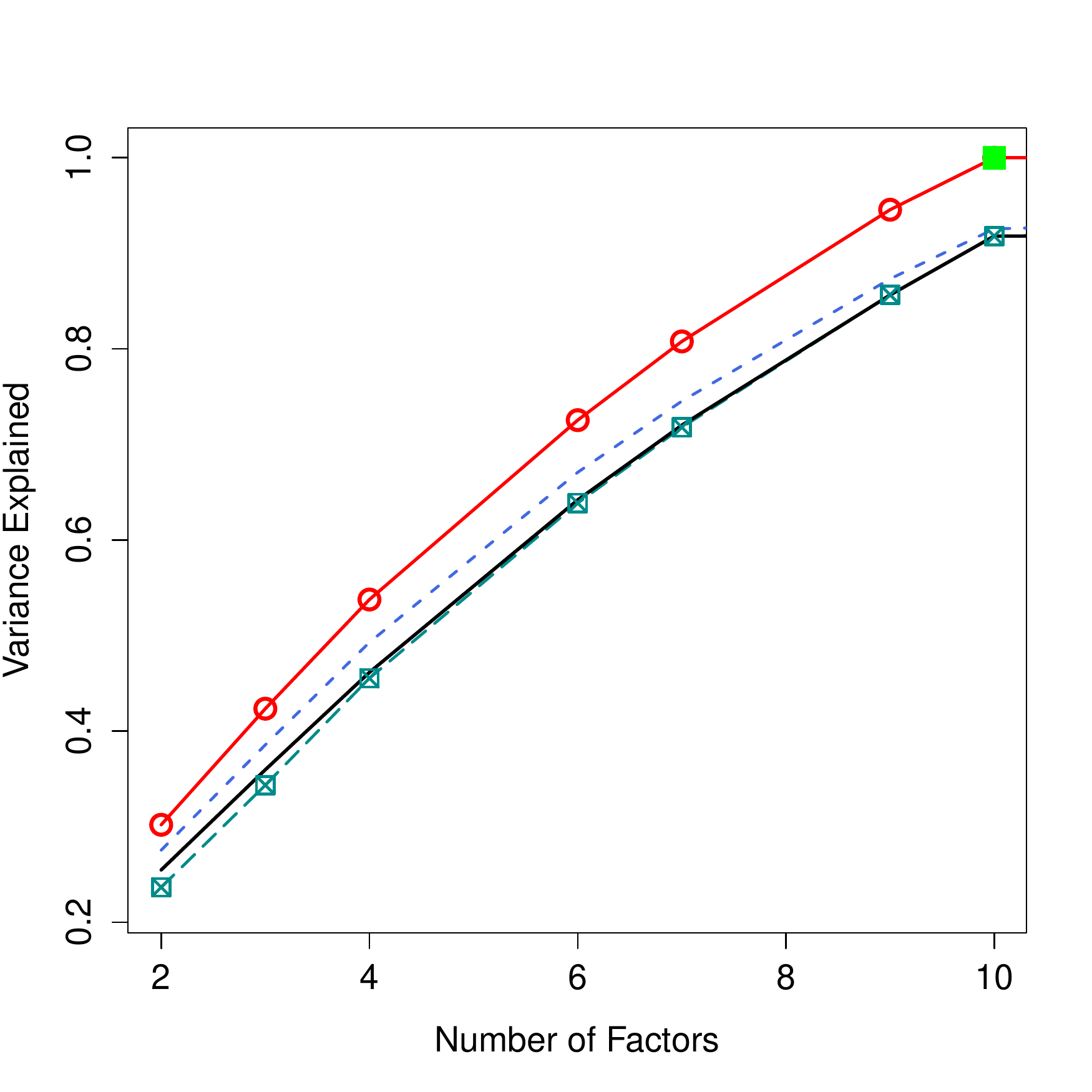} \\

\rotatebox{90}{~~~~~~~~~\sf {\scriptsize{$\lambda_{\min}(\B\Sigma - \hat{\B\Phi})$ }}}&\includegraphics[width=0.33\textwidth, trim = 1cm 2.2cm 0cm 1.8cm, clip = true ]{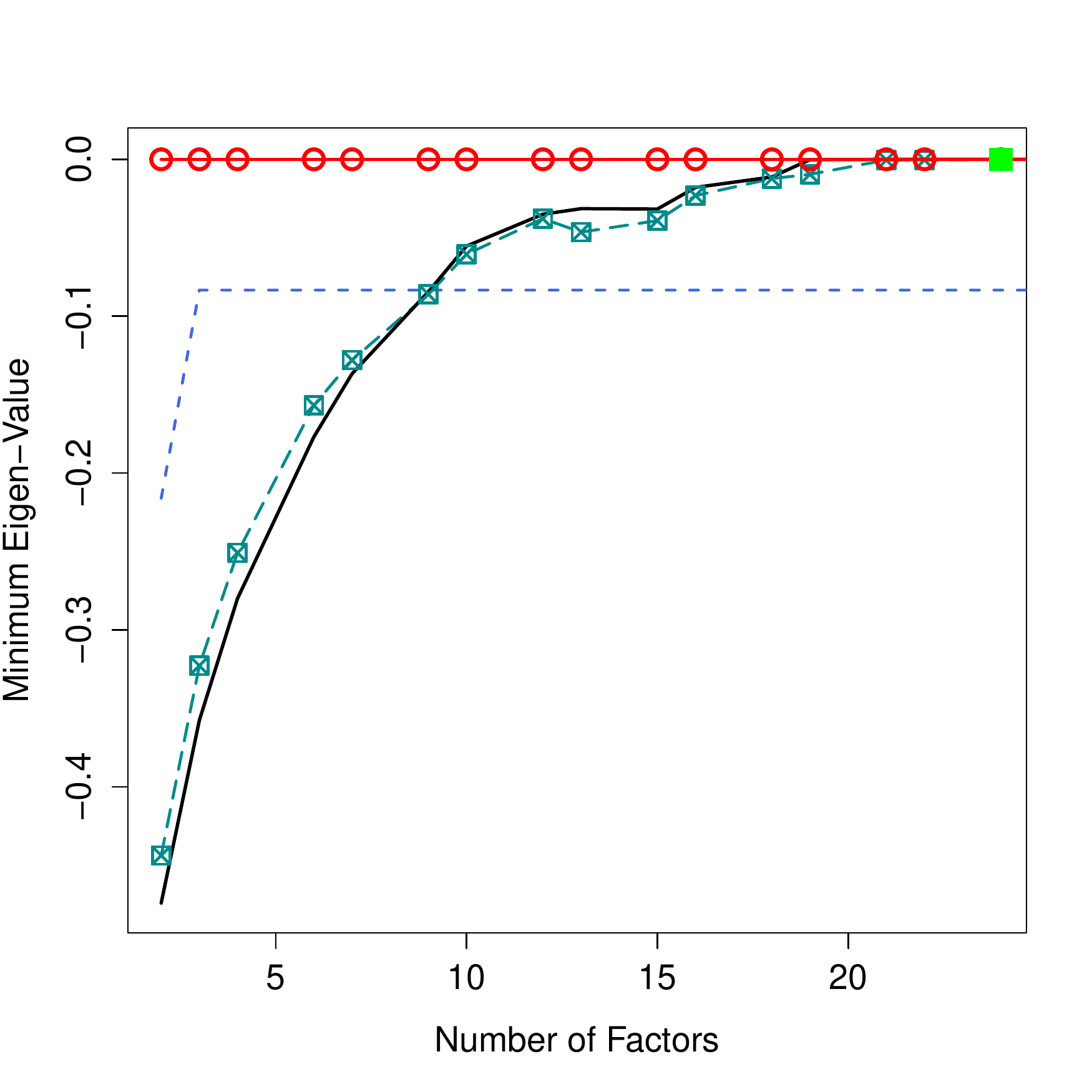}&
\includegraphics[width=0.33\textwidth, trim =  1cm 2.2cm 0cm 1.8cm, clip = true ]{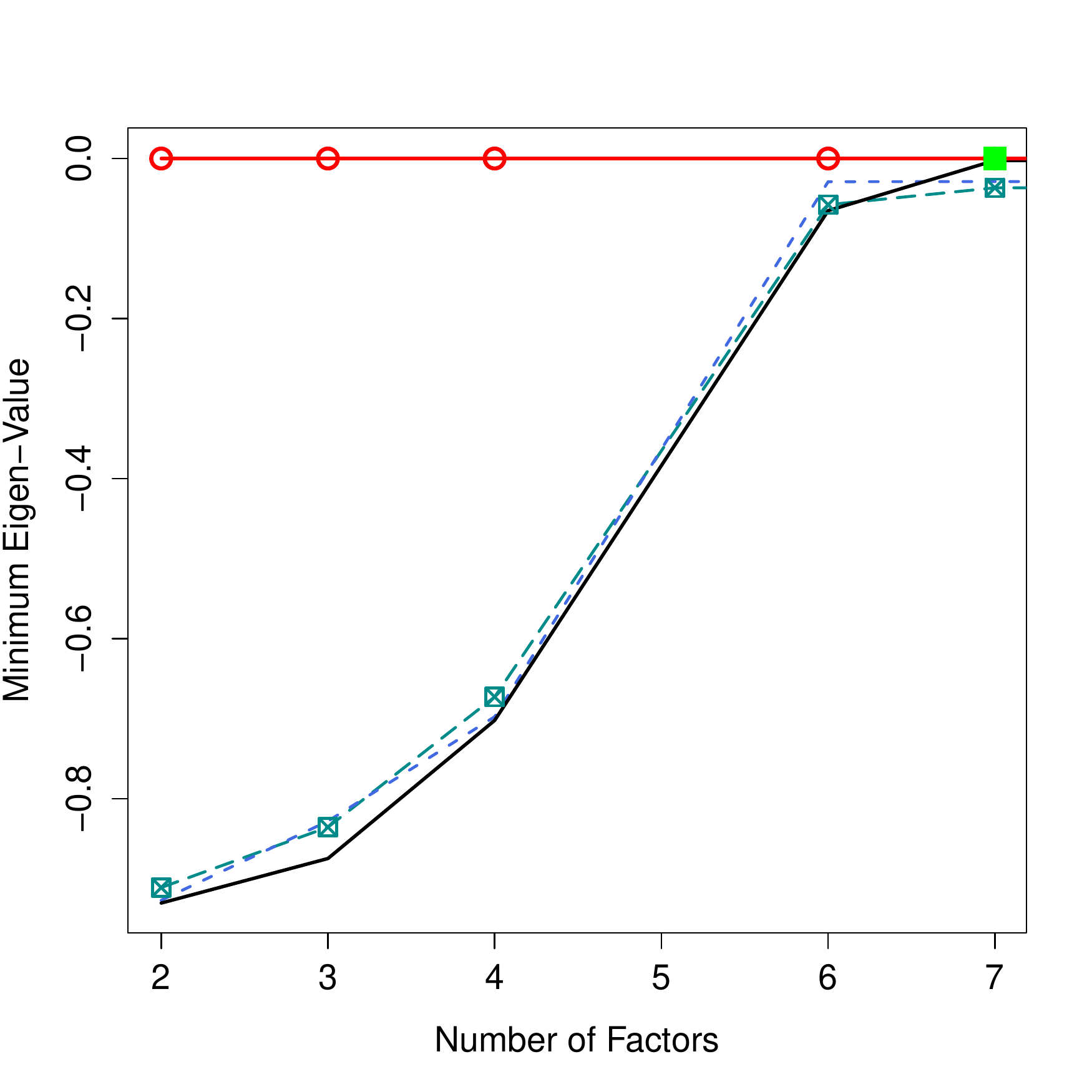}&
\includegraphics[width=0.33\textwidth, trim =  1cm 2.2cm 0cm 1.8cm, clip = true ]{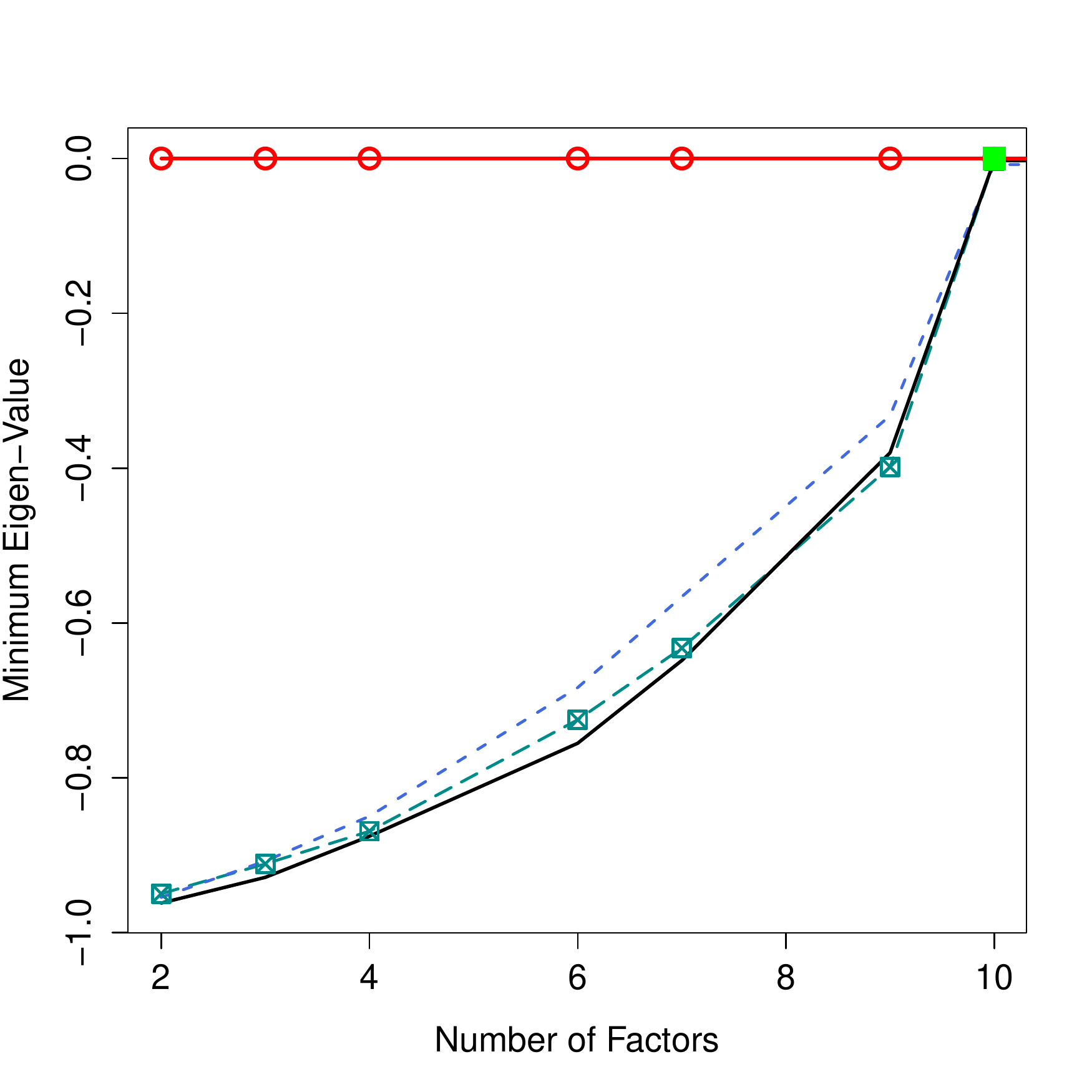} \\

\rotatebox{90}{\sf {\scriptsize{~~~~~~~~~~~~~~Error in $\B\Theta$}}}&\includegraphics[width=0.33\textwidth, trim = 1cm 2.2cm 0cm 1.8cm, clip = true ]{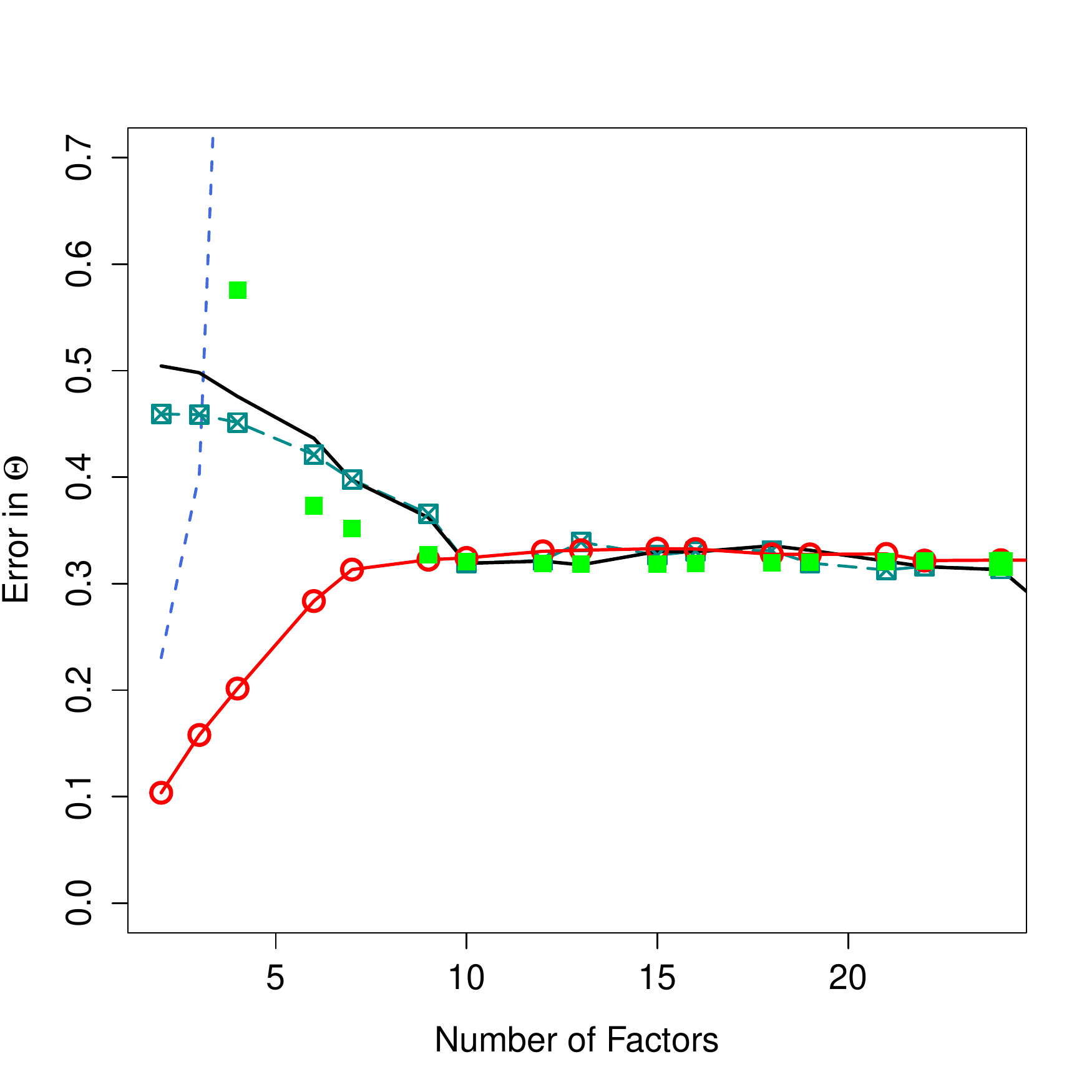}&
\includegraphics[width=0.33\textwidth, trim =  1cm 2.2cm 0cm 1.8cm, clip = true ]{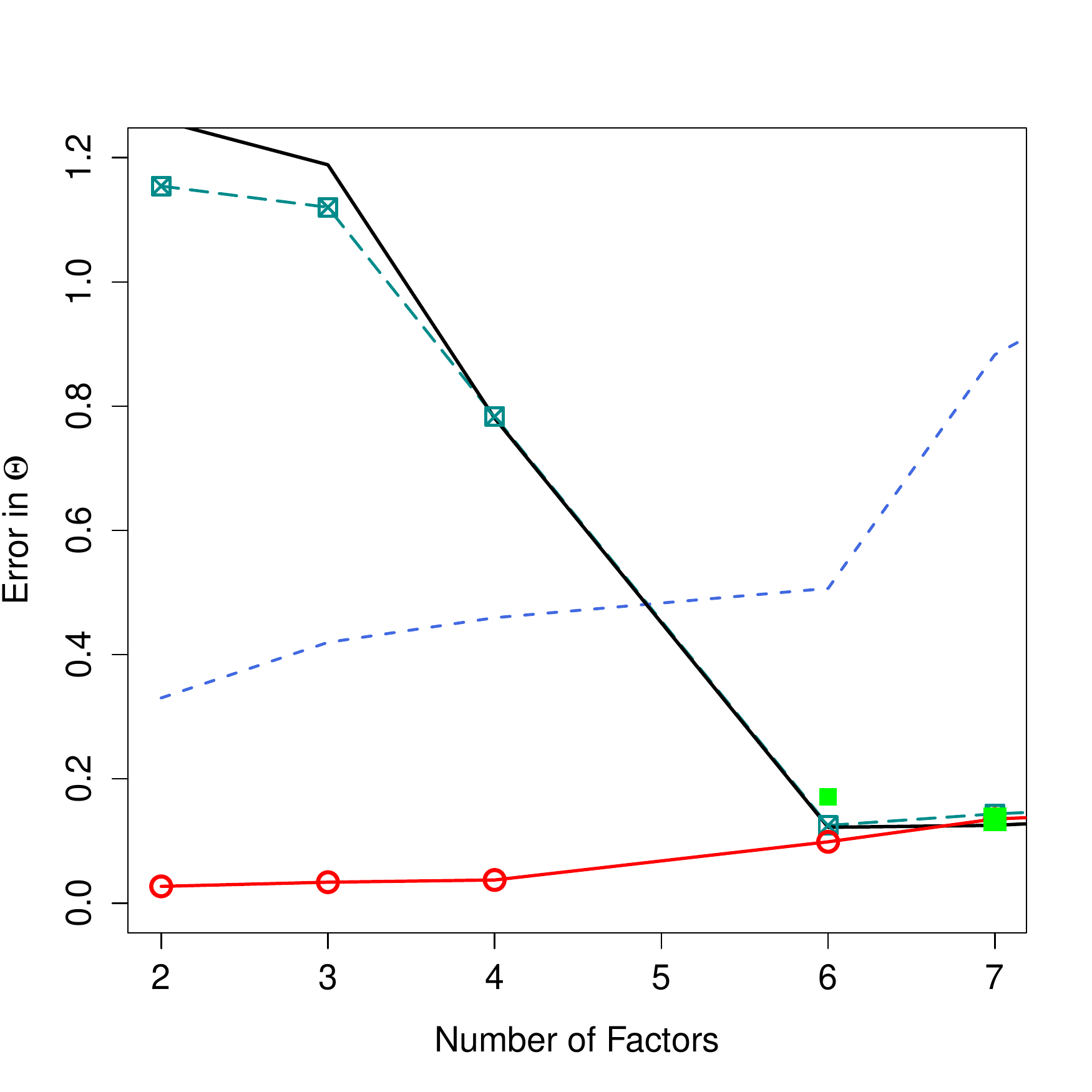}&
\includegraphics[width=0.33\textwidth, trim =  1cm 2.2cm 0cm 1.8cm, clip = true ]{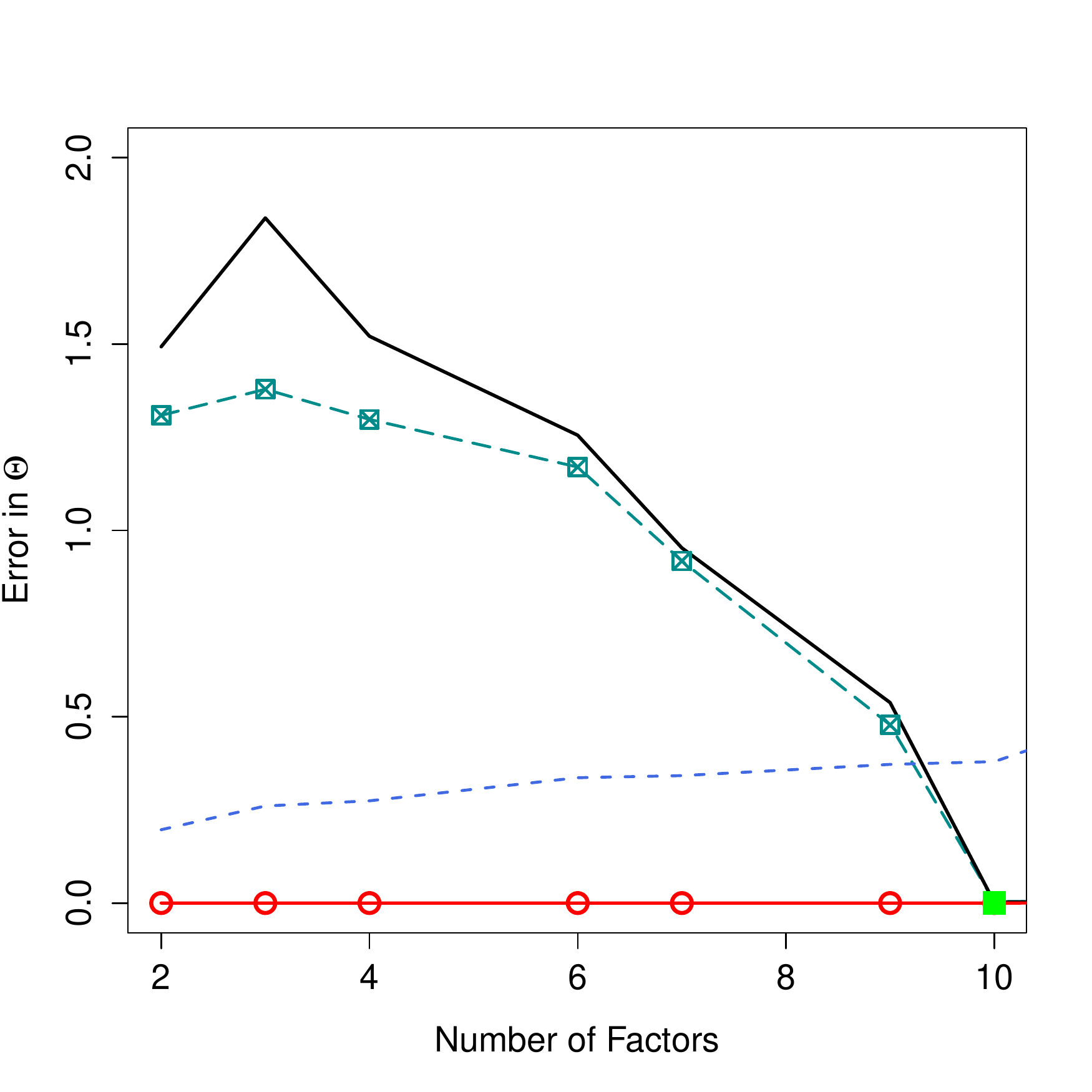} \\

& \sf \scriptsize Number of factors  &\sf \scriptsize Number of factors  &\sf \scriptsize Number of factors  \\
 \end{tabular}}
\caption{{\small{Performances of different methods, for instances of $B_1(100)$, $B_2(5/10/100)$, and $B_3(5/10/100)$.
We see that $(\textsc{CFA}_{1})$ exhibits very good performance across all instances, significantly outperforming the competing methods (the results of $(\textsc{CFA}_{2})$ were similar).}} }\label{fig:rev-types478}
\end{figure}

\begin{figure}
\centering
\resizebox{1.02\textwidth}{0.14\textheight}{\begin{tabular}{c c c c}
&\scriptsize{\textsf{bfi}}  &\scriptsize{ \textsf{Neo} }& \scriptsize{\textsf{Harman}} \\
\rotatebox{90}{ ~~~~~~\sf {\scriptsize{Variance Explained}} }&\includegraphics[width=0.3\textwidth, trim = 1cm 1.6cm 0cm 2cm, clip = true ]{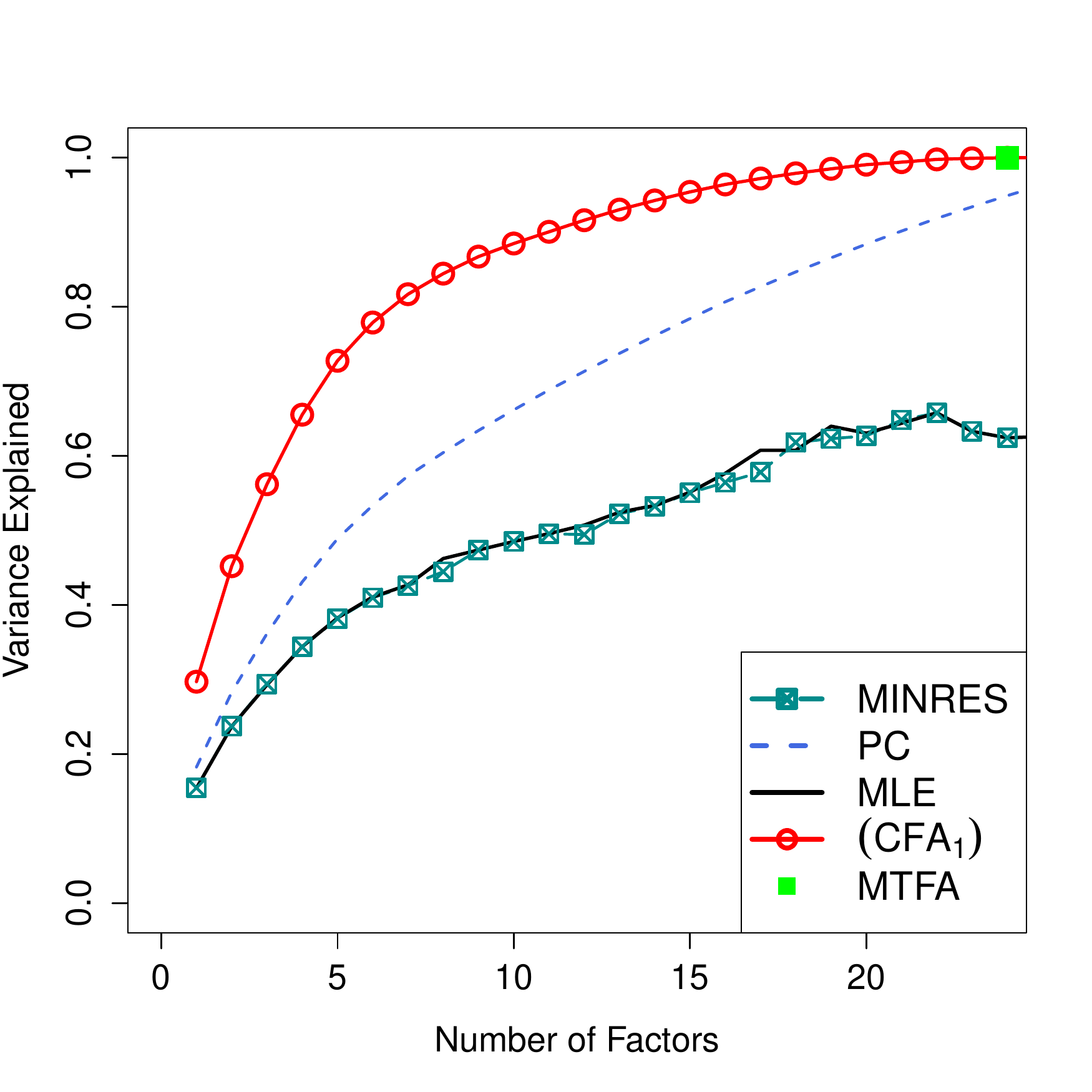}&
\includegraphics[width=0.3\textwidth,   trim = 1cm 1.6cm 0cm 2cm, clip = true ]{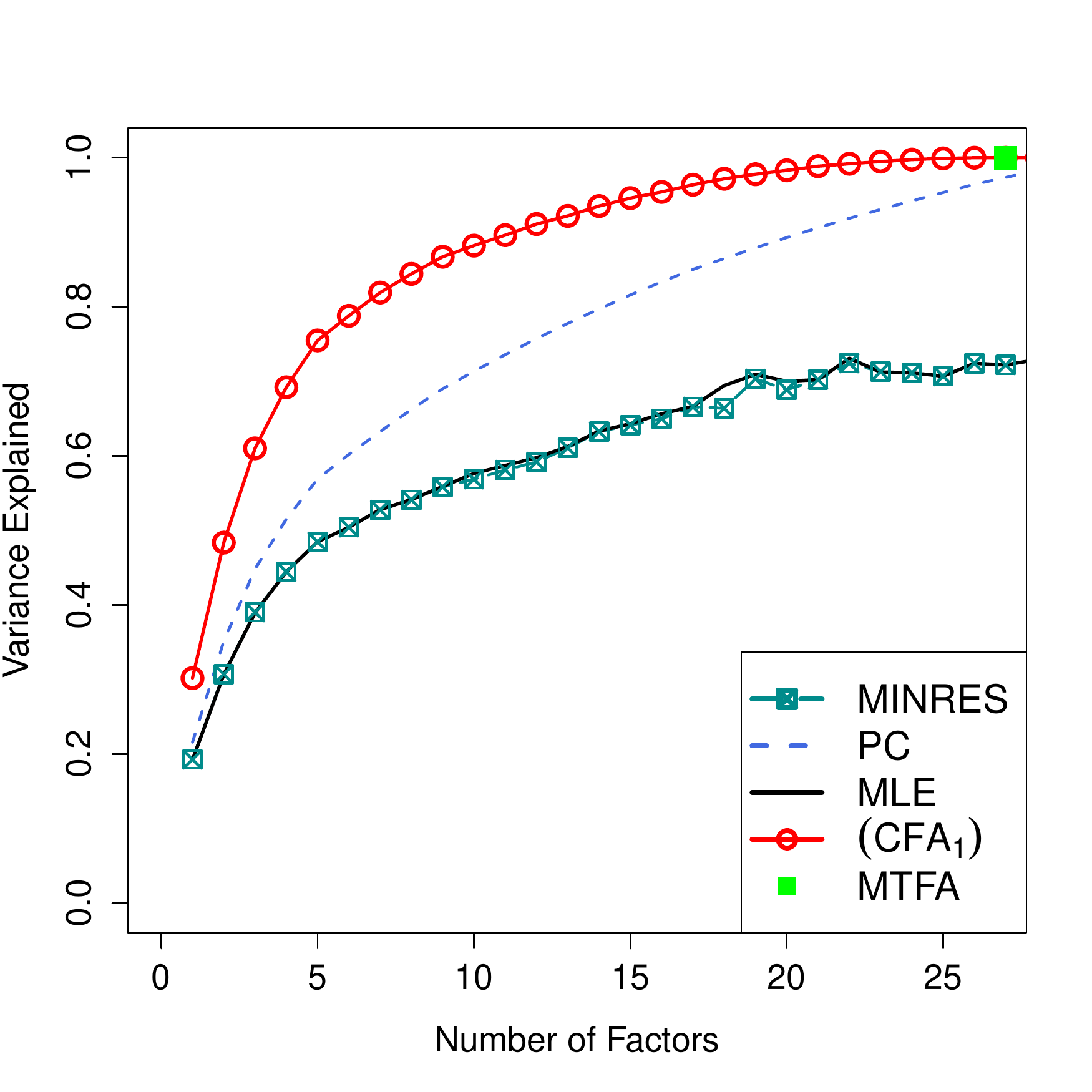}& 
\includegraphics[width=0.3\textwidth,  trim = 1cm 1.6cm 0cm 2cm, clip = true ]{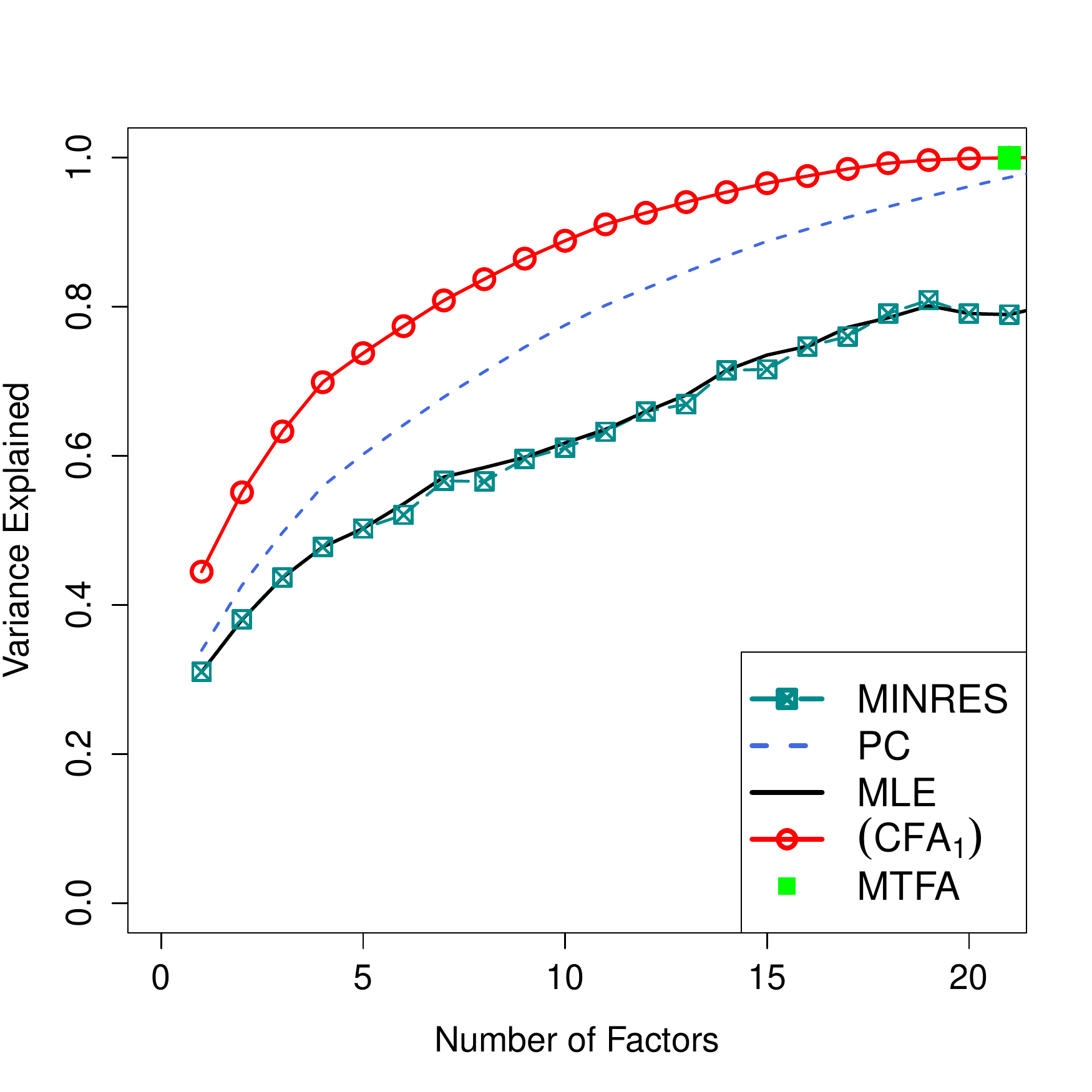} \\
& \sf \scriptsize Number of factors  &\sf \scriptsize Number of factors  &\sf \scriptsize Number of factors  \\
\end{tabular}}
\caption{{\small{Figure showing proportion of variance explained by different methods for real-data examples. 
We see that in terms of the proportion of explained variance, $(\text{CFA}_{1})$ 
delivers the largest values for different values of $r$, which is indeed desirable. $(\text{CFA}_{1})$ also shows nice flexibility in delivering different models with varying $r$, in contrast 
to \textsc{\small{MTFA}} which delivers one model with proportion of variance explained equal to one. The results of $(\text{CFA}_{2})$ were found to be similar to $(\text{CFA}_{1})$, and are thus not reported in the display.} } }\label{fig:real-one}
\end{figure}

\paragraph{Summary:} All methods of Category (B) (see Section~\ref{sec:categories}) 
used in the experimental comparisons perform worse than Category (A) in terms of  measures $\text{Error}(\B\Phi)$, $\text{Error}(\B\Theta)$ and Explained Variance for small/moderate values of $r$. They also lead to indefinite estimates of $\B\Sigma - \widehat{\B\Phi}$.
  \textsc{\small{MTFA}} performs well in estimating $\B\Phi$ but fails in estimating $\B\Theta$ mainly due to the lack in flexibility of imposing a rank constraint; in some cases 
 the trace heuristic falls short of doing a good job in approximating the rank function when compared to its non-convex counterpart $(\text{CFA}_{q})$. The estimation methods proposed herein, 
have a significant edge over existing methods in producing high quality solutions, across various performance metrics.

 \subsection{Real data examples}\label{sec:real-data-1}
This section describes the performance of different FA methods on some real-world 
benchmark real datasets popularly used in the context of FA.
We considered a few benchmark datasets used commonly in the case of factor analysis to compare our procedure with existing methods. These examples can be founded in the \texttt{R} libraries \texttt{datasets} \citep{R}, \texttt{psych} \citep{psych}, and \texttt{FactoMineR} \citep{FactoMineR} and are as follows:

\begin{itemize}
\item The {\sf bfi} data-set
has 2800 observations on 28 variables (25 personality self reported items and 3 demographic variables).
\item The {\sf neo} data-set has 1000 measurements for $p=30$ dimensions.  
\item {\sf Harman} data-setis a  correlation matrix of 24 psychological tests given to 145
     seventh and eight-grade children.
     
     \item The \textsf{geomorphology} data set is a collection of geomorphological data collected across $p=10$ variables and $75$ observations.\footnote{The \textsf{geomorphology} data set originally has $p=11$, but we remove the one categorical feature, leaving $p=10$.}

\item \textsf{JO} data set records athletic performance in Olympic events, and involves $24$ observations with $p=58$. This example is distinctive because the corresponding correlation matrix $\s$ is not full rank (having more variables than observations).
\end{itemize}
We present the results in Figure~\ref{fig:real-one} --- we also experimented with other methods: \textsc{\small{WLS}}, \textsc{\small{GLS}} the results were similar to \textsc{\small{MINRES}} 
and hence have not been shown in the figure.  For the real examples, most of the performance measures described in Section~\ref{sec:perf-meas1} do not apply; however, the notion of explained 
variance~\eqref{pro-var} does apply. We used this metric to compare the performance of different competing estimation procedures.
We observe that solutions delivered by Category (B) explain the maximum amount of residual variance for a given rank $r$, which is indeed desirable, especially in the light of its analogy with PCA on the
residual covariance matrix $\B\Sigma - \B\Phi$.

\subsection{Certificates of Optimality via Algorithm \ref{alg:bb}}\label{sec:comp}

We now turn our attention to certificates of optimality using Algorithm \ref{alg:bb}. Computational results of Algorithm \ref{alg:bb} for a variety of problem sizes across all six classes can be found in Tables \ref{tab:A1_small}, \ref{tab:A1_large}, \ref{tab:A2}, \ref{tab:B1}, \ref{tab:B2}, \ref{tab:B3}, \ref{tab:geo}, \ref{tab:harman}, \ref{tab:JO}, and \ref{tab:largeScale_1}. In general, we provide results for $p$ ranging between $10$ and $4000$. Parametric choices are outlined in depth in Table \ref{tab:A1_small}.\footnote{All computational experiments are performed in a shared cluster computing environment with highly variable demand, and therefore runtimes are not necessarily a reliable measure of problem complexity; hence, the number of nodes considered is always displayed. Further, Algorithm \ref{alg:bb} is highly parallelizable, like many branch-and-bound algorithms; however, our implementation is serial. Therefore, with improvements in code design, it is very likely that runtimes can be substantially improved beyond those shown here.}

\begin{table}
\centering
{\small \begin{tabular}{|c|c|r|r|r|r|r|r|r|}
\hline
\multicolumn{2}{|c|}{} & \multicolumn{3}{c|}{Root node} &  \multicolumn{2}{c|}{Terminal node} & \multicolumn{2}{c|}{}\\\hline
Problem size &\multirow{2}{*}{Instance} & Upper & CE & Weyl & Upper & Lower & \multirow{2}{*}{Nodes} & \multirow{2}{*}{Time (s)}\\
($R/p$) & & bound & LB & LB & bound & bound & & \\\hline
 \multirow{3}{*}{2/10} & 1 & 1.54 & 1.44 & 1.43 & 1.54 & 1.44 & 1 & 0.14 \\
 & 2 & 1.50 & 1.40 & 1.40 & 1.50 & 1.40 & 1 & 0.38 \\
 & 3 & 1.44 & 1.25 & 1.32 & 1.44 & 1.35 & 5 & 0.20\\\hline
  \multirow{3}{*}{3/10} & 1 & 0.88 & 0.49 & 0.70 & 0.88 & 0.78 & 78 & 4.08\\
 & 2 & 0.49 & 0.25 & 0.39 & 0.49 & 0.39 & 1 & 1.27\\
 & 3 & 0.52 & 0.17 & 0.42 & 0.52 & 0.42 & 14 & 0.34\\\hline
  \multirow{3}{*}{5/10} & 1 & 0.43 & $-0.90$ & 0.10 & 0.43 & 0.33 & 28163 & 19432.70\\
 & 2 & 0.06 & $-0.75$ & 0.00 & 0.06 & 0.00 & 1 & 0.20\\\
 & 3 & 0.17 & $-0.98$ & 0.00 & 0.17 & 0.07 & 3213 & 380.66\\\hline
  \multirow{3}{*}{2/20} & 1 & 3.99 & 3.91 & 3.94 & 3.99 & 3.94 & 1 & 0.19\\
 & 2 & 4.64 & 4.58 & 4.60 & 4.64 & 4.60 & 1 & 2.07\\
 & 3 & 3.34 & 3.26 & 3.28 & 3.34 & 3.28 & 1 & 0.66\\\hline
  \multirow{3}{*}{3/20} & 1 & 2.33 & 2.06 & 2.24 & 2.33 & 2.33 & 1 & 0.21\\
 & 2 & 2.55 & 2.38 & 2.49 & 2.55 & 2.49 & 1 & 6.62\\
 & 3 & 2.04 & 1.86 & 1.97 & 2.04 & 1.97 & 1 & 3.51\\\hline
  \multirow{3}{*}{5/20} & 1 & 0.61 & $-0.07$ & 0.49 & 0.61 & 0.51 & 626 & 75.98\\
 & 2 & 0.50 & $-0.01$ & 0.40 & 0.50 & 0.40 & 41 & 2.88\\
 & 3 & 0.92 & 0.18 &0.81 & 0.92 & 0.82 & 267 & 24.81 \\\hline 
\end{tabular}}
\caption{Computational results for Algorithm \ref{alg:bb} for class $C=A_1(R/p)$. All computations are performed in \texttt{julia} using SDO solvers \texttt{MOSEK} (for primal feasible solutions) and \texttt{SCS} (for dual feasible solutions within tolerance $10^{-3}$). Computation time does not include preprocessing (such as computation of $\uu$ as in \eqref{eqn:udefinition} and finding an initial incumbent feasible solution $\ph_f$ as computed via the conditional gradient algorithm in Section \ref{sec:CG-method1}). We always use default tolerance $\tol=0.1$ for algorithm termination. Parameters for branching, pruning, node selection, etc., are detailed throughout Section \ref{sec:alg}. For each problem size, we run three instances (the instance number is the random seed provided to \texttt{julia}). Upper bounds denote $z_f$, which is the best feasible solution found thus far (either at the root node or at algorithm termination). At the root node, we display two lower bounds: the lower bound arising from convex envelopes (denoted ``CE LB'') and the one arising from the Weyl bound (denoted ``Weyl LB''). Note that for lower bound at the termination node, we mean the worst bound $\max\{z^c,w^c\}$ (see Section \ref{ssec:nodeSelect}; $z^c$ is from the convex envelope approach, while $w^c$ comes from Weyl's method). ``Nodes'' indicates the number of nodes considered in the course of execution, while ``Time (s)'' denotes the runtime (in seconds). We set $r^*$, the rank used within Algorithm \ref{alg:bb}, to $r^*=R-1$, where $R$ is the generative rank displayed. All results displayed to two decimals. Computations run on high-demand, shared cluster computing environment with variable architectures. Runtime is capped at 400000s (approximately 5 days), and any instance which is still running at that time is marked with an asterisk next to its runtime.}
\label{tab:A1_small}
\end{table}

\begin{table}[h!]
\centering
{\small \begin{tabular}{|c|c|r|r|r|r|r|r|r|}
\hline
\multicolumn{2}{|c|}{} & \multicolumn{3}{c|}{Root node} &  \multicolumn{2}{c|}{Terminal node} & \multicolumn{2}{c|}{}\\\hline
Problem size &\multirow{2}{*}{Instance} & Upper & CE & Weyl & Upper & Lower & \multirow{2}{*}{Nodes} & \multirow{2}{*}{Time (s)}\\
($R/p$) & & bound & LB & LB & bound & bound & & \\\hline
 \multirow{3}{*}{3/50} & 1 & 7.16 & 7.09 & 7.14 & 7.16 & 7.14 & 1 & 18.66\\
 & 2 & 6.74 & 6.66 & 6.71 & 6.74 & 6.71 & 1 & 13.19 \\
 & 3 & 6.78 & 6.74 & 6.77 & 6.78 & 6.77 & 1 & 33.28\\\hline
 \multirow{3}{*}{5/50} & 1 & 3.04 & 2.86 & 3.01 & 3.04 & 3.01 & 1 & 19.13\\
 & 2 & 3.32 & 3.12 & 3.29 & 3.32 & 3.29 & 1 & 7.33\\
 & 3 & 3.70 & 3.56 & 3.67 & 3.70 & 3.67 & 1 & 52.53\\\hline
 \multirow{3}{*}{10/50} & 1 & 0.88 & $-0.15$ & 0.81 & 0.88 & 0.81 & 1 & 6.22\\
 & 2 & 1.20 & 0.08 & 1.12 & 1.20 & 1.12 & 1 & 17.41\\
 & 3 & 1.14 & $-0.02$ & 1.07 & 1.14 & 1.07 & 1 & 8.88\\\hline
 \multirow{3}{*}{3/100} & 1 & 13.72 & 13.69 & 13.72 & 13.72 & 13.72 & 1 & 125.17\\
 &2 & 14.07 & 14.03 & 14.06 & 14.07 & 14.06 & 1 & 63.25\\
 & 3 & 14.67 & 14.64 & 14.66 & 14.67 & 14.66 & 1 & 28.36\\\hline
  \multirow{3}{*}{5/100} & 1 & 8.54 & 8.46 & 8.53 & 8.54 & 8.53 & 1 & 117.71\\
  & 2 & 7.16 & 7.04 & 7.14 & 7.16 & 7.14 & 1 & 70.09\\
  & 3 & 7.07 & 6.94 & 7.06 & 7.07 & 7.06 & 1 & 78.60\\\hline
  \multirow{3}{*}{10/100} & 1 & 2.62 & 2.10 & 2.58 & 2.62 & 2.58 & 1 & 36.91\\
  &2 & 2.81 & 2.38 & 2.78 & 2.81 & 2.78 & 1 & 9.25\\
  &3 & 3.16 & 2.64 & 3.12 & 3.16 & 3.12 & 1 & 49.04\\\hline
\end{tabular}}
\caption{Computational results for larger examples from  class $C=A_1(R/p)$. Same parameters as in Table \ref{tab:A1_small} but with larger problem instances. Again we set $r^*=R-1$ here.}
\label{tab:A1_large}
\end{table}

\begin{table}[h!]
\centering
{\small \begin{tabular}{|c|c|c|r|r|r|r|r|r|r|}
\hline
\multicolumn{3}{|c|}{} & \multicolumn{3}{c|}{Root node} &  \multicolumn{2}{c|}{Terminal node} & \multicolumn{2}{c|}{}\\\hline
Problem size &\multirow{2}{*}{Instance} & $r^*$ & Upper & CE & Weyl & Upper & Lower & \multirow{2}{*}{Nodes} & \multirow{2}{*}{Time (s)}\\
($p$) & & used & bound & LB & LB & bound & bound & & \\\hline
  \multirow{3}{*}{10} & 1 & \multirow{3}{*}{2} & 0.98 & $-0.29$ & 0.26 & 0.96 & 0.86 & 1955 & 91.99\\
  &2 && 0.85 & $-0.28$ & 0.29 & 0.85 & 0.75 & 736 & 37.70\\
  & 3 && 0.89 & $-0.29$ & 0.23 & 0.89 & 0.79 & 1749 & 53.30\\\hline
    \multirow{3}{*}{10} & 1 & \multirow{3}{*}{3} & 0.53 & $-0.43$ & 0.13 & 0.53 & 0.43 & 3822 & 409.68\\
  &2 && 0.58 & $-0.45$ & 0.15 & 0.58 & 0.48 & 8813 & 1623.38\\
  & 3 && 0.64 & $-0.44$ & 0.10 & 0.59 & 0.49 & 13061 & 1671.61\\\hline
 \multirow{3}{*}{20} & 1 & \multirow{3}{*}{2} & 5.13 & 4.14 & 2.13 & 5.13 & 5.03 & 29724 & 36857.46\\
 & 2 & &4.68 & 3.44 & 2.00 & 4.68 & 4.58 & 28707 & 35441.37\\
 & 3 & &4.83 & 3.67 & 2.00 & 4.83 & 4.73 & 17992 & 11059.55\\\hline
  \multirow{3}{*}{20} & 1 & \multirow{3}{*}{3}& 4.26 & 2.68 & 1.55 & 4.26 & 4.05 & 86687 & 400002.3*\\
  & 2 && 3.88 & 2.00 & 1.46 & 3.88 & 3.71 & 132721 & 400003.7*\\
  & 3 && 4.03 & 2.25 & 1.48 & 4.03 & 3.84 & 100042 & 400002.1*\\\hline
  \multirow{3}{*}{50} & 1 & \multirow{3}{*}{3} & 17.83 & 16.54 & 11.49 & 17.83 & 17.44 & 35228 & 400002.5*\\
  & 2 && 18.57 & 17.34 & 12.19 & 18.57 & 18.20 & 34006 & 400007.1*\\
  & 3 && 18.21 & 17.06 & 11.96 & 18.21 & 17.86 & 35874 & 400004.9*\\\hline
  \multirow{3}{*}{100} & 1 &\multirow{3}{*}{3}& 38.65 & 37.50 & 33.08 & 38.65 & 38.18 & 20282 & 400015.2*\\
  &2 && 38.84 & 37.66 & 33.25 & 38.84 & 38.37 & 17306 & 400013.7*\\
  & 3 && 39.02 & 37.86 & 33.67 & 39.02 & 38.58 & 19889 & 400004.9*\\\hline
  \multirow{3}{*}{100} & 1 & \multirow{3}{*}{5} & 30.98 & 29.05 & 25.69 & 30.98 & 29.91 & 16653 & 400005.6*\\
  &2 && 31.37 & 29.39 & 26.07 & 31.37 & 30.32 & 26442 & 400021.4*\\
  & 3 && 31.20 & 29.24 & 26.13 & 31.20 & 30.15 & 19445 & 400043.4*\\\hline
\end{tabular}}
\caption{Computational results for class $C=A_2(p)$. All parameters as per Table \ref{tab:A1_small}. Here we show the behavior across a variety of choices of the parameter $r^*$. Recall that the class $A_2$ is generated by a common variance component with high rank, and exponentially decaying eigenvalues.}
\label{tab:A2}
\end{table}

\begin{table}[h!]
\centering
{\small \begin{tabular}{|c|c|c|r|r|r|r|r|r|r|}
\hline
\multicolumn{3}{|c|}{} & \multicolumn{3}{c|}{Root node} &  \multicolumn{2}{c|}{Terminal node} & \multicolumn{2}{c|}{}\\\hline
Problem size &\multirow{2}{*}{Instance} & $r^*$ & Upper & CE & Weyl & Upper & Lower & \multirow{2}{*}{Nodes} & \multirow{2}{*}{Time (s)}\\
($r/p$) & & chosen & bound & LB & LB & bound & bound & & \\\hline
 \multirow{3}{*}{4/10} & 1 & 1 & 0.25 & 0.06 & 0.11 & 0.25 & 0.20 & 24 & 0.66\\
 &2&1 & 0.17 & $-0.04$ & 0.04 & 0.17 & 0.12 & 17 & 0.50\\
 & 3 & 1 & 0.16 & $-0.08$ & 0.00 & 0.16 & 0.11 & 26 & 0.16\\\hline
  \multirow{3}{*}{6/10} & 1 & 2 &0.09 & $-0.27$ & 0.00 & 0.04 & 0.09 & 77 & 0.62\\
  & 2 & 2 & 0.07 & $-0.36$ & 0.00 & 0.07 & 0.02 & 118 & 0.75\\
  & 3 & 2 & 0.07 & $-0.36$ & 0.00 & 0.07 & 0.02 & 125 & 3.79\\\hline
  \multirow{3}{*}{10/10} & 1 & 3 & 0.10 & $-0.46$ & 0.00 & 0.10 & 0.05 & 2162 & 166.44\\
  & 2 & 3 &0.10 & $-0.62$ & 0.00 & 0.10 & 0.05 & 3282 & 406.82\\
  & 3 &4 & 0.06 & $-0.49$ & 0.00 & 0.04 & 0.00 & 122 & 11.69\\\hline
  \multirow{3}{*}{4/20} & 1 & 1 & 0.26 & 0.10 & 0.13 & 0.26 & 0.21 & 21 & 0.17\\
  &2 & 1 & 0.19 & 0.01 & 0.07 & 0.19 & 0.14 & 15 & 0.40\\
  &3 & 1 & 0.19 & $-0.04$ & 0.04 & 0.19 & 0.14 & 26 & 0.59\\\hline
  \multirow{3}{*}{6/20} & 1 &  2 & 0.10 & $-0.24$ & 0.01 & 0.10 & 0.05 & 73 & 1.47\\
  &2 & 2 & 0.08 & $-0.32$ & 0.00 & 0.08 & 0.03 & 103 & 2.82 \\
  & 3 & 2 & 0.08 & $-0.30$ & 0.00 & 0.08 & 0.03 & 96 & 1.64\\\hline
  \multirow{3}{*}{10/20} & 1 & 3 & 0.13 & $-0.44$ & 0.00 & 0.13 & 0.08 & 2602 & 113.82\\
  &2 & 3 & 0.12 & $-0.59$ & 0.00 & 0.12 & 0.07 & 3971 & 188.93\\
  &3 & 3 & 0.15 & $-0.45$ & 0.01 & 0.15 & 0.10 & 3910 & 531.04\\\hline
  \multirow{3}{*}{6/50} & 1 & 2 & 0.11 & $-0.15$ & 0.04 & 0.11 &0.06 & 50 & 0.66\\
  &2 & 2 & 0.10 & $-0.18$ & 0.01  & 0.10 & 0.05 & 61 & 1.95\\
  &3 & 2 & 0.10 & $-0.20$ & 0.01 & 0.10 & 0.05 & 64 & 0.58 \\\hline
  \multirow{3}{*}{10/50} & 1 & 3 & 0.17 & $-0.35$ & 0.03 & 0.17 & 0.12 & 2283 & 333.62\\
  &2 & 3 & 0.17 & $-0.46$ & 0.02 & 0.17 & 0.12 & 4104 & 544.71\\
  &3 & 3 & 0.19 & $-0.36$ & 0.03 & 0.19 & 0.14 & 4362 & 620.61\\\hline
  \multirow{3}{*}{6/100} & 1 & 2 & 0.12 & $-0.07$ & 0.05 & 0.12 & 0.05 & 1 & 0.71 \\
  &2 & 2 & 0.11 & $-0.09$ & 0.04 & 0.11 & 0.07 & 12 & 0.82\\
  &3 & 2 & 0.11 & $-0.12$ & 0.03 & 0.11 & 0.06 & 30 & 0.46\\\hline
  \multirow{3}{*}{10/100} & 1 & 3 & 0.19 & $-0.22$ & 0.06 & 0.19 & 0.14 & 1535 & 237.19\\
  &2 & 3 & 0.19 & $-0.28$ & 0.05 & 0.19 & 0.14 & 2625 & 383.27\\
  &3 & 3 & 0.20 & $-0.26$ & 0.05 & 0.20 & 0.15 & 2829 & 355.08\\\hline
\end{tabular}}
\caption{Computational results for class $C=B_1(R/p)$. We choose $r^*$ during computation as the largest $r$ such that the computational gradient method of Section \ref{sec:CG-method1} is strictly positive (up to additive tolerance $0.05$; we use a smaller value here because the objective values are smaller across this class). For this class, examples can be preprocessed because $\s\sim B_1(R/p)$ has a block of size $R\times R$ in the upper left, with all other entries set to zero except the diagonal. Hence, it suffices to perform factor analysis with the truncated matrix $\s_{1:R,1:R}$.}
\label{tab:B1}
\end{table}

\begin{table}[h!]
\centering
{\small \begin{tabular}{|c|c|c|r|r|r|r|r|r|r|}
\hline
\multicolumn{3}{|c|}{} & \multicolumn{3}{c|}{Root node} &  \multicolumn{2}{c|}{Terminal node} & \multicolumn{2}{c|}{}\\\hline
Problem size &\multirow{2}{*}{Instance} & $r^*$ & Upper & CE & Weyl & Upper & Lower & \multirow{2}{*}{Nodes} & \multirow{2}{*}{Time (s)}\\
($r/R/p$) & & chosen & bound & LB & LB & bound & bound & & \\\hline
 \multirow{3}{*}{2/5/20} & 1 & 4 &1.15 & 0.4 & 1.03 & 1.15 & 1.05 & 517 & 146.72\\
 &2 &4& 1.08 & 0.50 & 0.97 & 1.08 & 0.98 & 366 & 60.56\\
 & 3 & 4 & 0.89 & 0.29 & 0.79 & 0.89 & 0.79 & 1 & 0.56\\\hline
  \multirow{3}{*}{3/5/20} & 1 & 4 & 1.24 & 0.46 & 1.11 & 1.24 & 1.14 & 933 & 232.74\\
  &2 & 4 & 0.99 & 0.55 & 0.92 & 0.99 & 0.92 & 1 & 0.87\\
  &3 & 4 & 0.72 & 0.08 & 0.62 & 0.72 & 0.62 & 123 & 20.16\\\hline
  \multirow{3}{*}{2/5/100} & 1 & 4 & 7.54 & 7.46 & 7.53 & 7.54 & 7.53 & 1 & 84.20\\
  &2 & 4 & 6.45 & 6.34 & 6.44 & 6.45 & 6.44 & 1 & 119.62\\
  &3 & 4 & 7.20 & 7.09 & 7.18 & 7.20 & 7.18 & 1 & 97.17\\\hline
  \multirow{3}{*}{3/5/100} & 1 & 4 & 8.40 & 8.32 & 8.39 & 8.40 & 8.39 & 1 & 135.73\\
  &2 & 4 & 5.79 & 5.70 & 5.78 & 5.79 & 5.78 & 1 & 38.19\\
  & 3 & 4 & 7.58 & 7.47 & 7.57 & 7.58 & 7.57 & 1 & 73.51\\\hline
  \multirow{3}{*}{2/10/100} & 1 & 9 & 3.56 & 3.11 & 3.53 & 3.56 & 3.53 & 1 & 71.10\\
  &2 & 9 & 3.00 & 2.61 & 2.97 & 3.00 & 2.97& 1 & 85.47\\
  &3 & 9 & 2.58 & 2.14 & 2.55 & 2.55 &2.58 & 1 & 22.99\\\hline
  \multirow{3}{*}{3/10/100} & 1 & 9 & 3.23 & 2.80 & 3.20 & 3.23 & 3.20 & 1 & 50.63\\
  &2 & 9 & 2.94 & 2.56 & 2.91 & 2.94 & 2.91 & 1 & 51.99\\
  &3 & 9 & 2.57 & 2.12 & 2.54 & 2.57 & 2.54 & 1 & 72.45\\\hline
\end{tabular}}
\caption{Computational results for class $B_2(r/R/p)$. Parameters as set in Table \ref{tab:A1_small}. Here $r^*$ is chosen during computation as largest $k$ such that the computational gradient method of Section \ref{sec:CG-method1} is strictly positive (up to 0.1). Here $r^*=R-1$ ends up being an appropriate choice.}
\label{tab:B2}
\end{table}

\begin{table}
\centering
{\small \begin{tabular}{|c|c|c|r|r|r|r|r|r|r|}
\hline
\multicolumn{3}{|c|}{} & \multicolumn{3}{c|}{Root node} &  \multicolumn{2}{c|}{Terminal node} & \multicolumn{2}{c|}{}\\\hline
Problem size &\multirow{2}{*}{Instance} & $r^*$ & Upper & CE & Weyl & Upper & Lower & \multirow{2}{*}{Nodes} & \multirow{2}{*}{Time (s)}\\
($r/R/p$) & & chosen & bound & LB & LB & bound & bound & & \\\hline
 \multirow{3}{*}{2/5/20} & 1 & 3 & 0.30 & $-0.30 $ & 0.13 & 0.30 & 0.20 & 777 & 219.75\\
 &2 & 3 & 0.22 & $-0.11$ & 0.12 & 0.22 & 0.12 & 1 & 2.61\\
 &3 & 3 & 0.19 & $-0.17$ & 0.10 & 0.19 & 0.10 & 1 & 1.57 \\\hline
  \multirow{3}{*}{3/5/20} & 1 & 2 & 1.00 & 0.62 & 0.82 & 1.00 & 0.90 & 217 & 48.43\\
  &2 & 2 & 0.95 & 0.47 & 0.74 & 0.95 & 0.85 & 305 & 26.14\\
  &3 & 2 & 0.45 & 0.06 & 0.30 & 0.45 & 0.36 & 98 & 12.31\\\hline
  \multirow{3}{*}{2/5/100} & 1 &  3 & 0.55 & 0.24 & 0.40 & 0.55 & 0.45 & 7099 & 47651.57 \\
  &2 & 3 & 0.28 & $-0.07$ & 0.17 & 0.28 & 0.18 & 266 & 1431.88\\
  &3 & 3 & 0.19 & $-0.07$ & 0.11 & 0.19 & 0.11 & 1 & 574.61\\\hline
  \multirow{3}{*}{3/5/100} & 1 & 2 & 1.03 & 0.70 & 0.81 & 1.03 & 0.93 & 4906 & 34514.69\\
  &2 & 2 & 0.90 & 0.66 & 0.78 & 0.90 & 0.80 & 95 & 596.96\\
  &3 & 2 & 0.55 & 0.39 & 0.47 & 0.55 & 0.47 & 1 & 1787.14\\\hline
  \multirow{3}{*}{2/10/100} & 1 & 8 & 0.12 & $-0.58 $ & 0.05 & 0.12 & 0.05 & 1 & 77.38\\
  &2 & 8 & 0.11 & $-0.53$ & 0.05 & 0.11 & 0.05 & 1 & 41.24\\
  &3 & 8 & 0.10 & $-0.43$ & 0.05 & 0.10 & 0.05 & 1 & 110.17\\\hline
 \multirow{3}{*}{3/10/100} & 1 & 6 & 0.33 & $-0.73$ & 0.20 & 0.33 & 0.22 & 27770 & 400000.0*\\
 & 2 & 7 & 0.25 & $-0.43$ & 0.17 & 0.25 & 0.17 & 1 & 128.63\\
 & 3 & 7 & 0.31 & $-0.28$ & 0.24 & 0.31 & 0.24 & 1 & 213.72\\\hline
\end{tabular}}
\caption{Computational results for class $B_3(r/R/p)$. Parameters as set in Table \ref{tab:A1_small}. As in Tables \ref{tab:B1} and \ref{tab:B2}, here $r^*$ is chosen during computation as largest $r$ such that the computational gradient method of Section \ref{sec:CG-method1} produces a feasible solution with strictly positive objective (up to 0.1).}
\label{tab:B3}
\end{table}

\begin{table}
\centering
{\small \begin{tabular}{|c|r|r|r|r|r|r|r|}
\hline
\multicolumn{1}{|c|}{} & \multicolumn{3}{c|}{Root node} &  \multicolumn{2}{c|}{Terminal node} & \multicolumn{2}{c|}{}\\\hline
  $r^*$ & Upper & CE & Weyl & Upper & Lower & \multirow{2}{*}{Nodes} & \multirow{2}{*}{Time (s)}\\
 used & bound & LB & LB & bound & bound & & \\\hline
 1 & 4.06 & 3.78 & 2.53 & 4.06 & 3.96 & 44 & 0.64\\
 2 & 2.64 & 2.04 & 1.42 & 2.64 & 2.54 & 1885 & 142.06\\
 3 & 1.56 & 0.62 & 0.61 & 1.56 & 1.46 & 11056 & 2458.23\\
 4 & 0.88 & $-0.33$ & 0.28 & 0.88 & 0.78 & 66877 & 60612.61\\
 5 & 0.36 & $-0.88$ & 0.0 & 0.36 & 0.25 & 155759 & 400012.4*\\\hline
\end{tabular}}
\caption{Computational results for the \textsf{geomorphology} example ($p=10$). Parameters as set in Table \ref{tab:A1_small}. Here we display results for the choices of $r^*\in\{1,2,3,4,5\}$ (for $r^*> 5$, the upper bound is below the numerical tolerance, so we do not include those).}
\label{tab:geo}
\end{table}

\begin{table}
\centering
{\small \begin{tabular}{|c|r|r|r|r|r|r|r|}
\hline
\multicolumn{1}{|c|}{} & \multicolumn{3}{c|}{Root node} &  \multicolumn{2}{c|}{Terminal node} & \multicolumn{2}{c|}{}\\\hline
  $r^*$ & Upper & CE & Weyl & Upper & Lower & \multirow{2}{*}{Nodes} & \multirow{2}{*}{Time (s)}\\
 used & bound & LB & LB & bound & bound & & \\\hline
 1 & 9.88 & 9.64 & 5.89 & 9.88 & 9.78 & 158 & 30.13\\\hline
 2 & 7.98 & 7.54 & 4.22 & 7.98 & 7.88 & 31710 & 49837.17\\\hline
 3 & 6.53 & 5.85 & 3.01 & 6.53 & 6.35 & 81935 & 400003.8*\\\hline
\end{tabular}}
\caption{Computational results for the \textsf{Harman} example ($p=24$). Parameters as set in Table \ref{tab:A1_small}. Here we display results for the choices of $r^*\in\{1,2,3\}$.}
\label{tab:harman}
\end{table}

\begin{table}
\centering
\subfloat{\small \begin{tabular}{|c|r|r|r|r|r|r|r|}
\hline
  \multirow{2}{*}{$r^*$} & Optimal & Time \\
   & value & (s)\\\hline
 1 & 51.85 & 5.59 \\\hline
 2 & 46.30 & 6.86 \\\hline
 3 & 41.29 & 6.08\\\hline
 4 & 36.54 & 6.08 \\\hline
 5 & 32.36 & 1.21 \\\hline
 6 & 28.72 & 7.83 \\\hline
 7 & 25.39 & 1.13 \\\hline
 8 & 22.10 & 3.87\\\hline
 9 & 19.20 & 4.97 \\\hline
 10 & 16.63 & 4.17 \\\hline
 11 & 14.30 & 7.87 \\\hline
\end{tabular}}
\quad\quad\quad\subfloat{\small \begin{tabular}{|c|r|r|r|r|r|r|r|}
\hline
  \multirow{2}{*}{$r^*$} & Optimal & Time \\
   & value & (s)\\\hline
 12 & 12.52 & 3.37\\\hline
 13 & 10.81 & 2.97\\\hline
 14 & 9.25 & 1.07\\\hline
 15 & 7.78 & 6.21 \\\hline
 16 & 6.44 & 5.40 \\\hline
 17 & 5.15 & 0.91\\\hline
 18 & 3.98 & 8.50\\\hline
 19 & 2.84 & 5.33\\\hline
 20 & 1.87 & 2.15\\\hline
 21 & 1.07 & 2.13\\\hline
 22 & 0.48 & 1.65\\\hline
\end{tabular}}
\caption{Computational results for the example \textsf{JO} ($p=58$). Parameters as set in Table \ref{tab:A1_small}. Here we display results for the choices of $r^*\in\{1,2,\ldots,22\}$ (for $r^*> 22$, the upper bound is below the numerical tolerance, so we do not include those). For all choices of $r^*$, the optimal solution is found and certified to be optimal at the root node, and therefore we do not display any information other than the optimal values.}
\label{tab:JO}
\end{table}


\begin{table}
\centering
\subfloat{\small \begin{tabular}{|c|c|r|r|}
\hline
\multirow{2}{*}{Example} & \multirow{2}{*}{$r^*$}  & Upper & Lower\\
& & bound & bound  \\\hline
\multirow{5}{*}{$A_1(100/1000)$} & 10 & 460.35 &457.46\\
& 30 & 322.35 &313.62\\
& 50 & 198.70 & 197.13\\
& 70 & 103.72 & 102.78\\
& 90 & 28.65 & 28.34\\\hline
\multirow{5}{*}{$A_2(1000)$} & 10 & 184.22 & 183.18\\
& 20 & 61.24 &60.34\\
& 30 & 20.25 &19.50\\
& 40 & 6.63 & 6.02\\
& 50 & 2.17 & 1.71\\\hline
\multirow{4}{*}{$B_2(10/90/1000)$} & 20 &379.02&377.23\\
& 40 & 236.10 & 234.83\\
& 60 & 122.08& 121.33\\
& 80 & 34.49 & 34.24\\\hline
\multirow{4}{*}{$B_2(20/90/1000)$} & 20 & 378.37 & 376.64\\
& 40 & 235.16 & 233.94\\\
& 60 & 121.69&120.96\\
& 80 & 34.33&34.09\\\hline
\multirow{4}{*}{$B_3(10/50/1000)$} & 10 & 384.40 & 384.04\\
& 20 & 233.25 &232.97\\
& 30 & 105.68 & 105.49\\
& 40 & 0.69 &0.60\\\hline
\multirow{4}{*}{$B_3(20/150/1000)$} & 20 & 426.31 & 422.13\\
& 45 & 285.56 &282.33\\
& 70 & 174.00 &171.70\\
& 95 & 86.54 &85.15\\
& 120 & 20.30 &19.80\\\hline
\end{tabular}}
\quad\subfloat{\small \begin{tabular}{|c|c|r|r|}
\hline
\multirow{2}{*}{Example} & \multirow{2}{*}{$r^*$}  & Upper & Lower\\
& & bound & bound  \\\hline
\multirow{6}{*}{$A_1(360/4000)$} & 100 &1334.40 & 1327.60\\
& 150 &991.64 &986.15\\
& 200 &693.02 &688.87\\
& 250 & 434.64 &431.80\\
& 300 & 213.29 &211.75\\
& 350 & 30.75 &30.50\\\hline
\multirow{7}{*}{$A_2(4000)$} & 10 & 733.41 & 733.09\\
& 20 & 244.13 & 243.85\\
& 30 & 79.94 &79.68\\
& 40 & 26.17& 25.96\\
& 50 & 8.57 & 8.39\\
& 60 & 2.80 & 2.66\\
& 70 &0.90 & 0.79\\\hline
\multirow{6}{*}{$B_2(80/360/4000)$} & 100 & 1360.1 & 1354.04\\
& 150 & 1008.90 & 1004.04 \\
& 200 & 703.78 & 700.21\\
& 250 & 440.56& 438.16\\
& 300 & 215.87 & 214.62\\
& 350 & 31.02 & 30.91\\\hline
\multirow{4}{*}{$B_3(120/360/4000)$} & 100 & 1081.20 &1078.74\\
& 150 & 630.25 & 628.54\\
& 200 & 253.50 & 252.50\\
& 240 & 7.89 & 7.46\\\hline
\end{tabular}}
\caption{Computational results across several classes for larger scale instances with $p\in\{1000,4000\}$. Here for a given example, the random seed is set as $1$. Results are displayed across a variety of choices of rank $r^*$. Class $B_1$ is not shown because of the reduction noted in Table \ref{tab:B1} (which reduces this class to a much smaller-dimensional class, and hence it is not truly large scale). The ``Upper bound'' denotes the upper bound found by the conditional gradient method, while ``Lower bound'' denotes the Weyl bound at the root node (no convex envelope bounds via Algorithm \ref{alg:bb} are shown here because of the large nature of the SDO-based convex envelope lower bounds which are prohibitive to solve for problems of this size). These instances are completed in \texttt{MATLAB} because of its superior support (over \texttt{julia}) for solving large-scale quadratic optimization problems (via \texttt{quadprog}), which is the only computation necessary for computing Weyl bounds.}
\label{tab:largeScale_1}
\end{table}

\subsubsection{Root Node Gap}

Let us first consider the gap at the root node. In classes $A_1$, $B_1$, $B_2$, and $B_3$, we see that the Weyl bound at the root node often provides a better bound than the one given by using convex envelopes. Indeed, the bound provided by Weyl's method can in many instances certify optimality (up to numerical tolerance) at the \emph{root node}. For example, this is the case in many instances of classes $A_1$ and half of the instances in $B_3$. Given that Weyl's method is computationally inexpensive (only requiring the computation of $p$ convex quadratic optimization problems and an eigenvalue decomposition), this suggests that Weyl's inequality as used within the context of factor analysis is particularly fruitful.

In contrast, in class $A_2$ and the real examples, we see that the convex envelope bound tends to perform better. Because of the structure of Weyl's inequality, Weyl's method is well-suited for correlation matrices $\s$ with very quickly decaying eigenvalues. Examples in these two classes do not have such a spectrum, and indeed Weyl's method does not provide the best root node bound.\footnote{However, it is worth remarking that Weyl's method still provides lower bounds on the rank of solutions to the noiseless factor analysis problem. Hence, even in settings where Weyl's method is not necessarily well-suited for proving optimality for the noisy factor analysis problem, it can still be applied successfully to lower bound rank for noiseless factor analysis.} Because neither Weyl's method nor the convex envelope bound strictly dominate one another at the root node across all examples, our approach incorporating both can leverage the advantages of each.

Observe that the root node gap (either in terms of the absolute difference between the initial feasible solution found and the better of the convex envelope bound and the Weyl bound) tends to be smaller when $r^*$ is much smaller than $p$. This suggests that the approach we take is well-suited to certify optimality of particularly low-rank decompositions in noisy factor analysis settings. We see that this phenomenon is true across all classes.

Finally, we remark that if the true convex envelope of the objective over the set of semidefinite constraints was taken, then the convex envelope objective would always be non-negative. However, because we have taken the convex envelope of the objective over the polyhedral constraints only, this is not the case.

\subsubsection{Performance at Termination}

It is particularly encouraging that the initial feasible solutions provided via the conditional gradient methods remain the best feasible solution throughout the progress of the Algorithm \ref{alg:bb} for all but three problem instances.\footnote{Of course, this need not be universally true, as for such an optimization problem the feasible $\ph\in\F$ provided by a conditional gradient method are only guaranteed to be locally optimal.} This is an important observation to make because without a provable optimality scheme such as the one we consider here, it is difficult to quantify the performance of heuristic upper bound methods. As we demonstrate here, despite the only local guarantees of solutions obtained via a conditional gradient scheme, they tend to perform quite well in the setting of factor analysis. Indeed, even in the three examples where the best feasible solution is improved, the improved solution is found very early in the branching process. 

Across the different example classes, we see that in general the gap tends to decrease more when $r^*$ is small relative to $p$ and $p$ is smaller. To appropriately contextualize and appreciate the number of nodes solved for a problem with $p=100$ on the timescale of 100s, with state-of-the-art implementations of interior point methods, solving a single node in the branch and bound tree can take on the order of 40s (for the specifically structured problems of interest and on the same machine). In other words, if one were to na\"ively use interior point methods, it would only be possible to solve approximately three nodes during a 100s time limit. In contrast, by using a first-order method approach which facilitates warm starts, we are able to solve hundreds of nodes in the same amount of time.

We see that Algorithm \ref{alg:bb} performs particularly well for classes $A_1$ (Tables \ref{tab:A1_small} and \ref{tab:A1_large}) and $B_1$ (Table \ref{tab:B1}), for problems of reasonable size with relatively small underlying ranks. This is highly encouraging. Class $A_1$ forms a set of prototypical examples for which theoretical recovery guarantees perform well; in stark contrast, problems such as those in $B_1$ which have highly structured underlying factors tend to not satisfy the requirements of such recovery guarantees. Indeed, if $\s\sim B_1(R/p)$, then there is generally appears to be a rank $R/2$ matrix $\T\psd$ so that $\s - \T$ is positive semidefinite and diagonal. In such a problem, for $r^*$ on the order of $R/2-1$, we provide nearly optimal solutions within a reasonable time frame.

Further, we note that similar results to those obtained for classes $A_1$ and $B_1$ are obtained for the other classes and are detailed in Tables \ref{tab:A2} (class $A_2$), \ref{tab:B2} (class $B_2$), \ref{tab:B3} (class $B_3$), and \ref{tab:geo} and \ref{tab:harman}. In a class such as $A_2$, which is generated as a high rank matrix (with decaying spectrum) with added individual variances, theoretical recovery guarantees do not generally apply, so again it is encouraging to see that our approach still makes significant progress towards proving optimality. Further, as shown in Table \ref{tab:largeScale_1}, for a variety of problems with $p$ on the order of $1000$ or $4000$, solutions can be found in seconds and optimality can be certified within minutes via Weyl bounds, with no need for convex envelope bounds as computed via Algorithm \ref{alg:bb}, so long as the rank $r^*$ is sufficiently small (for classes $A_1$, $B_2$, and $B_3$ on the order of hundreds, and for $A_2$ on the order of tens). This strongly supports the value of such an eigenvalue-based approach. When computing lower bounds solely via Weyl's method, the only necessary computations are solving quadratic programs (to find $\uu$ as in Section \ref{ssec:params}) and an eigenvalue decomposition. As Table \ref{tab:largeScale_1} suggests, for sufficiently small rank, one can still quickly find certifiably optimal solutions even for very large-scale factor analysis problems.

A particularly interesting example that stands apart from the rest is \textsf{JO} from the set of real examples. In particular, the correlation matrix for this example is not full rank (as there are fewer observations than variables). As displayed in Table \ref{tab:JO}, solutions found via conditional gradient methods are certified to be optimal at the root node. Indeed, the initial $\mb u_0$ as in Section \ref{ssec:params} can be quickly verified to be $\mb u_0=\mb 0$ for this example, i.e., factor analysis is only possible with no individual variances. Hence, Algorithm \ref{alg:bb} terminates at the root node without needing any branching. 

Finally, we note that all synthetic examples we have considered have equal proportions of common and individual variances (although, of course, this is not exploited by our approach as this information is not \emph{a priori} possible to specify without additional contextual information). If one modifies the classes so that the proportion of the common variances is higher than the proportion of individual variances (in the generative example), then Algorithm \ref{alg:bb} is able to deliver better bounds on a smaller time scale. (Results are not included here.) This is not particularly surprising because the branch-and-bound approach we take  essentially hinges on how well the products $W_{ii}\Phi_{ii}$ can be approximated. When there is underlying factor structure with a lower proportion of individual variances, the scale of $W_{ii}\Phi_{ii}$ is smaller and hence these products are easier to approximate well.

\subsection{Additional Considerations}

We now turn our attention to assessing the benefits of various algorithmic modifications as presented in Section \ref{sec:alg}. We illustrate the impact of these by focusing on four representative examples from across the classes: $A_1(3/10)$, $B_1(6/50)$, $B_3(3/5/20)$, and $G:=\textsf{geomorphology}$.

\begin{table}
\centering
{\small \begin{tabular}{|c|c|r|r|r|r|r|r|}
\hline
\multirow{2}{*}{Example} & \multirow{2}{*}{$r^*$}  & \multicolumn{6}{c|}{Nodes considered for $\epsilon=$}  \\
  & & 0.0 & 0.1 & 0.2 & 0.3 & 0.4 & 0.5\\\hline
 $A_1(3/10)$ & 2 & 103 & 72 & 59 & 62 & 78 & 101\\\hline
 $B_1(6/50)$ & 2 & 39 & 33 & 29 & 40 & 50 & 53\\\hline
 $B_3(3/5/20)$ & 2 & 8245 & 762 & 278 & 189 & 217 & 203 \\\hline
 \textsf{geomorphology} & 1 & 937 & 162 & 59 & 46 & 44 & 52\\\hline
\end{tabular}}
\caption{Computational results for effect of branching strategy across different examples. In particular, we consider how the number of nodes needed to prove optimality (up to numerical tolerance) changes across different choices of $\epsilon$, where $\epsilon$ is as in Section \ref{ssec:branching}. Recall that $\epsilon=0$ corresponds to a na\"ive choice of branching, while $\epsilon>0$ corresponds to a modified branching. All other parameters as in Table \ref{tab:A1_small}, including the branching index. For synthetic examples, the instances are initialized with random seed 1.}
\label{tab:ANl}
\end{table}

\begin{table}
\centering
{\small \begin{tabular}{|c|c|c|c|}
\hline
\multirow{2}{*}{Example} & \multirow{2}{*}{$r^*$}  & \multicolumn{2}{c|}{Nodes considered for }  \\
& & Na\"ive strategy & Modified strategy\\\hline
 $A_1(3/10)$ & 2 & 99 & 78\\\hline
 $B_1(6/50)$ & 2 & 135 & 50 \\\hline
 $B_3(3/5/20)$ & 2 & 375 & 217 \\\hline
 \textsf{geomorphology} & 1 & 43 & 44\\\hline
\end{tabular}}
\caption{Computational results for effect of node selection strategy across different examples. In particular, we consider how the number of nodes needed to prove optimality (up to numerical tolerance) changes across the na\"ive and modified branching strategies as described in Section \ref{ssec:branching}. All other parameters as in Table \ref{tab:A1_small}. For synthetic examples, the instances are initialized with random seed 1.}
\label{tab:nodeSelect}
\end{table}

\begin{table}
\centering
{\small \begin{tabular}{|c|c|c|c|c|}
\hline
\multirow{2}{*}{Example} & \multirow{2}{*}{$r^*$}  & Upper & \multicolumn{1}{c|}{CE LB} & \multicolumn{1}{|c|}{CE LB}  \\
& & Bound & with tightening & without tightening\\\hline
 $A_1(3/10)$ & 2 & 0.88 & 0.70 & 0.39 \\\hline
 $B_1(6/50)$ & 2 & 0.11 & $-0.15 $ & $-0.17$\\\hline
 $B_3(3/5/20)$ & 2 & 1.00 & 0.62 & 0.51 \\\hline
 \textsf{geomorphology} & 1 & 4.06 & 3.78 & 3.78\\\hline
\end{tabular}}
\caption{Computational results for effect of root node bound tightening across different examples. In particular, we consider how the convex envelope lower bound (denoted ``CE LB'') compares with and without bound tightening at the root node (see Section \ref{ssec:weyl}). All other parameters as in Table \ref{tab:A1_small}. ``Upper bound'' displayed is as in Table \ref{tab:A1_small} and is included for scale (i.e., to compare the relative impact of tightening). For synthetic examples, the instances are initialized with random seed 1.}
\label{tab:rootTighten}
\end{table}

\begin{table}
\centering
{\small \begin{tabular}{|c|c|c|c|c|}
\hline
\multirow{2}{*}{Example} & \multirow{2}{*}{$r^*$}  & \multicolumn{3}{|c|}{Nodes considered for $\tol=$}  \\
  & & 0.10 & 0.05 & 0.025\\\hline
 $A_1(3/10)$ & 2 & 78 & 287 & 726\\\hline
 $B_1(6/50)$ & 2 & 1 & 50 &262\\\hline
 $B_3(3/5/20)$ & 2 & 217 &  1396& 5132\\\hline
 \textsf{geomorphology} & 1 & 44 & 217 & 978 \\\hline
\end{tabular}}
\caption{Computational results for numerical tolerance $\tol$ across different examples. In particular, we consider how the number of nodes needed to prove optimality changes as a function of the additive numerical tolerance $\tol$ chosen for algorithm termination. All other parameters as in Table \ref{tab:A1_small}. For synthetic examples, the instances are initialized with random seed 1.}
\label{tab:halveGap}
\end{table}

\subsubsection{Performance of Branching Strategy}

We begin by considering the impact of our branching strategy as developed in Section \ref{ssec:branching}. The results across the four examples are shown in Table \ref{tab:ANl}. Recall that $\epsilon\in[0,1)$ controls the extent of deviation from the canonical branching location, with $\epsilon=0$ corresponding to no deviation. Across all examples, we see that the number of nodes considered to prove optimality is approximately convex in $\epsilon\in[0,0.5]$. In particular, for all examples, the ``optimal'' choice of $\epsilon$ is not the canonical choice of $\epsilon=0$. This contrast is stark for the examples $B_3(3/5/20)$ and $G$. Indeed, for these two examples, the number of nodes considered when $\epsilon=0$ is over five times larger than for any $\epsilon\in\{0.1,0.2,0.3,0.4,0.5\}$.

In other words, the alternative branching strategy can have a substantial impact on the number of nodes considered in the branch-and-bound tree. As a direct consequence, this strategy can drastically reduce the computation time needed to prove optimality. As the examples suggest, it is likely that $\epsilon$ should be chosen dynamically during algorithm execution, as the particular choice depends on a given problem's structure. However, we set $\epsilon=0.4$ for all other computational experiments because this appears to offer a distinct benefit over the na\"ive strategy of setting $\epsilon=0$.

\subsubsection{Performance of Node Selection Strategy}

We now turn our attention to the node selection strategy as detailed in Section \ref{ssec:nodeSelect}. Recall that node selection considers how to pick which node to consider next in the current branch-and-bound tree at any iteration of Algorithm \ref{alg:bb}. We compare two strategies: the na\"ive strategy which selects the node with worst convex envelope bound (as explicitly written in Algorithm \ref{alg:bb}) and the modified strategy which employs randomness and Weyl bounds to consider nodes which might not be advantageous to fathom when only convex envelope bounds are considered.

The comparison is shown in Table \ref{tab:nodeSelect}. We see that this strategy is advantageous overall (with only the example $G$ giving a negligible decrease in performance). The benefit is particularly strong for examples from the $B$ classes which have highly structured underlying factors. For such examples, there is a large difference between the convex envelope bounds and the Weyl bounds at the root node (see e.g. Table \ref{tab:B1}). Hence, an alternative branching strategy which incorporates the eigenvalue information provided by Weyl bounds has potential to improve beyond the na\"ive strategy. Indeed, this appears to be the case across all examples where such behavior occurs.

\subsubsection{Performance of Bound Tightening}

All computational results employ bound tightening as developed in Section \ref{ssec:weyl}, but only at the root node. Bound tightening, which requires repeated eigenvalue computations, is a computationally expensive process. For this reason, we have chosen not to employ bound tightening at every node in the branch-and-bound tree. From a variety of computational experiments, we observed that the most important node for bound tightening is the root node, and therefore it is a reasonable choice to only employ bound tightening there. Consequently, we employ pruning via Weyl's method (detailed in Section \ref{ssec:weyl} as well) at all nodes in the branch-and-bound.\footnote{Recall that bound tightening can be thought of as optimal pruning via Weyl's method.}

In Table \ref{tab:rootTighten} we show the impact of bound tightening at the root node in terms of the improvement in the lower bounds provided by convex envelopes. The results for the class $A_1$ are particularly distinctive. Indeed, for this class bound tightening has a substantial impact on the quality of the convex envelope bound (for the example $A_1(3/10)$ given, the improvement is from a relative gap at the root node of $56\%$ to a gap of $20\%$). For the examples shown, bound tightening offers the least improvement in the real example \textsf{geomorphology}. In light of Table \ref{tab:geo} this is not too surprising, as Weyl's method (at the root node) is not particular effective for this example. As Weyl's inequality is central to bound tightening, problems for which Weyl's inequality is not particularly effective tend to experience less benefit from bound tightening at the root node.

\subsubsection{Influence of Optimality Tolerance}

We close this section by assessing the influence of the optimality tolerance for termination, \tol. In particular, we study how the number of nodes to prove optimality changes as a function of additive gap necessary for termination. The corresponding results across the four examples are shown in Table \ref{tab:halveGap}. Not surprisingly, as the gap necessary for termination is progressively halved, the corresponding number of nodes considered increases substantially. However, it is important to note that even though the gap at termination is smaller as this tolerance decreases (by design), for these examples the best feasible solution remains unchanged. In other words, the increase in the number of nodes is the price for more accurately proving optimality and not for finding better feasible solutions. Indeed, as noted earlier, the solutions found via conditional gradient methods at the outset are of remarkably high quality.

\section{Conclusions}\label{sec:conclude}

We analyze the classical rank-constrained FA problem from a computational perspective. 
We proposed a general flexible family of rank-constrained, nonlinear SDO based formulations for the
task of approximating an observed covariance matrix $\B\Sigma$, 
as the sum of a PSD low-rank component $\B\Theta$ and a diagonal matrix $\B\Phi$ (with nonnegative entries) subject to 
$\B\Sigma - \B\Phi$ being PSD.   
Our framework enables us to estimate the underlying factors and \emph{unique variances} under 
the restriction that the residual covariance matrix is semidefinite --- this is important for statistical interpretability and
understanding the proportion of variance explained by a given number of factors. 
This constraint, however, seems to ignored by most other widely used methods in FA. 

 We introduce a novel \emph{exact} reformulation of the rank-constrained FA problem
as a smooth optimization problem with convex compact constraints.
  We present a unified algorithmic framework, utilizing modern techniques in 
nonlinear optimization and first order methods in convex optimization to obtain high quality solutions for the FA problem. At the same time, we use techniques from discrete and global optimization to demonstrate that these 
solutions are often provably optimal. 
We provide computational evidence demonstrating that the methods proposed herein, provide high quality estimates with improved accuracy, 
when compared to existing, popularly used methods in FA, in terms of a wide variety of metrics.

In this work we have demonstrated that a previously intractable rank optimization problem can be solved to provable optimality. We envision that ideas similar to those used here can be used to address an even larger class of matrix estimation problems. In this way, we anticipate significant progress on such problems in the next decade, particularly in light of myriad advances throughout distinct areas of modern optimization.

\section*{Appendix A}

This appendix contains all proofs for results presented in the main text.

\begin{proof}{\textbf{of Proposition \ref{lem:reform-1}:}}
\begin{enumerate}[(a)]
\item We start by observing that for any two real symmetric matrices $\M{A}, \M{B}$ (with dimensions $p \times p$) and the matrix $q$-norm, 
a result due to Mirsky\footnote{This is also known as the $q$-Wielandt-Hoffman inequality.}
(see for example~\cite{stewart-sun:1990}, pp.\ 205)
states:
\begin{equation}\label{ineq-1}
 \| \M{A} - \M{B}  \|_{q} \geq \| \B{\lambda}(\M{A}) - \B{\lambda}(\M{B}) \|_{q},
\end{equation}
where $\B{\lambda}(\M{A})$ and $\B{\lambda}(\M{B})$ denote the vector of eigenvalues of $\M{A}$ and  $\M{B}$, respectively, arranged in decreasing order, i.e., 
 $\lambda_{1}(\M A) \geq \lambda_2(\M A) \geq \ldots \geq \lambda_{p}(\M A)$ and $\lambda_{1}(\M B) \geq \lambda_2(\M B) \geq \ldots \geq \lambda_{p}(\M B).$
Using this result for Problem~(\ref{obj-2-0}), it is easy to see that for fixed $\B\Phi$:
\begin{equation}\label{equiv-mat-sing}
\left\{ \B\Theta : \B{\Theta} \succeq \M{0}, \rnk(\B{\Theta}) \leq r  \right\}   = \left \{ \B\Theta :  \B\lambda(\B\Theta) \B{\geq} \M{0},  \| \B\lambda(\B\Theta) \|_0 \leq r \right\},
\end{equation}
where, the notation (in bold)  ``$\B{\geq}$'' denotes component-wise inequality and  $\| \B{\lambda}(\B\Theta) \|_0$ counts the number of non-zero elements of $\B{\lambda}(\B\Theta)$.
If we partially minimize the objective function in Problem~(\ref{obj-2-0}), with respect to $\B\Theta$ (with $\B\Phi$ fixed), and use~\eqref{ineq-1} along with~\eqref{equiv-mat-sing}, 
we have the following inequality:
\begin{equation}\label{form-min0}
\begin{array}  { lll ll }
                    & \inf\limits_{\B\Theta}                    & \left\| \left(\B{\Sigma} - \B\Phi \right) -  \B{\Theta}  \right\|^q_q \;                       & \sbt &  \B{\Theta} \succeq \M{0}, \rnk(\B{\Theta}) \leq r  \\
 \geq  \;\; \; &  \inf\limits_{\B{\lambda}(\B\Theta)}  &   \left\| \B\lambda(\B{\Sigma} - \B{\Phi} ) - \B\lambda(\B\Theta) \right\|_q^q  &  \sbt &  \B\lambda(\B\Theta) \B{\geq} \M{0},  \| \B\lambda(\B\Theta) \|_0 \leq r.
\end{array}
\end{equation}
Since $\B{\Sigma} - \B{\Phi} \succeq \M 0,$ it follows that the minimum objective value of the r.h.s.~optimization Problem~\eqref{form-min0}  
is given by $\sum_{i=r+1}^{p} \lambda^q_i(\B{\Sigma} - \B{\Phi} )$ and is achieved for 
$\lambda_{i}(\B\Theta) = \lambda_{i}(\B{\Sigma} - \B{\Phi} )$ for $i = 1, \ldots, r$. This leads to the following inequality:
\begin{equation}\label{form-min1}
\begin{aligned}
\inf_{\B{\Theta} : \B{\Theta} \succeq \M{0}, \rnk(\B{\Theta}) \leq r } \;\;\;&
\left\| \left(\B{\Sigma} - \B\Phi \right) -  \B{\Theta}  \right\|^q_q \;
\geq \;  \sum\limits_{i=r+1}^{p} \lambda^q_i(\B{\Sigma} - \B{\Phi} ).
\end{aligned}
\end{equation}
Furthermore, if $\M{U}_{p \times p}$ denotes the matrix of $p$ eigenvectors  
of $\B\Sigma - \B\Phi$, then the following choice of $\B\Theta^*$:
\begin{equation}\label{form-2-0}
\B\Theta^* :=  \M{U} \; \diag \big ( \lambda_{1}( \B{\Sigma} - \B{\Phi}  ), \ldots, \lambda_{r}(\B{\Sigma} - \B{\Phi} ) , 0, \ldots, 0 \big) \; \M{U}',
\end{equation}
gives equality in~\eqref{form-min1}.
This leads to the following result:
\begin{equation}\label{form-1}
\begin{aligned}
\inf_{\B{\Theta} : \B{\Theta} \succeq \M{0}, \rnk(\B{\Theta}) \leq r } \;\;\;&
\| \B{\Sigma} - (\B{\Theta} + \B{\Phi})  \|^q_q \;
= \; \| \B{\Sigma} - (\B{\Theta}^* + \B{\Phi})  \|^q_q \; = \; \sum_{i=r+1}^{p} \lambda^q_i(\B{\Sigma} - \B{\Phi} ). 
\end{aligned}
\end{equation}
\item The minimizer $\B\Theta^*$ of Problem~\eqref{form-1} is given by~\eqref{form-2-0}.  In particular, if $\B\Phi^*$ solves 
Problem~\eqref{obj-2-0-margin} and we compute $\B\Theta^*$ via~\eqref{form-2-0} (with $\B\Phi = \B\Phi^*$), then the tuple 
$(\B\Phi^*, \B\Theta^*)$ solves Problem~\eqref{obj-2-0}.
 This completes the proof of the proposition.
 \end{enumerate}
 \end{proof}

\begin{proof}{\textbf{of Proposition \ref{equi-minus-0}:}}
 We build upon the proof of Proposition~\ref{lem:reform-1}. Note that any $\B\Phi$ that is feasible for Problems~\eqref{obj-2-0-0} and~\eqref{obj-2-0} is PSD.
Observe that $\B\Theta^*$ (appearing in the proof of Proposition~\ref{lem:reform-1}) as given by~\eqref{form-2-0}  satisfies:
\begin{equation}\label{diff-psd-1}
\B\Sigma - \B\Phi - \B\Theta^{*} \succeq \M{0} \implies \B\Sigma - \B\Theta^* \succeq \M{0},
\end{equation}
where, the right hand side of~\eqref{diff-psd-1}  follows since $\B\Phi \succeq \M{0}$. We have thus established that the solution $ \B\Theta^*$ to the following problem
\begin{equation}\label{form-partial-1-1}
\begin{aligned}
\mini_{\B{\Theta}} \;\;\;&
\left\| \left(\B{\Sigma} - \B\Phi \right) -  \B{\Theta}  \right\|^q_q \\
\sbt\;\;\; &  \B{\Theta} \succeq \M{0} \\
& \rnk(\B{\Theta}) \leq r ,
\end{aligned}
\end{equation}
is feasible for the following optimization problem:
\begin{equation}\label{form-partial-1-1-2}
\begin{aligned}
\mini_{\B{\Theta}} \;\;\;&
\left\| \left(\B{\Sigma} - \B\Phi \right) -  \B{\Theta}  \right\|^q_q \\
\sbt\;\;\; &  \B{\Theta} \succeq \M{0} \\
& \rnk(\B{\Theta}) \leq r \\
& \B\Sigma  - \B\Theta \succeq \M{0}.
\end{aligned}
\end{equation}
Since Problem~\eqref{form-partial-1-1-2} involves minimization over a subset of the feasible set of Problem~\eqref{form-partial-1-1}, it follows that 
$\B\Theta^{*}$ is also a minimizer for Problem~\eqref{form-partial-1-1-2}. This completes the proof of the equivalence. 
\end{proof}

\begin{proof}{\textbf{of Theorem \ref{thm:FA-reform1-gen-q}:}}
\begin{enumerate}[(a)]
\item The proof 
is based on ideas appearing in~\citet{overton-92}, where it was shown that the sum of the top $r$ eigenvalues of a real symmetric matrix can be written as the solution to a linear SDO problem.

By an elegant classical result due to Fan~\citep{stewart-sun:1990}, the smallest $(p-r)$ eigenvalues of a real symmetric matrix $\M{A}$ can be written as:
\begin{equation}\label{sum-low-r-eig-orth}
\begin{array}{lc l}
\sum\limits_{i = r + 1}^{ p } \lambda_i (\M{A})= &  \inf\limits_{ }& \tr \left ( \M{V}'\M{A} \M{V} \right) \\
                                                             &\sbt \;\;& \M{V}'\M{V} = \M{I},
                                                             \end{array}
                                                                                                                                     \end{equation}        
where, the optimization variable                                                                $\M{V} \in \R^{p \times (p-r)} $.                                                                                                                          
                                                                         We will show that the solution to the above non-convex problem can be 
                                                          obtained via the following linear (convex) SDO problem:                                                                                                                                   
\begin{equation}\label{sum-low-r-eig}
\begin{array}{rl}
 \mini\limits_{}  \;\;\; & \tr  \left(\M{W} \M{A} \right)\\
                                                           \sbt \;\;\;&\M{I} \succeq \M{W} \succeq \M{0} \\
                                                           &  \tr (\M W) = p - r.
                                                                                                                                     \end{array}
                                                                                                                                     \end{equation}
Clearly, Problem~\eqref{sum-low-r-eig} is a convex relaxation of Problem~\eqref{sum-low-r-eig-orth}---hence its minimum value is at least smaller than
$\sum_{i = r + 1}^{ p } \lambda_i (\M{A})$. 
By standard results in convex analysis~\citep{rock-conv-96}, it follows that 
the minimum of the above linear SDO problem~\eqref{sum-low-r-eig} is attained at the extreme points of the feasible set of~\eqref{sum-low-r-eig}.  
The extreme points~\citep[see for example,][]{overton-92,sdpSurvey} of this set are given by the set of orthonormal matrices of rank $p-r$: 
$$\big\{ \M{V}\M{V}' :  \M{V} \in \R^{p \times (p - r)} : \M{V}'\M{V}=\M{I} \big \}.$$
It thus follows that the (global) minima of Problems~\eqref{sum-low-r-eig} and~\eqref{sum-low-r-eig-orth} are the same. 
Applying this result to the PSD matrix 
$\M{A} = (\B\Sigma - \B\Phi)^{q}$ appearing in the objective of~\eqref{obj-2-0-margin}, we arrive at~\eqref{FA-reform1-gen-q}. This completes the proof of part {\bf (a)}.

\item  The statement follows from~\eqref{form-2-0}.
\end{enumerate}
\end{proof}

\begin{proof}{\textbf{of Proposition \ref{prop:margin-1}:}}
Observe that, for every fixed $\B{\Phi}$, the function $g_{q}(\M{W}, \B\Phi)$ is concave (in fact, it is linear). Since $G_{q}(\M{W})$ is obtained by 
taking a point-wise 
infimum  (with respect to $\B\Phi$) of the concave function $g_{q}(\M{W},\B{\Phi})$, the
 resulting function $G_{q}(\M{W})$ is concave~\citep{BV2004}. 

The expression of the sub-gradient~\eqref{eq:subdiff1} is an immediate consequence of Danskin's Theorem~\citep{bertsekas-99,rock-conv-96}.
\end{proof}

\begin{proof}{\textbf{of Theorem \ref{lem:conv-rate-1}:}}
Using the concavity of the function $G_{q}({\M{W}})$ we have:
\begin{equation} \label{proof:eq-rate1}
\begin{array}{l c l c l }
G_{q}({\M{W}}^{(i+1)} ) & \leq & G_{q}({\M{W}}^{(i)} ) &+& \langle \nabla G_{q}({\M{W}}^{(i)} ) , {\M{W}}^{(i+1)} - {\M{W}}^{(i)}  \rangle \\
&=& G_{q}({\M{W}}^{(i)} ) &+& \Delta(\M{W}^{(i)}) .
\end{array}
\end{equation}
Note that
$\Delta(\M{W}^{(i)}) \leq 0$ and  $G_{q}({\M{W}}^{(i+1)} )  \leq G_{q}({\M{W}}^{(i)} )$ for all $i\leq k$. 
Hence the (decreasing) sequence of objective values converge to $G_{q}({\M{W}}^{(\infty)})$ (say) and
$\Delta({\M{W}}^{(\infty)}) \geq 0$.
Adding up the terms in~\eqref{proof:eq-rate1}
from $i = 1, \ldots, k$ we have:
\begin{equation}
\begin{aligned}
G_{q}({\M{W}}^{(\infty)} ) - G_{q}({\M{W}}^{(1)} )  \leq  G_{q}({\M{W}}^{(k+1)} ) - G_{q}({\M{W}}^{(1)} )  \leq 
-k \min\limits_{i=1, \ldots, k} \left \{-\Delta(\M{W}^{(k)}) \right\}
\end{aligned}
\end{equation}
from which~\eqref{complexity-rule1-adapt} follows. 
\end{proof}

\bibliographystyle{plainnat_my_nourl}

\bibliography{References_comb}

\end{document}